\documentclass[a4paper]{article}
\usepackage{amsmath,amsthm,amssymb}
\usepackage{enumitem}
\usepackage[dvipsnames]{xcolor}
\usepackage{booktabs,cite}
\usepackage{graphicx,cancel}
\usepackage[toc,page]{appendix}
\usepackage[hidelinks]{hyperref}
\usepackage{comment}
\usepackage{mathtools,braket}
\usepackage[top=2.8cm,bottom=2.8cm,left=2.8cm,right=2.8cm]{geometry}
\usepackage{float}
\usepackage{pict2e}
\usepackage{adjustbox}
\allowdisplaybreaks

\usepackage{tikz} 
\usetikzlibrary{matrix,decorations.pathreplacing,calc,fit,backgrounds,shapes.multipart,arrows.meta,positioning,calc}
 
\tikzset{nodino/.style={minimum height=1.2cm, minimum width=3cm}}

\newtheorem{theorem}{Theorem}[section]
\newtheorem*{theorem*}{Theorem}
\newtheorem{proposition}[theorem]{Proposition}
\newtheorem{lemma}[theorem]{Lemma}

\theoremstyle{definition}

\newtheorem{definition}{Definition}[section]
\newtheorem{remark}{Remark}[section]

\definecolor{frenchrose}{rgb}{0.96, 0.29, 0.54}
\definecolor{persianblue}{rgb}{0.11, 0.22, 0.73}
\definecolor{jade}{rgb}{0.0, 0.66, 0.42}
\definecolor{limegreen}{rgb}{0.2, 0.8, 0.2}
\definecolor{electricviolet}{rgb}{0.56, 0.0, 1.0}
\definecolor{electricblue}{RGB}{44, 117, 255}
\definecolor{bluebonnet}{RGB}{51, 38, 246}
\definecolor{strawberrypink}{RGB}{246, 38, 129}

\numberwithin{equation}{section}

\expandafter\def\expandafter\normalsize\expandafter{%
    \normalsize%
    \setlength\abovedisplayskip{7pt}%
    \setlength\belowdisplayskip{7pt}%
    \setlength\abovedisplayshortskip{5pt}%
    \setlength\belowdisplayshortskip{5pt}%
}

\DeclareFontFamily{U}{mathx}{\hyphenchar\font45}
\DeclareFontShape{U}{mathx}{m}{n}{<-> mathx10}{}
\DeclareSymbolFont{mathx}{U}{mathx}{m}{n}
\DeclareMathAccent{\widebar}{0}{mathx}{"73}

\title{\vspace{1.8cm} Matrix Dressing Beyond Pfaff--Toda}

\author{
Sylvain Carpentier%
\thanks{E-mail: \href{mailto:sylcar@snu.ac.kr}{\texttt{sylcar@snu.ac.kr}}}
\and
Marta Dell'Atti%
\thanks{E-mail: \href{mailto:dellattimarta@gmail.com}{\texttt{dellattimarta@gmail.com}}}
}

\date{}

\begin{document}

\maketitle

\begin{center}
\itshape To the memory of Igor Krichever
\end{center}

\begin{abstract}
We develop a matrix pseudodifference operator framework that places the
Adler--Pfaff and Pfaff--Toda hierarchies in a common algebraic setting.
The bi-infinite Adler--Pfaff hierarchy then becomes a
$2\times2$ matrix pseudodifference Lax hierarchy.
We introduce the Matrix Pfaff--Toda hierarchy by dressing the matrix Laurent algebra $M_{2N}\bigl(\mathbb{C}[\mathcal{S},\mathcal{S}^{-1}]\bigr)$ using a Pfaff--Toda splitting of the matrix pseudodifference algebra.
In the one-component case, its diagonal sector recovers the continuous Pfaff--Toda hierarchy, while the off-diagonal directions supply the missing odd flows, and the powers of a single bare operator reconstruct the full Adler--Pfaff hierarchy.
For general~$N$, on the regular factorization locus, the multicomponent Pfaff--Toda tau-functions of Savchenko and Zabrodin realize the commuting diagonal sector, with its dressing equations derived directly from the fermionic bilinear identity. Reductions to the block-diagonal and scalar cases recover multicomponent and scalar $2D$ Toda, identifying the even Pfaff hierarchy with an anti-diagonal Toda subhierarchy, and yielding the Krichever--Zabrodin C--Toda hierarchy.
\end{abstract}

\vspace{1mm}

\tableofcontents

\vspace{.5cm}

\section{Introduction}

The Pfaff lattice was introduced in the work of Adler, Horozov and van Moerbeke~\cite{vanM_Pfaff_skew}, and subsequently developed in several directions \cite{Adler1999ThePL,adler2002,vanM_Pfaff_tau_function,Adler2000_skew}. It is an integrable hierarchy admitting a bi-infinite matrix realization, closely related to skew-orthogonal polynomials and to orthogonal and symplectic random matrix ensembles (see also the survey~\cite{vanMoerbekenotes}, the finite-dimensional geometric formulation~\cite{KodamaPierce,Kodama_2010}, and the semi-infinite restriction in~\cite{benassi2021symmetric,BDM2}). 
The Pfaff lattice is intimately related to the charged BKP and DKP hierarchies. The DKP hierarchy was discovered by Jimbo and Miwa~\cite{JimboMiwa} and rediscovered by Hirota and Ohta as the coupled KP hierarchy~\cite{HirotaOhta}. The DKP hierarchy was later studied by Kac and van de Leur, together with its BKP counterpart, from the viewpoint of the infinite-dimensional Clifford algebra and its spin representations~\cite{KacLeur,KacLeur2}. The closely related so-called large BKP hierarchy is the object of study in~\cite{vdLeurOrlov}. While the small BKP hierarchy is the square root of the KP hierarchy, the large BKP is the square root of the two-component KP hierarchy~\cite{vdLeurOrlov}.

The Lax representation of the Pfaff lattice is
governed by a symplectic splitting of the bi-infinite matrix algebra and
differs in an essential way from the ordinary triangular splitting
underlying the Toda lattice hierarchy~\cite{UT84}. Although the Pfaff and Toda
hierarchies are closely related~\cite{adler2002}, their operator
descriptions are therefore rather different: Toda is conventionally formulated
in terms of difference operators and dressing transformations, whereas the
Pfaff hierarchy is formulated through its bi-infinite matrix realization. 

The Pfaff--Toda hierarchy provides a natural link between the Pfaff and Toda pictures. In Takasaki's formulation \cite{Taka}, it is described by two
opposite dressing operators related by a Pfaff-type constraint. More recently,
Savchenko and Zabrodin have developed a multicomponent Pfaff-type theory in
fermionic, bilinear and dispersionless form \cite{SavchenkoZabrodin,SavZab25DKP, SavZab26DispPfaffToda,SavZab26Elliptic}, including the multicomponent Pfaff--Toda hierarchy of \cite{SavchenkoZabrodin} in the representation-theoretic formalism of Kac and van de Leur. These works suggest that multicomponent Pfaff--Toda should also admit a matrix pseudodifference operator dressing representation.

 The purpose of this paper is to place these hierarchies in a common matrix pseudodifference operator framework. 
Our main results are summarized as follows:
\newpage
\medskip
\hypertarget{res1}{\textbf{1. A matrix pseudodifference Lax realization of the
Adler--Pfaff hierarchy.}}

Our starting point is the bi-infinite Adler--Pfaff hierarchy, whose $2\times2$ block Hessenberg Lax operator is governed by the
Adler--van Moerbeke symplectic splitting~\cite{adler2002,vanM_Pfaff_skew}.
The block index induces a shift $\mathcal S$ on the coefficient
algebra $\mathcal V_{\rm Pf}$ which commutes with the Adler--Pfaff
flows. The corresponding algebra $\mathcal C_{\rm Pf}$ of
translation-covariant bi-infinite matrices is naturally identified
with the algebra of matrix pseudodifference operators $M_2\bigl(\mathcal V_{\rm Pf}((\mathcal S))\bigr)$
through the isomorphism
\begin{equation*}
\Phi(M)=\sum_{k\in\mathbb Z}M_{0,-k}\mathcal S^k.
\end{equation*}
On finite-band matrices, transposition is carried by $\Phi$ to the
formal adjoint. Under this identification, the two summands of the
Adler--van Moerbeke splitting are transported to
\begin{equation*}
\mathfrak g^{\rm Pf}_+=\left\{
X\in M_2\bigl(
\mathcal V_{\rm Pf}[\mathcal S,\mathcal S^{-1}]
\bigr)
\ \middle|\
X^*=J_0XJ_0
\right\}, \qquad 
\mathfrak g^{\rm Pf}_-
=
\mathcal V_{\rm Pf}I_2
+
\mathcal S M_2\bigl(\mathcal V_{\rm Pf}[[\mathcal S]]\bigr),
\end{equation*}
with corresponding projections $P_\pm^{\rm Pf}$. 

\begin{theorem*}[Pseudodifference Adler--Pfaff hierarchy]
The bi-infinite Adler--Pfaff hierarchy is equivalent, under the
identification above, to the $2\times2$ matrix pseudodifference
Lax hierarchy
\begin{equation*}
\partial_{t_k}\mathcal L_{\rm Pf}=\bigl[
P_+^{\rm Pf}(\mathcal L_{\rm Pf}^k),
\mathcal L_{\rm Pf}
\bigr]=-\bigl[
P_-^{\rm Pf}(\mathcal L_{\rm Pf}^k),
\mathcal L_{\rm Pf}
\bigr],
\qquad k\geq1,
\end{equation*}
where
\begin{equation*}
\mathcal L_{\rm Pf}=b_0E_{21}\mathcal S^{-1}+\sum_{j\geq0}A_{0,j}\mathcal S^j
~\in
\mathcal V_{\rm Pf}E_{21}\mathcal S^{-1}+M_2\bigl(\mathcal V_{\rm Pf}[[\mathcal S]]\bigr).
\end{equation*}
\end{theorem*}
Here and below, $E_{ij}$ denotes the standard unit matrix with $1$ at position $(i,j)$. 
We extend this hierarchy to a slightly larger universal
phase space, whose support is precisely the one needed for the comparison with the symmetrically gauged Pfaff--Toda operator later in the text. 

\medskip 
\hypertarget{res2}{\textbf{2. The Matrix Pfaff--Toda hierarchy.}}

Already for $N=1$, Takasaki's Pfaff--Toda dressing datum is
genuinely $2\times2$ matrix-valued. However, the associated bare
operators are restricted to two commuting diagonal families \cite{Taka}. We define a dressing representation of the full Laurent algebra 
\begin{equation*}
\mathfrak a_N= M_{2N}\bigl(\mathbb C[\mathcal S,\mathcal S^{-1}]\bigr).
\end{equation*}
We call the resulting system the \textit{Matrix Pfaff--Toda hierarchy}.
The dressing space is described by operators of the form
\begin{equation*}
U_N=
\begin{pmatrix}
I_N+a_0+\mathcal O(\mathcal S)
&
b_{-1}\mathcal S^{-1}+\mathcal O(1)
\\[1mm]
c_1\mathcal S+\mathcal O(\mathcal S^2)
&
\Delta+\mathcal O(\mathcal S)
\end{pmatrix},
\end{equation*}
where $a_0,b_{-1},c_1$ are strictly lower triangular and $\Delta$ is
lower triangular and invertible. Here and below, $\mathcal O(\mathcal S^k)$ is understood with respect to the
one-sided completion in use: it denotes terms of degree at least $k$
in an $\mathcal S$-adic expansion and of degree at most $k$ in an
$\mathcal S^{-1}$-adic expansion. We denote by $\mathcal V_N$ the localized
difference algebra generated by the coefficients of $U_N$. For general $N$, we extend Takasaki's splitting to the Lie algebra 
\begin{equation*}
M_{2N}\bigl(\mathcal V_N((\mathcal S))\bigr)
=
\mathfrak g_+\oplus\mathfrak g_-,
\end{equation*}
whose positive and negative parts are respectively
\begin{equation*}
\mathfrak g_+=\left\{
X\in M_{2N}\bigl(\mathcal V_N[\mathcal S,\mathcal S^{-1}]\bigr)
\ \middle|\
X^*=J_NXJ_N
\right\},  \qquad \mathfrak{g}_{-}=
\left\{
\begin{pmatrix}
a+\mathcal O(\mathcal S)&b\mathcal S^{-1}+\mathcal O(1)\\
c\mathcal S+\mathcal O(\mathcal S^2)&d+\mathcal O(\mathcal S)
\end{pmatrix}\right\}
\end{equation*}
with $a,b,c$ strictly lower triangular, $d$ lower triangular, and 
\begin{equation*}
          J_N=
\begin{pmatrix}
0 & \mathcal S^{-1}I_N\\
-\mathcal S I_N & 0
\end{pmatrix}.
\end{equation*}
\begin{theorem*}[The Matrix Pfaff--Toda hierarchy]
Let $U_N$ be the dressing operator above. For
$x\in\mathfrak a_N$, define
\begin{equation*}
L_x:=U_N\,x\,U_N^{-1},
\qquad
P_x^\pm:=\Pi_\pm(L_x),
\end{equation*}
where $\Pi_\pm$ are the projections associated with the preceding splitting,
and let $D_x$ be the evolutionary derivation of $\mathcal V_N$ determined by
\begin{equation*}
D_x(U_N)
=
-P_x^-U_N
=
P_x^+U_N-U_Nx.
\end{equation*}
Then, for every $x,y\in\mathfrak a_N$,
\begin{equation*}
D_x(L_y)
=
[P_x^+,L_y]-L_{[x,y]},
\qquad
[D_x,D_y]
=
-D_{[x,y]},
\end{equation*}
and the auxiliary operators satisfy the zero-curvature identity
\begin{equation*}
D_x(P_y^+)-D_y(P_x^+)
=
[P_x^+,P_y^+]-P_{[x,y]}^+.
\end{equation*}
\end{theorem*} 
Although the full family of
matrix dressing directions is non-abelian, every commuting subalgebra
of bare operators gives a commuting subhierarchy. For $N=1$ the diagonal
sector recovers Takasaki's continuous Pfaff--Toda hierarchy
\cite{Taka}, while for arbitrary $N$ this is the sector realized by the Savchenko--Zabrodin tau-functions on the regular factorization locus \cite{SavchenkoZabrodin}.

\medskip 
\hypertarget{res3}{\textbf{3. The Adler--Pfaff hierarchy in the odd Pfaff--Toda extension.}}

For $N=1$, setting $U_1=U$, the diagonal bare algebra gives the two commuting families of Takasaki's continuous Pfaff--Toda hierarchy \cite{Taka}. The additional off-diagonal directions combine with them through the powers of the bare operator  
\begin{equation*}
\Lambda=
\begin{pmatrix}
0&1\\
\mathcal S^{-1}&0
\end{pmatrix},
\qquad
\Lambda^2=\mathcal S^{-1}I_2.
\end{equation*}
Indeed, we have 
\begin{equation*}
\Lambda^{2n}=T_n-\widebar T_n, ~ n \geq 1 , \qquad \Lambda^{2m+1}=Q_m^++Q_m^-, ~ m \geq 0,
\end{equation*}
where 
\begin{equation*}
T_n=\mathcal S^{-n}E_{11},
\qquad
\widebar T_n=-\mathcal S^{-n}E_{22},
\qquad
Q_m^+=\mathcal S^{-m-1}E_{21},
\qquad
Q_m^-=\mathcal S^{-m}E_{12}.
\end{equation*}
\begin{theorem*}[$N=1$: Pfaff--Toda, its odd extension, and the Adler--Pfaff hierarchy]
For $N=1$, the dressing flows associated with the two diagonal families
$T_n$ and $\widebar T_n$ coincide, under the identification
\begin{equation*}
\mathcal S=\,\textup{e}^{-\partial_s},
\end{equation*}
with Takasaki's continuous Pfaff--Toda hierarchy \cite{Taka}.
Moreover, after the fixed shift gauge~$Q$ and the diagonal degree-zero normalization~$H$ applied as $G:=HQ$ to $L_\Lambda=U\Lambda\, U^{-1}$, the operator $\mathcal L:=G L_\Lambda G^{-1} $ satisfies the Adler--Pfaff Lax equations on
the extended phase space mentioned at the end of \hyperlink{res1}{Result~1}:
\begin{equation*}
D_{\Lambda^k}\mathcal{L}=
\bigl[P^{\mathrm{Pf}}_+(\mathcal{L}^k),\mathcal{L}\bigr]=
\bigl[-P^{\mathrm{Pf}}_-(\mathcal{L}^k),\mathcal{L}\bigr],
\qquad k\geq1.
\end{equation*}
The commuting family generated by the powers of $\Lambda$
realizes the Adler--Pfaff hierarchy on this extended phase space.
\end{theorem*}
At the level of dressing
derivations, the relations between the bare operators imply
\begin{equation*}
D_{\Lambda^{2n}}=D_{T_n}-D_{\widebar T_n},
\qquad
D_{\Lambda^{2m+1}}=D_{Q_m^+}+D_{Q_m^-}.
\end{equation*}
In particular, the even powers of $\Lambda$ generate the even-flow
subhierarchy of Adler--Pfaff. The further reduction to the classical even Pfaff hierarchy~\cite{benassi2021symmetric,BDM2} will appear later.

\medskip 
\hypertarget{res4}{\textbf{4. Savchenko--Zabrodin tau-functions and the diagonal Matrix Pfaff--Toda hierarchy.}}

For general $N$, we compare the commuting diagonal sector of the Matrix Pfaff--Toda hierarchy with the multicomponent Pfaff--Toda hierarchy of Savchenko and Zabrodin \cite{SavchenkoZabrodin,SavZab25DKP,SavZab26DispPfaffToda,SavZab26Elliptic}. The former is defined through matrix pseudodifference dressing, whereas the latter is constructed from a Clifford-group element and a fermionic bilinear identity. We prove that the continuous Savchenko--Zabrodin flows realize the diagonal dressing equations of our matrix hierarchy.

The fermionic bilinear identity produces two opposite raw matrix
pseudodifference dressing operators. They are related by a skew
operator-valued form $K_N$, which need not coincide with the fixed
form $J_N$ of the Matrix Pfaff--Toda hierarchy. To compare the two
constructions, one must simultaneously normalize $K_N$ to $J_N$ and
place one dressing operator in the negative dressing group. This
defines the regular factorization locus, on which the normalization is
unique.

For $\alpha=1,\ldots,N$ and $n\geq1$, we introduce the commuting diagonal bare
operators
\begin{equation*}
T_{\alpha,n}=\mathcal S^{-n}
\begin{pmatrix}
E_{\alpha\alpha}&0\\
0&0
\end{pmatrix},\qquad \widebar T_{\alpha,n}=-\mathcal S^{-n}
\begin{pmatrix}
0&0\\
0&E_{\alpha\alpha}
\end{pmatrix}.
\end{equation*}
They generate the commuting diagonal sector of the Matrix
Pfaff--Toda hierarchy.

\begin{theorem*}[Savchenko--Zabrodin realization]
Let $g$ be a Clifford-group element satisfying the Savchen\-ko--Zabrodin
fermionic bilinear identity of \cite{SavchenkoZabrodin} and belonging to the regular
factorization locus. Let $\mathcal V_g$ be the corresponding coefficient
difference--differential algebra and let $Y_N$ be the normalized matrix
dressing operator associated with the tau-function of $g$.

Then evaluation of the universal dressing coefficients at those of $Y_N$
defines an $\mathcal S$-difference algebra homomorphism
\begin{equation*}
\operatorname{ev}_g:
\mathcal V_N\longrightarrow\mathcal V_g,
\qquad
\mathcal U_N\longmapsto Y_N,
\end{equation*}
which intertwines the diagonal Matrix Pfaff--Toda flows with the continuous
Savchenko--Zabrodin flows:
\begin{equation*}
\operatorname{ev}_g\circ D_{T_{\alpha,n}}
=
\frac{\partial}{\partial t_{\alpha,n}}
\circ\operatorname{ev}_g,
\qquad
\operatorname{ev}_g\circ D_{\widebar T_{\alpha,n}}
=
\frac{\partial}{\partial\widebar t_{\alpha,n}}
\circ\operatorname{ev}_g,
\end{equation*}
for every $\alpha=1,\ldots,N$ and $n\geq1$.
\end{theorem*}
Regularity is automatic for $N=1$, while for $N\geq2$ we exhibit an explicit $\binom{N}{2}$-parameter family of genuinely coupled
Clifford-group elements on which it holds.

\medskip 
\hypertarget{res5}{\textbf{5. Toda reductions and the global picture.}}

The block-diagonal dressing reduction of the Matrix Pfaff--Toda hierarchy, followed by the restriction of the bare algebra to its diagonal Laurent subalgebra and a change of dressing presentation,
yields the multicomponent 2D Toda hierarchy~\cite{UT84,MMAF09,TZ25}. For $N=1$ this becomes scalar $2D$ Toda.
On the Pfaff side, the even powers of $\Lambda$ generate the even
subhierarchy of Adler--Pfaff, while after symmetric gauging the anti-diagonal scalar Toda flows give the even Pfaff hierarchy~\cite{benassi2021symmetric,BDM2}.
These relations are summarized in the  diagram below.

A further reduction, which lies outside of the diagram, is the Krichever--Zabrodin constraint
\cite{KriZab} on the symmetrically gauged scalar Toda dressing pair,
$\widebar W^{\rm s}(W^{\rm s})^*=1$.
It is preserved by the anti-diagonal Toda flows and, under the
preceding Pfaff realization, defines an invariant reduction of the
even Pfaff hierarchy. The resulting scalar Lax hierarchy is C--Toda.

The diagram below provides a global summary of the hierarchies studied in this paper and the reductions relating them. The position of each hierarchy within the scheme depends on the nontrivial identifications established in the main body of the paper.
{\small
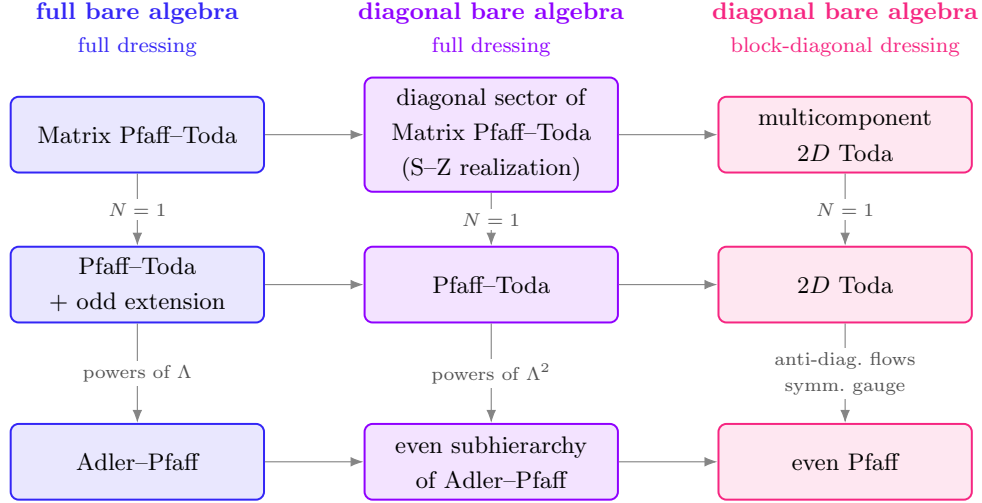
\begin{equation*}
\begin{tikzpicture}[>=Latex,nodino/.style={minimum height=1.0cm,minimum width=3.35cm,rounded corners=3pt,line width=.75pt,align=center,inner sep=4pt},
    arrow/.style={->,draw=black!48,line width=.42pt},lab/.style={font=\scriptsize,text=black!68,fill=white,inner sep=1.5pt},leftbox/.style={nodino,draw=bluebonnet,fill=bluebonnet!11},midbox/.style={nodino,draw=electricviolet,fill=electricviolet!11},rightbox/.style={nodino,draw=strawberrypink,fill=strawberrypink!11},every text node part/.style={align=center},scale=.9]

\node[text=bluebonnet,align=center] at (0,1.55)
    {{\small\bfseries full bare algebra}\\[-1pt]
     {\footnotesize full dressing}};

\node[text=electricviolet,align=center] at (5.2,1.55)
    {{\small\bfseries diagonal bare algebra}\\[-1pt]
     {\footnotesize full dressing}};

\node[text=strawberrypink,align=center] at (10.4,1.55)
    {{\small\bfseries diagonal bare algebra}\\[-1pt]
     {\footnotesize block-diagonal dressing}};

\node[leftbox] (a1) at (0,0)
    {Matrix Pfaff--Toda};

\node[midbox] (b1) at (5.2,0)
    {diagonal sector of \\
Matrix Pfaff--Toda \\
(S--Z realization)};

\node[rightbox] (c1) at (10.4,0)
    {multicomponent\\$2D$ Toda};

\node[leftbox] (a2) at (0,-2.2)
    {Pfaff--Toda\\$+$ odd extension};

\node[midbox] (b2) at (5.2,-2.2)
    {Pfaff--Toda};

\node[rightbox] (c2) at (10.4,-2.2)
    {$2D$ Toda};

\node[leftbox] (a3) at (0,-4.8)
    {Adler--Pfaff};

\node[midbox] (b3) at (5.2,-4.8)
    {even subhierarchy\\of Adler--Pfaff};

\node[rightbox] (c3) at (10.4,-4.8)
    {even Pfaff};

\draw[arrow] (a1) -- (b1);
\draw[arrow] (b1) -- (c1);

\draw[arrow] (a2) -- (b2);
\draw[arrow] (b2) -- (c2);

\draw[arrow] (a3) -- (b3);
\draw[arrow] (b3) -- (c3);

\draw[arrow] (a1) -- (a2) node[midway,lab] {$N=1$};
\draw[arrow] (b1) -- (b2) node[midway,lab] {$N=1$};
\draw[arrow] (c1) -- (c2) node[midway,lab] {$N=1$};

\draw[arrow] (a2) -- (a3) node[midway,lab] {powers of $\Lambda$};
\draw[arrow] (b2) -- (b3) node[midway,lab] {powers of $\Lambda^2$};
\draw[arrow] (c2) -- (c3) node[midway,lab] {anti-diag.\ flows \\ symm.\ gauge};

\end{tikzpicture}
\end{equation*}}
In the diagram, horizontal arrows denote restrictions of the bare algebra and dressing class, while vertical arrows denote the one-component specialization $N=1$ and the subsequent passage to the even or anti-diagonal reductions. All these reductions preserve the dressing equations, so that the integrable structures established for the Matrix Pfaff--Toda hierarchy descend to each node of the scheme.

\medskip

The paper is organized as follows. In Section~\ref{sec:adler-pfaff} we introduce the bi-infinite Adler--Pfaff hierarchy and we derive its pseudodifference realization (\hyperlink{res1}{Result~1}). In Section~\ref{sec:matrix-pfaff-toda-representation} we develop the Matrix Pfaff--Toda dressing representation, including the Pfaff--Toda splitting, the zero-curvature Lax representation of the full Laurent algebra (\hyperlink{res2}{Result~2}), and the Savchenko--Zabrodin realization theorem (\hyperlink{res4}{Result~4}), whose proof is deferred to Appendix~\ref{sec:SZ-identification}. Section~\ref{sec:scalar-extension} specializes to $N=1$, recovering Takasaki's Pfaff--Toda hierarchy and reconstructing Adler--Pfaff from the odd extension via powers of $\Lambda$ (\hyperlink{res3}{Result~3}). Finally, in Section~\ref{sec:toda-reductions} we identify in our setting the Toda reductions, including the multicomponent $2D$ Toda hierarchy, the even Pfaff subhierarchy, and the C--Toda reduction (\hyperlink{res5}{Result~5}).

\paragraph{Acknowledgements.}
SC is supported by BK21 Seoul National University. The authors used OpenAI ChatGPT during the preparation of this manuscript for language and organizational assistance and as an interactive aid for exploratory discussion and checking computations. The mathematical arguments, proofs, and final formulations presented in the manuscript were developed and finalized by the authors, who take full responsibility for the content.

\section{Adler--Pfaff hierarchy in pseudodifference form}
\label{sec:adler-pfaff}
The Pfaff lattice was introduced by Adler, Horozov and van Moerbeke
\cite{vanM_Pfaff_skew,adler2002} in the context
of orthogonal and symplectic random matrix ensembles, where partition
functions are expressed as Pfaffians of skew-symmetric moment matrices.
It admits a natural formulation in terms of bi-infinite matrices, where
the splitting of the algebra is induced by a non-trivial symplectic
involution. This splitting, together with the associated Lax equations,
gives rise to an integrable hierarchy whose phase space is the Hessenberg
orbit of the first sub-diagonal shift.

We formulate the Adler--Pfaff hierarchy on its bi-infinite Hessenberg phase space and identify its translation-covariant matrix realization with a $2\times 2$ matrix pseudodifference algebra. Under this
identification, the Adler--Pfaff splitting and Lax equations are transported directly to the pseudodifference setting. The resulting pseudodifference Lax hierarchy is the natural $2\times2$ matrix extension of the scalar shift-operator formalism for lattice integrable systems developed by Blaszak and Marciniak \cite{BlaMar} and Suris \cite{Suris1996,Suris1999}.

\subsection{The bi-infinite Adler--Pfaff hierarchy}
\label{subsec:biinfinite-pfaff}
We recall the bi-infinite Adler--Pfaff hierarchy on its block-Hessenberg phase space and fix the notation needed for its pseudodifference realization below. 

Let $\mathcal V$ be a commutative $\mathbb C$-algebra. We denote by $\mathfrak{gl}_{\infty,2}^{\textup{\,sup}}(\mathcal V)$ the associative algebra of bi-infinite matrices written in $2\times2$ blocks with finitely many super-diagonals
\begin{equation}\label{eq:gl_sup}
\mathfrak{gl}_{\infty,2}^{\textup{\,sup}}(\mathcal V)
:=\left\{ X=(X_{ij})_{i,j\in\mathbb Z},~X_{ij}\in  M_2(\mathcal V) \mid
\exists N_X\in\mathbb N_0:\ X_{ij}=0\ \text{for}\ j-i>N_X
\right\}.
\end{equation}
The condition that only finitely many block super-diagonals are non-zero ensures that matrix multiplication is well-defined: for fixed $i,j$, only finitely many intermediate indices $k$ contribute to the sum $\sum_k X_{ik}Y_{kj}$. Set 
\begin{equation}
J_\infty=\operatorname{diag}_{n\in\mathbb Z}(J_0), \qquad J_0= \begin{pmatrix}
0&1\\
-1&0
\end{pmatrix}, \qquad J_0^2 = -I_2.
\end{equation}
The condition $X=J_\infty X^\top J_\infty$ forces $X$ to have
finite block bandwidth. We define the Lie subalgebra
\begin{equation}\label{eq:symplectic_subalgebra}
\mathfrak{sp}_\infty(\mathcal V):=\left\{
X\in\mathfrak{gl}^{\,\mathrm{sup}}_{\infty,2}(\mathcal V) \mid  X=J_\infty X^\top J_\infty
\right\}.
\end{equation}
A direct block-wise computation gives the Lie algebra splitting 
\begin{equation}\label{eq:splitting}
\mathfrak{gl}_{\infty,2}^{\,\textup{sup}}(\mathcal V)= \mathfrak t(\mathcal V)\oplus\mathfrak{sp}_\infty(\mathcal V),
\end{equation}
where the complement subalgebra 
 $\mathfrak t(\mathcal V)$ consists of elements that are block lower triangular with scalar diagonal blocks
\begin{equation}\label{eq:triangular_subalgebra}
    \mathfrak t(\mathcal V):= \left\{X\in\mathfrak{gl}_{\infty,2}^{\,\textup{sup}}(\mathcal V) \mid X_{ij}=0\ \text{for }i<j,\quad X_{ii}\in \mathcal V I_2 \right\}.
\end{equation}
The projection onto $\mathfrak t(\mathcal V)$ parallel to $\mathfrak{sp}_\infty(\mathcal V)$ is given explicitly by
\begin{equation}\label{eq:P_triangle}
\bigl(P_{\mathfrak t}(X)\bigr)_{ij}=
\begin{cases}
0, & i<j,\\[.5mm]
\frac12\,\textup{tr}(X_{ii})I_2, & i=j,\\[.5mm]
X_{ij}-J_0\, X_{ji}^\top, J_0, & i>j,
\end{cases}
\end{equation}
We then set 
\begin{equation}\label{eq:P_sp}
    P_\mathfrak{sp}:=1-P_\mathfrak{t} \, . 
\end{equation}
We construct the bi-infinite Adler--Pfaff Hessenberg operator. For $n\in\mathbb Z$ and $m\geq0$, let $A_{n,m}$ be a generic
$2\times2$ matrix and let $b_n$ be a generic scalar.  Let
$\mathcal V_{\mathrm{Pf}}$ be the polynomial algebra over $\mathbb C$
freely generated by their entries. Let $E_{21} \in {M}_2(\mathbb C)$ be the standard matrix unit. The Adler--Pfaff Hessenberg operator $L_{\textup{Pf}}\in \mathfrak{gl}_{\infty,2}^{\,\textup{sup}}(\mathcal V_{\text{Pf}})$ is then determined by 
    \begin{equation}
        (L_{\mathrm{Pf}})_{n,n+1}=b_nE_{21}, \qquad  (L_{\mathrm{Pf}})_{n,n-m}=A_{n,m}, \qquad m\ge 0,
    \end{equation}
    and the remaining blocks equal to zero, i.e.\ 
    \begin{equation}\label{eq:Adler_Pfaff_Lax}
        L_{\mathrm{Pf}}=\begin{pmatrix}
    \ddots & \ddots &         &        &        \\
    \ddots & A_{-1,0} & b_{-1}E_{21} &        &        \\
    \ddots & A_{0,1}& A_{0,0}    & b_0E_{21} &        \\
    \ddots & A_{1,2}& A_{1,1}& A_{1,0}    &b_1 E_{21} \\
       & \ddots & \ddots & \ddots & \ddots
    \end{pmatrix}.
    \end{equation}
     We denote by 
    \begin{equation}
        \mathcal H_{\textup{Pf}} = \bigl\{  L \in \mathfrak{gl}_{\infty,2}^{\,\textup{sup}}(\mathcal V_{\text{Pf}}) \mid L_{ij}=0~\textup{for }j>i+1,~L_{i,i+1} \in \mathcal{V}_{\textup{Pf}} E_{21} \bigr\} 
    \end{equation}
the corresponding block-Hessenberg phase space.  
The Adler--Pfaff hierarchy in the bi-infinite setting is defined by
\begin{equation}
\partial_{t_k} L_{\mathrm{Pf}}
=-\bigl[P_{\mathfrak t}(L_{\mathrm{Pf}}^k),L_{\mathrm{Pf}}\bigr]
= \bigl[P_{\mathfrak{sp}}(L_{\mathrm{Pf}}^k),L_{\mathrm{Pf}}\bigr], \qquad k\ge 1.
\label{adler2}
\end{equation}
Let
$
Q_k:=P_{\mathfrak t}(L_{\mathrm{Pf}}^k).
$
Since $Q_k$ is block lower triangular, the commutator
$[Q_k,L_{\mathrm{Pf}}]$ has no blocks above the first block
super-diagonal. Moreover, the diagonal blocks $(Q_k)_{ii}$ are of the form $q_{k,i}I_2$, hence the $(i,i+1)$-block of $[Q_k,L_{\mathrm{Pf}}]$
is proportional to $E_{21}$ for all $i \in \mathbb Z$, proving that
the vector fields in~\eqref{adler2} preserve
$\mathcal H_{\mathrm{Pf}}$.
 The Lax equation gives the time derivative of every matrix coefficient of $L_{\mathrm{Pf}}$ as a polynomial in the coefficients of $L_{\mathrm{Pf}}$. Since $\mathcal V_{\mathrm{Pf}}$ is the polynomial algebra freely generated by these coefficients, these formulas uniquely define derivations
\begin{equation}
    \partial_{t_k}\colon \mathcal V_{\textup{Pf}}\longrightarrow \mathcal V_{\textup{Pf}}, \qquad k\ge 1.
\end{equation}
These derivations pairwise commute
($[\partial_{t_k},\partial_{t_\ell}]=0$ for $k,\ell\ge1$), which follows from the standard zero-curvature condition for Lax hierarchies.

\begin{remark}[Even Pfaff]
\label{rem:even-pfaff}
Let $\mathcal I_{\mathrm{ev}}\subset \mathcal V_{\mathrm{Pf}}$ be the ideal generated by the diagonal entries of all the matrices $A_{n,r}$. Equivalently, on the quotient $\mathcal V_{\mathrm{Pf}}/\mathcal I_{\mathrm{ev}}$ every $2\times2$ block of $L_{\mathrm{Pf}}$ is anti-diagonal. The even Pfaff hierarchy is the hierarchy induced by the even derivations $\partial_{t_{2k}}$, $k\geq1$, on the quotient
\begin{equation*}
\mathcal V_{\mathrm{Pf}}^{\mathrm{ev}}:=\mathcal V_{\mathrm{Pf}}/\mathcal I_{\mathrm{ev}}.
\end{equation*}
For every $k\ge 1$, $\partial_{t_{2k}}(\mathcal I_{\mathrm{ev}})\subset \mathcal I_{\mathrm{ev}}$, so the even flows are well-defined on the quotient. Indeed, modulo $\mathcal I_{\mathrm{ev}}$ every block of $L_{\mathrm{Pf}}$
is anti-diagonal. Hence every block of $L_{\mathrm{Pf}}^{2k}$ and
 $P_{\mathfrak t}(L_{\mathrm{Pf}}^{2k})$ is diagonal, so their
commutator with $L_{\mathrm{Pf}}$ is block anti-diagonal. The semi-infinite version of this system was introduced in~\cite{benassi2021symmetric,BDM2} as the Pfaffian analogue of the Volterra hierarchy. We return to the even Pfaff reduction in Section~\ref{sec:toda-reductions}.
\end{remark}

\subsection{Translation-covariant matrices and pseudodifference operators}
\label{subsec:pfaff-isomorphism}
The Adler--Pfaff Lax operator has a built-in translation covariance: shifting the block indices by one is equivalent to applying a shift $\mathcal S$ to the coefficients. Starting from the bi-infinite algebra introduced in Section~\ref{subsec:biinfinite-pfaff}, here we construct an explicit isomorphism $\Phi$ with the algebra of $2\times 2$ matrix pseudodifference operators.

We encode translation of the block index as a shift on the
coefficient algebra. Define the automorphism $\mathcal S$ of $\mathcal V_{\textup{Pf}}$ whose action on its generators is given by 
\begin{equation}
    \mathcal{S}(b_n) = b_{n-1}, \qquad \mathcal{S}(A_{n,m}) = A_{n-1,m}, \qquad n \in \mathbb Z, \, m \ge 0.
\end{equation}
We extend $\mathcal S$ coefficient-wise to
$\mathfrak{gl}^{\,\mathrm{sup}}_{\infty,2}(\mathcal V_{\mathrm{Pf}})$,
\begin{equation*}
(\mathcal S M)_{ij}:=\mathcal S(M_{ij}), ~ i,j \in \mathbb Z.
\end{equation*}
The Lax operator~\eqref{eq:Adler_Pfaff_Lax} for the Adler--Pfaff hierarchy satisfies the covariance relation
\begin{equation}
(L_{\textup{Pf}})_{i+1,j+1} = \mathcal S^{-1}((L_{\textup{Pf}})_{ij}), ~ i,j \in \mathbb Z.
\label{Liscovariant}
\end{equation}
\begin{definition}
We introduce the algebra of translation-covariant bi-infinite block matrices: 
\begin{equation}
    \mathcal C_{\textup{Pf}} = \left\{ M \in  \mathfrak{gl}_{\infty,2}^{\textup{\,sup}}(\mathcal V_{\textup{Pf}}) \mid  M_{i+1,j+1} = \mathcal{S}^{-1}(M_{ij}), ~ i,j \in \mathbb Z
    \right\}. 
\end{equation}
\end{definition}
It is preserved
by matrix multiplication,
$P_{\mathfrak t}$ and $P_{\mathfrak{sp}}$.
 \begin{lemma}
The shift automorphism $\mathcal{S}$ commutes with the derivations $(\partial_{t_k})_{k \geq 1}$.
\end{lemma}

\begin{proof}
  Let us denote the RHS of~\eqref{adler2} by $F_k$. For all $i, j \in \mathbb Z$ we have 
\begin{equation}
\partial_{t_k}
(L_{\mathrm{Pf}})_{ij}
=
(F_k)_{i,j}
\label{auxeqq}
\end{equation}
By the preceding observation, $\mathcal C_{\mathrm{Pf}}$ is a subalgebra of $\mathfrak{gl}_{\infty,2}^{\textup{\,sup}}(\mathcal V_{\textup{Pf}})$ preserved by $P_{\mathfrak t}$ and $P_{\mathfrak sp}$. Hence $F_k$ belongs to $\mathcal C_{\mathrm{Pf}}$:
\begin{equation}
(F_k)_{i+1,j+1}=\mathcal{S}^{-1}((F_k)_{ij}), ~ i,j \in \mathbb Z.
\label{Fiscovariant}
\end{equation}
Combining~\eqref{Liscovariant},~\eqref{auxeqq}, and~\eqref{Fiscovariant}, we deduce that
\begin{equation*}
\partial_{t_k}\mathcal S^{-1}
\bigl((L_{\mathrm{Pf}})_{ij}\bigr)
=
\mathcal S^{-1}\partial_{t_k}
\bigl((L_{\mathrm{Pf}})_{ij}\bigr), ~ i,j \in \mathbb Z.
\end{equation*}
As the coefficients of $L_{\mathrm{Pf}}$ generate
$\mathcal V_{\mathrm{Pf}}$, the claim follows.
\end{proof}

\begin{definition}
 The space of $\mathcal S$-adic pseudodifference operators on $\mathcal V_{\mathrm{Pf}}$ is 
\begin{equation}
\mathcal V_{\mathrm{Pf}}((\mathcal S)) = \Big\{\sum_{k=-n}^{\infty} a_k \mathcal S^k \mid a_k\in\mathcal V_{\mathrm{Pf}},\ n \in \mathbb N\Big\}.
\end{equation}
It is an associative algebra for the product 
\begin{equation}\label{eq:associative_shift}
    (a\mathcal S^k)(b \mathcal{S}^\ell) = a \, \mathcal{S}^k(b)\,\mathcal{S}^{k+\ell},  ~ a, b \in \mathcal V_{\mathrm{Pf}}, ~ k, \ell \in \mathbb Z.
\end{equation}
The algebra $\mathcal V_{\mathrm{Pf}}((\mathcal S))$ is a difference algebra for the coefficient-wise action of the shift $\mathcal S$, and the subalgebra of Laurent polynomials in $\mathcal S$ is denoted by
\begin{equation}
    \mathcal V_{\mathrm{Pf}}[\mathcal S,\mathcal S^{-1}] = \Big\{\sum_{k=-n}^{m} a_k \mathcal S^k \mid a_k\in\mathcal V_{\mathrm{Pf}},~ n,m \in \mathbb N \Big\}.
\end{equation}
Finally, for $P \in \mathcal V_{\rm Pf}((\mathcal S))$, we define the triangular splittings
\begin{equation}\label{eq:splitting_triangular}
P_{>0}= \sum_{k=1}^{\infty} a_k \mathcal S^k, \qquad P_0=a_0, \qquad P_{<0}=\sum_{k=-n}^{-1} a_k \mathcal S^k.
\end{equation}
\end{definition}
Because of the covariance condition, every $M\in\mathcal C_{\rm Pf}$ is determined by its zeroth block column. The shift on the coefficient algebra records translation of the block index, so it is natural to encode the block $M_{0,-k}$ as the coefficient of $\mathcal S^k$. This leads to the map $\Phi:\mathcal C_{\mathrm{Pf}}\to
M_2\bigl(\mathcal V_{\mathrm{Pf}}((\mathcal S))\bigr) $
\begin{equation}\label{eq:Phi}
\Phi(M):=
\sum_{k\in\mathbb Z}M_{0,-k}\mathcal S^k.
\end{equation}
\begin{proposition} $\Phi$ is an isomorphism of difference algebras.
\end{proposition}

\begin{proof}
    We consider two elements $M,N$ in $\mathcal C_{\mathrm{Pf}}$, and write
    \begin{equation*}
        \Phi(M)=\sum_{i\in \mathbb Z}A_i \mathcal{S}^i, \qquad \Phi(N)=\sum_{j \in \mathbb Z}B_j \mathcal{S}^j. 
    \end{equation*}
    Since $N$ is $\mathcal{S}$-covariant we have for all $k \in \mathbb Z$
    \begin{equation*}
        (MN)_{0,-k} = \sum_{r \in \mathbb Z} M_{0,-r}\,N_{-r,-k} = \sum_{r \in \mathbb Z} A_r\,N_{-r,-k}=\sum_{r \in \mathbb Z} A_r\,\mathcal{S}^r(B_{k-r}).
    \end{equation*}
    On the other hand, the product of $\Phi(M)$ and $\Phi(N)$ using~\eqref{eq:associative_shift} is 
    \begin{equation*}
        \Phi(M)\Phi(N) = \Bigl( \sum_{i\in \mathbb Z}A_i \mathcal{S}^i\Bigr) \Bigl( \sum_{j \in \mathbb Z}B_j \mathcal{S}^j \Bigr) = \sum_{k \in \mathbb Z} \Bigl( \sum_{i\in \mathbb Z}A_i \mathcal{S}^i(B_{k-i} )\Bigr) \mathcal{S}^k, 
    \end{equation*}
    therefore $\Phi(MN)=\Phi(M)\Phi(N)$. 

Conversely, let
$
A=\sum_{k\geq -n}A_k\mathcal S^k
\in M_2\bigl(\mathcal V_{\mathrm{Pf}}((\mathcal S))\bigr)
$
and define
\begin{equation*}
M_{ij}:=\mathcal S^{-i}(A_{i-j}).
\end{equation*}
Since $A_k=0$ for $k<-n$, we have $M_{ij}=0$ whenever
$j-i>n$, so
$
M\in
\mathfrak{gl}^{\,\mathrm{sup}}_{\infty,2}(\mathcal V_{\mathrm{Pf}}).
$
Moreover,
\begin{equation*}
M_{i+1,j+1}
=
\mathcal S^{-i-1}(A_{i-j})
=
\mathcal S^{-1}(M_{ij}),
\end{equation*}
hence $
M\in\mathcal C_{\mathrm{Pf}}.
$
Thus $\Phi(M)=A$, so the map $\Phi$ is bijective. Finally, by construction,
\begin{equation*}
\Phi(\mathcal S(M))=\mathcal S(\Phi(M)),
\end{equation*}
so $\Phi$ is an isomorphism of difference algebras.
\end{proof}
The formal adjoint is the anti-isomorphism
\begin{equation*}
*:M_2(\mathcal V_{\mathrm{Pf}}((\mathcal S)))
\longrightarrow
M_2(\mathcal V_{\mathrm{Pf}}((\mathcal S^{-1})))
\end{equation*}
defined by
\begin{equation*}
\Bigl(\sum_k A_k \mathcal S^k\Bigr)^{\!*}=\sum_k \mathcal S^{-k}A_k^\top .
\end{equation*}
On the Laurent subalgebra
$M_2(\mathcal V_{\mathrm{Pf}}[\mathcal S,\mathcal S^{-1}])$ this is an involutive
anti-automorphism.
On the finite-band subalgebra of $\mathcal C_{\mathrm{Pf}}$, one has
\begin{equation*}
\Phi(M^\top)=\Phi(M)^*,
\qquad
\Phi(J_\infty)=J_0.
\end{equation*}
More generally, matrix transpose exchanges the two one-sided
translation-covariant matrix completions, just as the formal adjoint
exchanges $
M_2\bigl(\mathcal V_{\mathrm{Pf}}((\mathcal S))\bigr)$ and $
M_2\bigl(\mathcal V_{\mathrm{Pf}}((\mathcal S^{-1}))\bigr).
$

\subsection{Pseudodifference realization of the Adler--Pfaff hierarchy}
\label{subsec:pfaff-pseudodifference}
Let $\mathcal L_{\rm{Pf}}:=\Phi(L_{\rm{Pf}})$. By construction of the map $\Phi$~\eqref{eq:Phi}, 
\begin{equation}
    \mathcal L_{\rm{Pf}} = b_0E_{21} \mathcal{S}^{-1}+\sum_{k\ge 0} A_{0,k} \,\mathcal{S}^{k}. 
\end{equation}
The pseudodifference image of the universal Hessenberg phase space is 
\begin{equation}\label{eq:phase_space_pseudodiff}
    \mathcal M_{\rm Pf}:=\Phi(\mathcal{H}_{\textup{Pf}} \cap \mathcal C_{\textup{Pf}}) = \mathcal V_{\textup{Pf}}\, E_{21}\mathcal S^{-1}+ M_2(\mathcal V_{\textup{Pf}}[[\mathcal S]] ). 
\end{equation}
We then transport the Pfaff splitting~\eqref{eq:splitting} to the Lie algebra 
\begin{equation}\label{eq:g_M2}
    \mathfrak g := M_2(\mathcal V_{\textup{Pf}}((\mathcal{S})) ),
\end{equation}
and define the following subspaces of $\mathfrak g$
\begin{align}
\label{eq:gPfaff_plus}
    \mathfrak g^{\textup{Pf}}_+ &:= \left\{ X\in M_2(\mathcal V_{\mathrm{Pf}}[\mathcal S,\mathcal S^{-1}]) \mid X^*=J_0XJ_0 \right\}, \\[.5mm]
\label{eq:gPfaff_minus}    
    \mathfrak g^{\textup{Pf}}_- &:= \mathcal V_{\textup{Pf}}\,I_2 + \mathcal S\,M_2\!\left(\mathcal V_{\textup{Pf}}[[\mathcal S]]\right). 
\end{align}

\begin{proposition}
Under the isomorphism $\Phi$~\eqref{eq:Phi},
\begin{equation*}
    \Phi\bigl(\mathfrak t(\mathcal V_{\textup{Pf}}) \cap \mathcal C_{\mathrm{Pf}} \bigr)=\mathfrak g^{\mathrm{Pf}}_-, \qquad \Phi\bigl(\mathfrak{sp}_\infty(\mathcal V_{\textup{Pf}}) \cap \mathcal C_{\mathrm{Pf}} \bigr) =\mathfrak g^{\textup{Pf}}_+.
\end{equation*}
Consequently, we have the direct sum decomposition 
\begin{equation}
    \mathfrak g = \mathfrak g^{\textup{Pf}}_+ \oplus \mathfrak g^{\textup{Pf}}_-.
    \label{decomp}
\end{equation} 
\end{proposition}

\begin{proof}
The first identity follows directly from the lower triangular support
condition and the scalar diagonal blocks. For the second identity, every element of
$\mathfrak{sp}_\infty(\mathcal V_{\mathrm{Pf}})$~\eqref{eq:symplectic_subalgebra} has finite block
bandwidth, and hence
\begin{equation*}
\Phi(M^\top)=\Phi(M)^*,
\qquad
\Phi(J_\infty)=J_0.
\end{equation*}
Thus, for $M\in\mathcal C_{\mathrm{Pf}}$ of finite block
bandwidth,
\begin{equation*}
M=J_\infty M^\top J_\infty
\quad\Longleftrightarrow\quad
\Phi(M)^*=J_0\Phi(M)J_0,
\end{equation*}
which proves the second identity. Since $P_{\mathfrak t}$ and
$P_{\mathfrak{sp}}$ preserve $\mathcal C_{\mathrm{Pf}}$, applying
$\Phi$ to
\begin{equation*}
\mathcal C_{\mathrm{Pf}}
=
\bigl(\mathfrak t(\mathcal V_{\mathrm{Pf}})
\cap\mathcal C_{\mathrm{Pf}}\bigr)
\oplus
\bigl(\mathfrak{sp}_\infty(\mathcal V_{\mathrm{Pf}})
\cap\mathcal C_{\mathrm{Pf}}\bigr)
\end{equation*}
gives the decomposition~\eqref{decomp}.
\end{proof}

Let $
P^{\mathrm{Pf}}_+:\mathfrak g\to\mathfrak g^{\mathrm{Pf}}_+
$
be the projection along $\mathfrak g^{\mathrm{Pf}}_-$, and set
$P^{\mathrm{Pf}}_-:=1-P^{\mathrm{Pf}}_+$. For
\begin{equation*}
X=
\begin{pmatrix}
A&B\\
C&D
\end{pmatrix}\in\mathfrak g,
\end{equation*}
one has
\begin{equation}\label{plussplittngpfaffpseudo}
P^{\mathrm{Pf}}_+(X)
=\begin{pmatrix}
A_{<0}+\frac12(A-D)_0-(D^*)_{>0}&B_{\leq0}+(B^*)_{>0}
\\[2mm]
C_{\leq0}+(C^*)_{>0}&D_{<0}+\frac12(D-A)_0-(A^*)_{>0}
\end{pmatrix},
\end{equation}
 in terms of the triangular splittings~\eqref{eq:splitting_triangular}. 
Indeed, if $Y=P_+^{\mathrm{Pf}}(X)$, then
\begin{equation*}
Y^*=J_0YJ_0,
\qquad
X-Y\in\mathfrak g_-^{\mathrm{Pf}}.
\end{equation*}
The second condition fixes the negative diagonal parts and the
non-positive off-diagonal parts of $Y$, while the scalar degree-zero
condition determines its diagonal zero mode. The relation
$Y^*=J_0YJ_0$ then determines the remaining positive parts, yielding
\eqref{plussplittngpfaffpseudo}.
\begin{theorem}[Pseudodifference Adler--Pfaff hierarchy] 
\label{thm:pseudodiff_AdlerPfaff}
Let $\mathcal L_{\rm{Pf}}=\Phi(L_{\mathrm{Pf}})$.
Under the isomorphism~$\Phi$ in~\eqref{eq:Phi}, the Adler--Pfaff hierarchy
\eqref{adler2} becomes
\begin{equation*}
\partial_{t_k} \mathcal L_{\rm{Pf}}=\bigl[P^{\mathrm{Pf}}_+(\mathcal L_{\rm{Pf}}^k),\mathcal L_{\rm{Pf}}\bigr]=-\bigl[P^{\mathrm{Pf}}_-(\mathcal L_{\rm{Pf}}^k),\mathcal L_{\rm{Pf}}\bigr],
\qquad k\geq1.
\end{equation*} 
\end{theorem}
\begin{proof}
This follows by applying $\Phi$ to~\eqref{adler2} and using
\begin{equation*}
\Phi P_{\mathfrak{sp}}=P^{\mathrm{Pf}}_+\Phi,
\qquad
\Phi P_{\mathfrak t}=P^{\mathrm{Pf}}_-\Phi.
\end{equation*}  
\end{proof}
Thus the bi-infinite Adler--Pfaff hierarchy is expressed entirely in $2\times2$ matrix pseudodifference form, which is the operator language used below to compare it with the Pfaff--Toda dressing construction.

\begin{remark}[Even Pfaff in pseudodifference form]
Under $\Phi$, the even Pfaff Lax operator of Remark~\ref{rem:even-pfaff} has an anti-diagonal shape:
\begin{equation}
    \mathcal L_{\rm{Pf}}^{\textup{ev}}=\begin{pmatrix}
0 & A\\
B & 0
\end{pmatrix}.
\end{equation}
The even Pfaff hierarchy is therefore the restriction of the even Adler--Pfaff flows to this anti-diagonal pseudodifference shape. This description will reappear, after the symmetric gauge, in Section~\ref{sec:toda-reductions}.
\end{remark}

\subsection{Extended Adler--Pfaff phase space}
\label{subsec:extended-pfaff-space}
The exact pseudodifference phase space $\mathcal M_{\rm Pf}$ in~\eqref{eq:phase_space_pseudodiff} does not contain three boundary coefficients that appear after the symmetric gauge of the one-component Pfaff--Toda operator in Section~\ref{sec:odd_extension}: the $\mathcal S^{-2}$-coefficient of the $(2,1)$-entry and the two diagonal $\mathcal S^{-1}$-coefficients. To construct a universal phase space containing the symmetrically gauged operator, we therefore adjoin three difference indeterminates $x,y,z$, representing these three additional coefficients, and set
\begin{equation*}
\mathcal{V}_{\rm ext}=\mathcal V_{\rm Pf}[\mathcal S^n x,\mathcal S^m y,\mathcal S^k z \, | \,  n,m,k \in \mathbb Z],
\end{equation*}
where all the generators are algebraically independent
over $\mathcal V_{\mathrm{Pf}}$.
Consider the following larger space
\begin{equation}\label{eq:M_ext_Pfaff}
    \mathcal M_{\mathrm{ext}}:=\left\{
L= \begin{pmatrix}
A & B\\
C & D
\end{pmatrix}  =  \begin{pmatrix}
\mathcal O(\mathcal S^{-1}) & \mathcal O(1)\\
\mathcal O(\mathcal S^{-2}) & \mathcal O(\mathcal S^{-1})
\end{pmatrix} \right\},
\end{equation}
or, entrywise, 
\begin{equation}
    A, D \in \mathcal S^{-1} \mathcal V_{\rm ext}[[\mathcal S]], \qquad B \in \mathcal V_{\rm ext}[[\mathcal S]], \qquad C \in \mathcal S^{-2}\mathcal V_{\rm ext}[[\mathcal S]]. 
\end{equation}
We denote by $\mathfrak g^{\rm ext}_+$ and $\mathfrak g^{\rm ext}_-$ the extensions of~\eqref{eq:gPfaff_plus} and~\eqref{eq:gPfaff_minus} to the coefficient algebra $\mathcal{V}_{\rm ext}$, and the corresponding projections by $P^{\rm ext}_+$ and $P^{\rm ext}_-$.
\begin{lemma}
\label{lem:left_right_invariance}
    The space $\mathcal M_{\mathrm{ext}}$ is invariant under the left and right action of $\mathfrak g^{\rm ext}_-$
    \begin{equation}
        \mathfrak g^{\rm ext}_- \mathcal M_{\mathrm{ext}} \subset \mathcal M_{\mathrm{ext}}, \qquad \mathcal M_{\mathrm{ext}}
\mathfrak g^{\rm ext}_- \subset \mathcal M_{\rm ext}.
    \end{equation}
\end{lemma}
\begin{proof}
Let $X \in \mathfrak g_-^{\rm ext}$. It can be written as 
\begin{equation*}
    X = a I_2 + X_{\ge 1}, 
\end{equation*}
with $a \in \mathcal V_{\rm ext}$, and $X_{\ge 1} \in \mathcal S M_2(\mathcal V_{\rm ext}[[\mathcal S]])$. Left and right multiplication by the scalar degree-zero term $aI_2$ preserve the support of $\mathcal M_{\mathrm{ext}}$, and by a direct check we have for $L \in \mathcal M_{\mathrm{ext}}$~\eqref{eq:M_ext_Pfaff}
\begin{equation*}
    X_{\ge 1} L = \begin{pmatrix}
        \mathcal O(\mathcal S^{-1}) & \mathcal O(1) \\
        \mathcal O(\mathcal S^{-1}) & \mathcal O(1)
    \end{pmatrix} \in \mathcal M_{\mathrm{ext}}, \qquad 
     L X_{\ge 1} = \begin{pmatrix}
        \mathcal O(1) & \mathcal O(1) \\
        \mathcal O(\mathcal S^{-1}) & \mathcal O(\mathcal S^{-1})
    \end{pmatrix} \in \mathcal M_{\mathrm{ext}}. 
\end{equation*}
\end{proof} 
We define the extended Lax operator as follows
\begin{equation*}
\mathcal L_{\rm ext}= xE_{21}\mathcal{S}^{-2}+(yE_{11}+zE_{22})\mathcal{S}^{-1}+\mathcal{L}_{\rm{Pf}}.
\end{equation*}
The coefficients of $\mathcal L_{\mathrm{ext}}$ freely generate
$\mathcal V_{\mathrm{ext}}$ as a difference algebra.
\begin{proposition}[Invariance of the extended phase space]
\label{prop:ext-invariance}
The hierarchy 
\begin{equation}
\partial_{t_k} \mathcal L_{\rm ext}= \bigl[P^{\mathrm{ext}}_+(\mathcal L_{\rm ext}^k),\mathcal L_{\rm ext}\bigr]=-\bigl[P^{\mathrm{ext}}_-(\mathcal L_{\rm ext}^k),\mathcal L_{\rm ext}\bigr], \qquad k \geq 1
\label{extendedhierarchy}
\end{equation}
extends the pseudodifference Adler--Pfaff hierarchy of
Theorem~\ref{thm:pseudodiff_AdlerPfaff} and defines a family of pairwise commuting derivations of $\mathcal V_{\rm ext}$.
\end{proposition}

\begin{proof}
    By Lemma~\ref{lem:left_right_invariance}, the RHS of~\eqref{extendedhierarchy} lies in $\mathcal M_{\rm ext}$. Since the difference algebra $\mathcal V_{\rm ext}$ is generated by the coefficients of $\mathcal L_{\rm ext}$,~\eqref{extendedhierarchy} indeed provides a well-defined evolutionary derivation of $\mathcal V_{\rm ext}$ for all $k \geq 1$. The fact that these derivations pairwise commute follows from the
standard zero-curvature argument associated with the Lie algebra
splitting
\begin{equation*}
\mathfrak g^{\mathrm{ext}}
=
\mathfrak g_+^{\mathrm{ext}}
\oplus
\mathfrak g_-^{\mathrm{ext}},
\end{equation*}
using the commutativity of the powers of
$\mathcal L_{\mathrm{ext}}$.
\end{proof}

\begin{remark}
The Adler--Pfaff hierarchy in its pseudodifference form can be defined on other invariant subspaces $\mathcal M$. For the sake of clarity we have singled out $\mathcal{M}_{\rm Pf}$ in~\eqref{eq:phase_space_pseudodiff} as the transport of the classical shape of the bi-infinite formulation, and $\mathcal{M}_{\rm ext}$ in~\eqref{eq:M_ext_Pfaff} as the extension of $\mathcal{M}_{\rm Pf}$  accommodating the embedding of Adler--Pfaff
into Pfaff--Toda described in Section~\ref{sec:odd_extension}. We leave a more systematic study of the Pfaff invariant subspaces
and their reductions for future work.
\end{remark}

\section{Matrix Pfaff--Toda dressing representation}
\label{sec:matrix-pfaff-toda-representation}
Section 2 recasts the Adler--Pfaff hierarchy in the language of $2\times2$ matrix pseudodifference operators. We now turn to the dressing side of the picture. The Pfaff--Toda hierarchy provides the natural starting point: its dressing datum is genuinely matrix-valued, while its continuous flows are generated by commuting diagonal families of bare operators. From the viewpoint of the dressing operator itself, however, there is no intrinsic reason to restrict to these diagonal directions. This motivates dressing the full matrix Laurent algebra.

The construction extends from $2\times2$ to $2N\times2N$ matrices and leads to the Matrix Pfaff--Toda hierarchy developed below. For general $N$, the multicomponent Pfaff--Toda hierarchy of Savchenko and Zabrodin \cite{SavchenkoZabrodin} will be shown in Section~\ref{sec:SZ-realization} to realize its commuting diagonal sector. The fermionic proof is given in Appendix~\ref{sec:SZ-identification}.

\subsection{Lax representations of associative algebras}
\label{sec:lax_representations}

We begin by introducing an abstract Lax formalism for an arbitrary associative algebra $\mathfrak a$, extending the usual Lax--Sato setting based on a finite family of commuting operators, to the case of a full associative algebra of bare operators. 

Let $\mathfrak a$ be an associative algebra over $\mathbb C$, equipped
with its commutator Lie bracket, and let $(\mathcal R,\mathcal S)$ be a
difference algebra. By a \emph{Lax representation} of
$\mathfrak a$ we mean the data of
\begin{equation*}
L:\mathfrak a\longrightarrow \mathcal R((\mathcal S)),
\qquad
D:\mathfrak a\longrightarrow \operatorname{Der}_{\mathcal S}(\mathcal R),
\qquad
\mathcal P:\mathfrak a\longrightarrow
\mathcal R[\mathcal S,\mathcal S^{-1}],
\end{equation*}
where $L$ is an algebra morphism, $D$ is a Lie algebra morphism, and
$\mathcal P$ is a $\mathbb{C}$-linear map such that for all $x,y \in \mathfrak{a}$,
\begin{equation}
D_x(L_y)
=
[\mathcal P_x,L_y]+L_{[x,y]}.
\label{eq:lax-representation}
\end{equation}
Here and below we use the notation
\begin{equation*}
L_x:=L(x),
\qquad
D_x:=D(x),
\qquad
\mathcal P_x:=\mathcal P(x).
\end{equation*}
The coefficient
algebra $\mathcal R$ may in general be noncommutative.
We denote by $\operatorname{Der}_{\mathcal S}(\mathcal R)$ the Lie algebra of
derivations of $\mathcal R$ commuting with the shift automorphism
$\mathcal S$, while $\mathcal R((\mathcal S))$ denotes the algebra of
pseudodifference operators with only finitely many negative powers
of $\mathcal S$. Every $D\in\operatorname{Der}_{\mathcal S}(\mathcal R)$ is extended
coefficient-wise to $\mathcal R((\mathcal S))$.
The auxiliary operators $\mathcal P_x$ determine the
curvature map
\begin{equation*}
\mathcal C:
\mathfrak a\times\mathfrak a
\longrightarrow
\mathcal R[\mathcal S,\mathcal S^{-1}],
\end{equation*}
defined by
\begin{equation}
\mathcal C(x,y)
:=
D_x(\mathcal P_y)
-D_y(\mathcal P_x)
-[\mathcal P_x,\mathcal P_y]
-\mathcal P_{[x,y]}.
\label{eq:lax-curvature}
\end{equation}
We call the Lax representation \emph{zero-curvature} if $\mathcal C \equiv 0$, equivalently
\begin{equation}
D_x(\mathcal P_y)-D_y(\mathcal P_x)
=
[\mathcal P_x,\mathcal P_y]
+\mathcal P_{[x,y]},
\qquad
x,y\in\mathfrak a.
\label{eq:zero-curvature}
\end{equation}
In particular, for commuting $x$, $y$, this
reduces to the usual zero-curvature equation
\begin{equation*}
D_x(\mathcal P_y)-D_y(\mathcal P_x)
=
[\mathcal P_x,\mathcal P_y].
\end{equation*}

\begin{remark}
The preceding notion extends the usual Lax--Sato formulation from a distinguished family of commuting operators to an entire associative algebra.  To recover the usual Lax--Sato situation,
suppose that the coefficient algebra decomposes as a direct sum
of $\mathcal S$-stable subalgebras
\begin{equation*}
\mathcal R
=
\mathcal R_1\oplus\cdots\oplus\mathcal R_r,
\qquad
\mathcal R_i\mathcal R_j=0
\quad (i\neq j),
\label{eq:R-decomposition}
\end{equation*}
and choose commuting elements
$x_1,\ldots,x_r\in\mathfrak a$ such that
\begin{equation*}
L_{x_i}\in\mathcal R_i((\mathcal S)).
\end{equation*}
Since $L$ is an algebra morphism,
\begin{equation*}
L_{x_i^n}=L_{x_i}^n,
\end{equation*}
and the restriction of the representation to the commuting
subalgebras generated by the $x_i$ gives the usual Lax--Sato
hierarchy on
\begin{equation*}
(L_{x_1},\ldots,L_{x_r}),
\qquad
\frac{\partial}{\partial t_n^i}
:=
D_{x_i^n},
\qquad
n\geq1,\quad i=1,\ldots,r.
\end{equation*}
When no such decomposition of the target algebra is available, a
formulation in terms of several individual Lax operators generally
requires additional compatibility relations between them.  The single map $L$ encodes these compatibility
relations simultaneously.
\end{remark}
\subsection{The Pfaff--Toda splitting}
\label{sec:multicomponent-hierarchy}

In this Section we construct the Lie algebra splitting that underlies the Matrix Pfaff--Toda dressing representation, where $\mathfrak g_+$ is defined as the fixed-point subalgebra of a symplectic involution, and $\mathfrak g_-$ is given by a prescribed triangular shape at the boundary degrees. 

Let $(\mathcal{V}, \mathcal{S})$ be a commutative difference algebra, and $N \in \mathbb{Z}_+$. Consider the associative algebra
\begin{equation*}
\mathfrak g = M_{2N}(\mathcal{V}((\mathcal{S}))).
\end{equation*}
The adjoint operation $*$ maps $\mathfrak g$ to $\tilde{ \mathfrak g}=M_{2N}(\mathcal{V}((\mathcal{S}^{-1})))$ and is defined by 
\begin{equation*}
(a \mathcal{S}^k)^*=\mathcal{S}^{-k}
a^\top, \qquad a \in M_{2N}(\mathcal{V}).
\end{equation*}
On the Laurent subalgebra $\mathfrak{g} \cap \tilde{\mathfrak g}=M_{2N}(\mathcal{V}[\mathcal{S},\mathcal{S}^{-1}])$, we consider the map 
\begin{equation}
 \theta (X)=J_NX^*J_N, \qquad J_N=
\begin{pmatrix}
0& \mathcal S^{-1}I_N\\
-\mathcal SI_N&0
\end{pmatrix},
\label{eq:JN}
\end{equation}
which is an involutive Lie algebra automorphism, since $J_N^2=-I_{2N}$ and $J_N^*=-J_N$. We denote its 
fixed-point Lie subalgebra by  \begin{equation} 
\mathfrak g_+ := \{ M \in M_{2N}(\mathcal{V}[\mathcal{S},\mathcal{S}^{-1}]) \mid M^*=J_NMJ_N \}. 
\label{adjointconstraint}
\end{equation}
The splitting below is modelled on the standard Iwasawa-type splitting of
loop algebras (see, for instance, \cite{PS86}). When $N=1$, the symmetric zero-mode decomposition
\begin{equation*}
\mathfrak{gl}_2=\mathfrak{sp}_2\oplus \mathbb C I_2
\end{equation*}
selects the scalar direction $\mathbb C I_2$ as a complement to the
fixed-point algebra. In order to recover Takasaki's Pfaff--Toda
projection \cite{Taka}, however, we replace this scalar direction by the asymmetric
choice $\mathbb C E_{22}$. 

For general $N$, we make the same choice on
each diagonal symplectic $2$-plane, together with the corresponding lower triangular ordering between the $N$ blocks. For any associative $\mathbb C$-algebra  $\mathcal{B}$, let $\mathfrak n_-\subset M_N(\mathcal B)$ be the space of strictly
lower triangular matrices, and $\mathfrak b_-\subset M_N(\mathcal B)$
the space of lower triangular matrices. For $k\in M_N(\mathcal B)$, we denote by $k_{\mathrm u}$ its upper triangular part including the diagonal, and by $k_{\mathrm{su}}$ its strictly upper triangular part. This asymmetric choice leads to the following Pfaff--Toda-adapted subspace of $\mathfrak g$:
\begin{equation}\label{negativesplitting}
\mathfrak{g}_{-}=
\left\{
\begin{pmatrix}
a+\mathcal O(\mathcal S)&b\mathcal S^{-1}+\mathcal O(1)\\
c\mathcal S+\mathcal O(\mathcal S^2)&d+\mathcal O(\mathcal S)
\end{pmatrix}
\;\middle|\;
a,b,c\in\mathfrak n_-,
\quad d\in\mathfrak b_-
\right\}.
\end{equation}

\begin{lemma}\label{lem:g_splitting}
    The subspaces $\mathfrak{g}_+$ and $\mathfrak{g}_{-}$ are Lie subalgebras of $\mathfrak{g}$. 
    Moreover, the projection onto $\mathfrak g_+$ along $\mathfrak g_-$ is given by:
    \begin{equation}
\Pi_+
\begin{pmatrix}
A&B\\[.5mm] 
C&D
\end{pmatrix}
=
\begin{pmatrix}
A_{<0}-(\mathcal S^{-1}D^*\mathcal S)_{>0}
+p(A,D) & B_{<-1} + (\mathcal S^{-1}B^*\mathcal S^{-1})_{\geq0} + q(B)\,\mathcal S^{-1} \\[1mm] 
C_{\leq0} + (\mathcal SC^*\mathcal S)_{>1} +\mathcal S\, r(C)
& D_{<0}-(\mathcal SA^*\mathcal S^{-1})_{>0}-
\mathcal S\,p(A,D)^\top\mathcal S^{-1}
\end{pmatrix},
\label{eq:Pi-plus}
\end{equation}
where 
\begin{equation}
\begin{aligned} 
\label{eq:pqr_lemma}
p(A,D)&:= (A_0)_{\textup{u}} - \big( ((\mathcal{S}^{-1}D\mathcal S)_0)_{\textup{su}} \big)^{\!\top}, \\[1mm]
q(B)&:= ((B \mathcal S)_0)_{\textup{u}} + \big( ((B \mathcal S)_0)_{\textup{su}} \big)^{\!\top},  \\[1mm]
r(C)&:=  ((\mathcal S^{-1} C)_0)_{\textup{u}} + \big( ((\mathcal S^{-1} C)_0)_{\textup{su}} \big)^{\!\top}.   
\end{aligned}
\end{equation}
 In particular, we have the Lie algebra decomposition 
    \begin{equation*} \mathfrak g= \mathfrak{g}_+ \oplus \mathfrak{g}_-, \qquad
    \textup{ im } \Pi_+= \mathfrak g_+, \qquad \textup{ ker } \Pi_+= \mathfrak g_-.
    \end{equation*}
\end{lemma}

\begin{proof}
Since $\theta$
is an involutive Lie algebra automorphism, its fixed-point space
$\mathfrak g_+=\operatorname{Fix}(\theta)$ is a Lie subalgebra.
We next show that $\mathfrak g_-$ is a Lie subalgebra. In fact, it is
an associative subalgebra. Let
\begin{equation*}
X=
\begin{pmatrix}
a+\mathcal O(\mathcal S) & b\mathcal S^{-1}+\mathcal O(1)\\
c\mathcal S+\mathcal O(\mathcal S^2) & d+\mathcal O(\mathcal S)
\end{pmatrix},
\qquad
Y=
\begin{pmatrix}
\widetilde a+\mathcal O(\mathcal S) & \widetilde b\mathcal S^{-1}+\mathcal O(1)\\
\widetilde c \mathcal S+\mathcal O(\mathcal S^2) & \widetilde d+\mathcal O(\mathcal S)
\end{pmatrix}
\end{equation*}
belong to $\mathfrak g_-$. The coefficients of $XY$ have the form
\begin{eqnarray*}
\begin{aligned}
(XY)_{11}
&=
a\widetilde a+
b \mathcal S^{-1}(\widetilde c)+\mathcal O(\mathcal S),& \qquad 
(XY)_{12}
&=
\bigl(a\widetilde b+
b\mathcal S^{-1}(\widetilde d)\bigr)\mathcal S^{-1}+\mathcal O(1),\\[1mm]
(XY)_{21}
&=
\bigl(c\mathcal S(\widetilde a)+d\widetilde c\bigr)\mathcal S+\mathcal O(\mathcal S^2),& \qquad 
(XY)_{22}
&=
c\mathcal S(\widetilde b)+d\widetilde d+\mathcal O(\mathcal S).
\end{aligned}
\end{eqnarray*}
Since the shift preserves triangularity, and
\begin{equation*}
\mathfrak n_-\mathfrak n_-\subset\mathfrak n_-,
\qquad
\mathfrak n_-\mathfrak b_-\subset\mathfrak n_-,
\qquad
\mathfrak b_-\mathfrak n_-\subset\mathfrak n_-,
\qquad
\mathfrak b_-\mathfrak b_-\subset\mathfrak b_-,
\end{equation*}
all four coefficients have the required 
form. Hence $XY\in\mathfrak g_-$, so $\mathfrak g_-$ is an associative,
and therefore a Lie subalgebra of $\mathfrak g$.

It remains to identify $\Pi_+$ with the projection onto
$\mathfrak g_+$ along $\mathfrak g_-$.  We first determine its kernel.  Take $X \in \textup{ ker } \Pi_+$. Looking at the diagonal coefficients of $\Pi_+(X)$, we see that both $A$ and $D$ must be in $M_N(\mathcal{V}[[\mathcal{S}]])$. At degree zero in $\mathcal{S}$, the diagonal component of $\Pi_+(X)=0$ reads 
\begin{equation*}
(A_0)_{\textup{u}}=(\mathcal{S}^{-1}(D_0)_{\textup{su}})^\top.
\end{equation*}
Since the LHS is upper triangular and the RHS is strictly lower triangular, we get $(A_0)_{\textup{u}}=(D_0)_{\textup{su}}=0$, hence $(A_0, D_0) \in \mathfrak n_- \times \mathfrak b_-$.
As for the $(1,2)$-entry of $\Pi_+(X)=0$,  the equation
\begin{equation*}
B_{<-1}
+
(\mathcal S^{-1}B^*\mathcal S^{-1})_{\geq0}
+
((B_{-1})_{\textup{u}}+((B_{-1})_{\textup{su}})^\top)\mathcal S^{-1} =0
\end{equation*} 
holds if and only if $B=B_{-1}\mathcal{S}^{-1}+\mathcal O(1)$ with 
$B_{-1} \in \mathfrak{n}_-$. Similarly, the $(2,1)$-entry of $\Pi_+(X)$ is zero if and only if $C=C_1 \mathcal{S} +\mathcal O(\mathcal{S}^2)$ with $C_1 \in \mathfrak{n}_-$. Hence we have shown that 
\begin{equation*}
\textup{ ker } \Pi_+ =\mathfrak g_-.
\end{equation*}
A similar computation shows that $X-\Pi_+(X)\in\mathfrak g_-=\ker\Pi_+$ for any $X \in \mathfrak g$. Moreover, taking into account the symplectic involution defining the fixed-point space we have  
\begin{equation*}
X=\begin{pmatrix}
    E & F \\
    G & H
\end{pmatrix}
\in \mathfrak g_+ \iff 
\begin{pmatrix}
    \,E^* & G^* \\
    F^* & H^* 
\end{pmatrix}
= 
\begin{pmatrix}
    -\mathcal{S}^{-1}H\mathcal{S} & \mathcal{S}^{-1}G \mathcal{S}^{-1}\\
    \mathcal{S} F \mathcal{S} & -\mathcal{S} E \mathcal{S}^{-1}
\end{pmatrix}
\iff X=\Pi_+(X).
\end{equation*}
Therefore $\textup{ im } \Pi_+=\mathfrak g_+$ and $\Pi_+^2=\Pi_+$, which concludes the proof.
\end{proof}
Thus the symbols $+$ and $-$ label the two Lie-algebra summands rather than
the signs of the Laurent degrees: $\mathfrak g_+$ has finite Laurent support,
whereas $\mathfrak g_-$ allows an infinite positive Laurent tail. Away from the boundary degrees, the
projection $\Pi_+$ is determined by the adjoint constraint~\eqref{adjointconstraint}. In particular, the terms $p,q,r$ in~\eqref{eq:Pi-plus} enforce the prescribed triangularity
at degree $0$ in the diagonal blocks, degree $-1$ in the upper right block,
and degree $1$ in the lower-left block.
We finally introduce 
\begin{equation}\label{eq:proj_Pi_minus}
\Pi_-:=1-\Pi_+.    
\end{equation}
\begin{remark}\label{rem:derivation_PfaffToda}
Let $\partial$ be an evolutionary derivation of $\mathcal V$, extended
coefficient-wise to $\mathfrak g$. Since $\partial$ commutes with $\mathcal S$
and acts entrywise on matrix coefficients, it commutes with the formal adjoint,
the Laurent truncations, and the triangular projections appearing
in~\eqref{eq:Pi-plus}. Hence
\begin{equation*}
\partial\Pi_\pm=\Pi_\pm\partial.
\end{equation*}
This allows us to differentiate the projected components of
dressed operators in the zero-curvature calculation below.
\end{remark}

\subsection{Dressing and the zero-curvature representation}
\label{sec:dressing_zero_curvature}

We now pass from the splitting to the associated dressing space, introducing a universal dressing operator $U_N$ whose infinitesimal variations are exactly $\mathfrak g_-$. For each bare matrix difference operator~$x$, the dressing equation defines a derivation $D_x$, producing a zero-curvature Lax representation of the full bare algebra. We retain Takasaki’s notation for the four blocks of $U_N$ \cite{Taka}, anticipating the $N=1$ comparison in Section 4.

Let $\mathcal B_N$ be the difference polynomial algebra freely generated by the coefficients
of the dressing operator
\begin{equation}\label{eq:dressing-U}
U_N=
\begin{pmatrix}
W_1 & \widebar V_1\\
W_2 & \widebar V_2
\end{pmatrix}
=
\begin{pmatrix}
I_N+a_0+\mathcal O(\mathcal S)
&
b_{-1}\mathcal S^{-1}+\mathcal O(1)
\\[1mm]
c_1\mathcal S+\mathcal O(\mathcal S^2)
&
\Delta+\mathcal O(\mathcal S)
\end{pmatrix},
\end{equation}
where $a_0,b_{-1},c_1\in\mathfrak n_-$, $\Delta\in\mathfrak b_-$. We define $\mathcal V_N$ as the localization of the difference algebra
$\mathcal B_N$ at $\det\Delta$ and all of its $\mathcal S$-shifts, so that
$\mathcal S^k(\Delta)$ is invertible for every $k\in\mathbb Z$. This is the
only localization required in the inversion argument below. Set 
\begin{equation*}
\mathfrak g_N:=M_{2N}(\mathcal V_N((\mathcal S))).
\end{equation*}
\begin{lemma}\label{lem:UN_invertible}
    $U_N$ is invertible in $\mathfrak g_N$.
\end{lemma}
\begin{proof}
Set
\begin{equation*}
Z=
\begin{pmatrix}
I_N&0\\
0&\mathcal S I_N
\end{pmatrix}.
\end{equation*}
By~\eqref{eq:dressing-U}, we have
\begin{equation*}
Z^{-1}U_NZ=Y+\mathcal O(\mathcal S),
\qquad
Y=
\begin{pmatrix}
I_N+a_0 & b_{-1}\\
\mathcal S^{-1}(c_1) & \mathcal S^{-1}(\Delta)
\end{pmatrix}.
\end{equation*}
Since $a_0\in\mathfrak n_-$, the matrix $I_N+a_0$ is unit lower
triangular and hence invertible. Its Schur complement in $Y$ is
\begin{equation*}
K=
\mathcal S^{-1}(\Delta)
-\mathcal S^{-1}(c_1)(I_N+a_0)^{-1}b_{-1}.
\end{equation*}
The second term is strictly lower triangular. Hence $K$ is lower triangular
with the same diagonal as $\mathcal S^{-1}(\Delta)$, and therefore
\begin{equation*}
\det K=\mathcal S^{-1}(\det\Delta).
\end{equation*}
By the localization defining $\mathcal V_N$, this determinant is invertible.
Thus $K$ is invertible, and so is $Y$.
We may therefore write
\begin{equation*}
Z^{-1}U_NZ
=
Y(I_{2N}+R),
\qquad
R\in
\mathcal S M_{2N}(\mathcal V_N[[\mathcal S]]).
\end{equation*}
The operator $I_{2N}+R$ is invertible by the $\mathcal S$-adic geometric
series. Hence $Z^{-1}U_NZ$, and therefore $U_N$, is invertible in
$\mathfrak g_N$.
\end{proof}
Applying the construction of Section~\ref{sec:multicomponent-hierarchy}, with
$\mathcal V=\mathcal V_N$ gives the Lie algebra splitting
\begin{equation*}
\mathfrak g_N=(\mathfrak g_N)_+\oplus(\mathfrak g_N)_-
\end{equation*}
with projections $\Pi_\pm$.
Since
\begin{equation}
U_N-I_{2N}\in(\mathfrak g_N)_-,
\label{trick}
\end{equation}
and the localization introduces no additional relations among the independent
coefficients of $U_N$, the admissible infinitesimal variations retain exactly
the form prescribed by $(\mathfrak g_N)_-$. Hence the tangent space to the
dressing space at $U_N$ is identified with $(\mathfrak g_N)_-$.
We define the associative algebra of bare matrix Laurent difference operators as
\begin{equation*}
\mathfrak a_N= M_{2N}(\mathbb{C}[\mathcal{S}, \mathcal{S}^{-1}])
\end{equation*}
and the algebra morphism $L\colon \mathfrak{a}_N \rightarrow \mathfrak g_N $
\begin{equation*}
L_x= U_N\, x \,U_N^{-1}.
\end{equation*}
For every $x \in \mathfrak a_N$, let $L_x=P_x^+ + P_x^-$ be the decomposition of $L_x$ in $(\mathfrak{g}_N)_+ \oplus (\mathfrak{g}_N)_-$.

\begin{lemma}
 For every $x\in\mathfrak a_N$,
\begin{equation*}
P_x^-U_N\in(\mathfrak g_N)_-.
\end{equation*}
Hence $P_x^-U_N$ defines a tangent vector to the dressing space.
\end{lemma}

\begin{proof}
As noted in the proof of Lemma~\ref{lem:g_splitting}, $(\mathfrak{g}_N)_-$ is an associative subalgebra of $\mathfrak{g}_N$. By~\eqref{trick}, we have 
\begin{equation*}
P_x^-U_N=P_x^-+P_x^-(U_N-I_{2N}) \in (\mathfrak{g}_N)_-.
\end{equation*}
\end{proof}
\begin{definition}\label{def:dressing_derivative}
For every bare operator $x\in\mathfrak a_N$, we define
$D_x\in\operatorname{Der}_{\mathcal S}(\mathcal V_N)$ by
\begin{equation*}
D_x(U_N):=-P_x^-U_N=P_x^+U_N-U_Nx.
\end{equation*}
Indeed, since the independent coefficients of $U_N$ freely generate
$\mathcal B_N$, this prescription determines a unique
$\mathcal S$-commuting derivation
\begin{equation*}
D_x:\mathcal B_N\longrightarrow\mathcal V_N.
\end{equation*}
By the universal property of localization, it extends uniquely to
$\mathcal V_N$.
\end{definition}
The sign convention is chosen to agree with the Pfaff--Toda time evolutions used below \cite{Taka}; accordingly, the Lax representation in the sense of Section~\ref{sec:lax_representations} will involve $-D$.
\begin{proposition}\label{prop:dressing_derivations_zero}
For all $x,y\in\mathfrak a_N$, the dressing derivations satisfy
\begin{equation*}
    \begin{aligned}
D_x(L_y)&=[P_x^+,L_y]-L_{[x,y]}, \\
[D_x,D_y]&=-D_{[x,y]}, \\
D_x(P_y^+)-D_y(P_x^+)&=
[P_x^+,P_y^+]-P_{[x,y]}^+.
\end{aligned} 
\end{equation*}
Equivalently, the maps $L$, $-D$ and $-P^+$ form a zero-curvature
Lax representation of $\mathfrak a_N$.
\end{proposition}
\begin{proof}
We verify successively the Lax relation, the Lie algebra property of
$-D$, and the zero-curvature identity.
For any $x, y \in \mathfrak a_N$, we have 
\begin{align*}
D_x(L_y)
&=
D_x(U_N)yU_N^{-1}
-
U_NyU_N^{-1}D_x(U_N)U_N^{-1}\\
&= [D_x(U_N)U_N^{-1}, U_N yU_N^{-1} ] \\
&= [-P_x^-, L_y]\\
&= [P_x^+, L_y]-L_{[x,y]}.
\end{align*}
To prove the Lie algebra relation, we compute the  commutator
of two dressing derivations.
By Remark~\ref{rem:derivation_PfaffToda},
\begin{equation*}
D_x(P_y^-)
=
D_x\bigl(\Pi_-(L_y)\bigr)
=
\Pi_-\bigl(D_x(L_y)\bigr)
=
-\Pi_-\bigl([P_x^-,L_y]\bigr).
\end{equation*}
Hence
\begin{align*}
D_xD_y(U_N)
&=
-D_x(P_y^-U_N)
\\
&=
\Pi_-\bigl([P_x^-,L_y]\bigr)U_N
+
P_y^-P_x^-U_N.
\end{align*}
Interchanging $x$ with $y$ and subtracting, we find
\begin{align*}
[D_x,D_y](U_N)
={}&
\Bigl(
\Pi_-\bigl(
[P_x^-,L_y]-[P_y^-,L_x]
\bigr)
-[P_x^-,P_y^-]
\Bigr)U_N.
\end{align*}
Since $(\mathfrak g_N)_-$ is a Lie subalgebra,
we get
\begin{equation*}
\begin{aligned}
\Pi_-\bigl([P_x^-,L_y]-[P_y^-,L_x]\bigr)
   -[P_x^-,P_y^-]
&=
\Pi_-\bigl(
[P_x^-,L_y]-[P_y^-,L_x]-[P_x^-,P_y^-]
\bigr)
\\
&=
\Pi_-\bigl([L_x,L_y]-[P_x^+,P_y^+]\bigr)
\\
&=
\Pi_-([L_x,L_y])
\\
&=
\Pi_-(L_{[x,y]})
\\
&=
P^-_{[x,y]}.
\end{aligned}
\end{equation*}
Hence, by Definition~\ref{def:dressing_derivative},
\begin{equation*}
[D_x,D_y](U_N)
=
-D_{[x,y]}(U_N).
\end{equation*}
Since the coefficients of $U_N$ freely generate $\mathcal B_N$, the equality
holds on all of $\mathcal B_N$. Moreover, since both derivations extend uniquely to the
localization $\mathcal V_N$, it follows that
\begin{equation*}
[D_x,D_y]=-D_{[x,y]}.
\end{equation*}
Finally, for all $x, y \in \mathfrak a_N$, we have 
\begin{equation*}
\begin{aligned}
D_x(P_y^+)&=-\Pi_+([P_x^-,L_y])\\
&= -\Pi_+( [P_y^-,P_x^+]-[P_x^+,P_y^+]+P_{[x,y]}^+)
\\
&= D_y(P_x^+)+[P_x^+,P_y^+]-P_{[x,y]}^+.
\end{aligned}
\end{equation*}
Therefore the curvature of the Lax representation $(L,-D,-P^+)$ vanishes.
\end{proof}
In particular, commuting bare operators generate commuting evolutionary flows. We refer to the resulting full matrix dressing system as the \textit{Matrix Pfaff--Toda hierarchy}: its complete family of dressing directions is non-abelian, whereas every commuting subalgebra of bare operators determines a commuting subhierarchy.

\subsection{Savchenko--Zabrodin realization of the diagonal sector}
\label{sec:SZ-realization}

We compare the Matrix Pfaff--Toda hierarchy with the Savchenko--Zabrodin multicomponent construction, outlining the normalization of its raw dressing pair and stating the realization theorem, whose proof is deferred to Appendix~\ref{sec:SZ-identification}.

Let $g$ be a Clifford-group element satisfying the fermionic bilinear
identity of \cite{SavchenkoZabrodin}. Its associated tau-function is
\begin{equation*}
    \tau(\boldsymbol n,\widebar{\boldsymbol n},
         \boldsymbol t,\widebar{\boldsymbol t})
    =
    \bra{ \boldsymbol n}
    \,\textup{e}^{J(\boldsymbol t)}
    g\,
    \,\textup{e}^{-\widebar J(\widebar{\boldsymbol t})}
    \ket{-\widebar{\boldsymbol n}}.
\end{equation*}
Following \cite{SavchenkoZabrodin}, we use discrete
variables $\boldsymbol s, \boldsymbol r \in \mathbb Z^N$, set $\boldsymbol n=\boldsymbol s+\boldsymbol r$ and 
    $\widebar{\boldsymbol n}=\boldsymbol r-\boldsymbol s$,
and introduce the matrix tau-function
\begin{equation*}
    \tau_{\alpha\beta}
    (\boldsymbol s,\boldsymbol r,
     \boldsymbol t,\widebar{\boldsymbol t})
    :=
    \tau(
        \boldsymbol n+e_\alpha-e_\beta,
        \widebar{\boldsymbol n},
        \boldsymbol t,\widebar{\boldsymbol t}),
    \qquad
    \alpha, \beta =1,\ldots,N, \qquad
    \tau_0:=\tau_{\alpha\alpha}.
\end{equation*}
The shift automorphism acts on the first argument by
\begin{equation*}
    (\mathcal Sf)(\boldsymbol s,\boldsymbol r,
         \boldsymbol t,\widebar{\boldsymbol t})
    =
    f(\boldsymbol s-\mathbf 1,\boldsymbol r,
      \boldsymbol t,\widebar{\boldsymbol t}),
    \qquad
    \mathbf 1=(1,\ldots,1).
\end{equation*}
Appendix~\ref{sec:SZ-identification} associates with every admissible $g$ a coefficient
difference--differential algebra $\mathcal V_g$, obtained from the
tau-functions and their discrete translates by localizing at $\tau_0$
and all of its $\mathcal S$-translates.
The fermionic bilinear identity produces a pair of raw matrix
pseudodifference dressing operators~(Lemma~\ref{lem:XNbarXN-def})
\begin{equation*}
    X_N\in M_{2N}\bigl(\mathcal V_g[[\mathcal S]]\bigr),
    \qquad
    \widebar X_N\in M_{2N}\bigl(\mathcal V_g[[\mathcal S^{-1}]]\bigr).
\end{equation*}
These operators do not in general have the canonical normalization of
the dressing operator $U_N$ in~\eqref{eq:dressing-U}. Instead, the bilinear
identity implies (Proposition~\ref{prop:raw-pfaff-relation}) the algebraic relation
\begin{equation*}
    X_NJ_N\widebar X_N^*
    =
    \widebar X_NJ_NX_N^*
    =:K_N.
\end{equation*}
This skew operator-valued
form $K_N$ need not coincide with the fixed form $J_N$ entering
the Pfaff--Toda splitting.
In Appendix~\ref{sec:SZ-identification} we introduce a regular factorization condition (Definition~\ref{def:regular-factorization-locus}) on the leading coefficients of these raw dressing operators. On the corresponding locus, Proposition~\ref{prop:GammaN} gives a unique matrix difference operator $\Gamma_N\in
    M_{2N}\bigl(\mathcal V_g[\mathcal S,\mathcal S^{-1}]\bigr)$
such that
\begin{equation*}
    \Gamma_NK_N\Gamma_N^*=J_N,
\end{equation*}
 where $ \Gamma_NX_N$ has the form~\eqref{eq:dressing-U}.
Defining $Y_N:= \Gamma_N X_N$ and 
    $\widebar Y_N:=\Gamma_N\widebar X_N,$
we obtain
\begin{equation*}
    Y_NJ_N\widebar Y_N^*=J_N,
    \qquad
    \widebar Y_N=-J_N(Y_N^*)^{-1}J_N.
\end{equation*}
In particular, $Y_N$ may be regarded as a specialization of the
universal dressing operator $U_N$ of the Matrix Pfaff--Toda hierarchy.
We consider the $2N$ commuting families of diagonal bare operators
\begin{equation*}
    T_{\alpha,n}
    :=
    \mathcal S^{-n}
    \begin{pmatrix}
        E_{\alpha\alpha} & 0\\
        0 & 0
    \end{pmatrix},
    \qquad
    \widebar T_{\alpha,n}
    :=
    -\mathcal S^{-n}
    \begin{pmatrix}
        0 & 0\\
        0 & E_{\alpha\alpha}
    \end{pmatrix},
    \qquad
    n\geq 1, \qquad  \alpha=1,\ldots,N,
\end{equation*}
generating a commuting sector of the Matrix Pfaff--Toda hierarchy constructed in Section~\ref{sec:dressing_zero_curvature}. We now show that the Savchenko--Zabrodin tau-functions realize precisely this commuting diagonal sector.

\begin{theorem}[Savchenko--Zabrodin realization]
\label{thm:SZ-realization}
Let $g$ be a Clifford-group element satisfying the Savchenko--Zabrodin
fermionic bilinear identity and belonging to the regular factorization
locus defined in Appendix~\ref{sec:SZ-identification}. Then the normalized dressing
operator $Y_N$ determines an $\mathcal S$-difference algebra homomorphism
\begin{equation*}
    \operatorname{ev}_g:
    \mathcal V_N\longrightarrow\mathcal V_g,
    \qquad
    U_N\longmapsto Y_N,
\end{equation*}
such that, for every $\alpha=1,\ldots,N$ and $n\geq1$,
\begin{equation*}
    \operatorname{ev}_g\circ D_{T_{\alpha,n}}
    =
    \frac{\partial}{\partial t_{\alpha,n}}
    \circ\operatorname{ev}_g,
    \qquad
    \operatorname{ev}_g\circ D_{\widebar T_{\alpha,n}}
    =
    \frac{\partial}{\partial\widebar t_{\alpha,n}}
    \circ\operatorname{ev}_g.
\end{equation*}
Therefore, every tau-function arising from a Savchenko--Zabrodin
Clifford-group element in the regular factorization locus gives a
realization of the commuting diagonal sector of the Matrix Pfaff--Toda
hierarchy.
\end{theorem}
The proof is given in Appendix~\ref{sec:SZ-identification} as Theorem~\ref{thm:fermionic_realisation_PfaffToda}. We briefly describe its mechanism. For a diagonal bare direction
$x\in\{T_{\alpha,n},\widebar T_{\alpha,n}\}$, we write
\begin{equation*}
\partial_{T_{\alpha,n}}:=\frac{\partial}{\partial t_{\alpha,n}},
\qquad
\partial_{\widebar T_{\alpha,n}}
:=
\frac{\partial}{\partial\widebar t_{\alpha,n}}.
\end{equation*}
Differentiating the fermionic bilinear identity produces two raw
auxiliary operators
\begin{equation*}
    B_x^{\rm raw}
    =
    (\partial_xX_N)X_N^{-1}
    +X_NxX_N^{-1},
\qquad
    \widebar B_x^{\rm raw}
    =
    (\partial_x\widebar X_N)\widebar X_N^{-1}
    +\widebar X_N\theta(x)\widebar X_N^{-1}.
\end{equation*}
The key step is Lemma~\ref{lem:Braw_barBraw}, which gives
$
    B_x^{\rm raw}=\widebar B_x^{\rm raw}.
$
After normalization by $\Gamma_N$, the corresponding auxiliary operators
remain equal:
$
    B_x=\widebar B_x.
$
Moreover, Proposition~\ref{prop:Bx_evolution} shows that the common auxiliary
operator has finite Laurent support and satisfies
\begin{equation*}
    B_x=\theta(B_x),
\end{equation*}
therefore
$B_x\in\mathfrak g_+$. 
Since $Y_N$ has the canonical dressing form~\eqref{eq:dressing-U}, its logarithmic derivative belongs to
$\mathfrak g_-$.
By uniqueness of the splitting
$
    \mathfrak g=\mathfrak g_+\oplus\mathfrak g_-,
$
we conclude that
\begin{equation*}
    B_x
    =
    \Pi_+\bigl(Y_NxY_N^{-1}\bigr),
    \qquad
    \partial_xY_N
    =
    -\Pi_-\bigl(Y_NxY_N^{-1}\bigr)Y_N.
\end{equation*}
This is precisely the specialization of the universal dressing equation
of Definition~\ref{def:dressing_derivative}. 

The regular factorization locus is nonempty. For $N=1$, regularity is automatic by Lemma~\ref{lem:regularization-N1}. For $N\geq2$,
Lemma~\ref{lem:gB_cliff} provides the family
\begin{equation*}
    g_B
    :=
    \exp\left(
        \frac12
        \sum_{\alpha,\beta=1}^N
        B_{\alpha\beta}\,
        \psi^{(\alpha)}_0\psi^{(\beta)}_0
    \right),
    \qquad
    B\in\mathfrak{so}_N(\mathbb C),
\end{equation*}
giving a $\binom{N}{2}$-parameter family of genuinely coupled
multicomponent examples. Hence the realization theorem applies beyond the decoupled sector.

\section{\texorpdfstring{$\boldsymbol {N= 1}$: from Pfaff--Toda to Adler--Pfaff}{PT}}
\label{sec:scalar-extension}

In this Section we specialize the Matrix Pfaff--Toda hierarchy to
$N=1$. Its diagonal bare algebra recovers Takasaki's continuous
Pfaff--Toda hierarchy, while the off-diagonal directions combine
with it, through the powers of a single bare operator, to recover
the Adler--Pfaff hierarchy.

We write $(U, \mathcal{V}, \mathfrak a, \mathfrak g, J)$ for $(U_1, \mathcal{V}_1, \mathfrak a_1, \mathfrak g_1, J_1)$.
When $N=1$ the triangular boundary terms disappear and the dressing
operator has the form
\begin{equation}\label{eq:scalar-U}
U=\begin{pmatrix}
W_1 & \widebar V_1\\
W_2 & \widebar V_2
\end{pmatrix}=
\begin{pmatrix}
1+\mathcal O(\mathcal S) & \mathcal O(1)\\
\mathcal O(\mathcal S^2) & \Delta+\mathcal O(\mathcal S)
\end{pmatrix} \in \mathfrak g := M_2(\mathcal{V}((\mathcal{S}))).
\end{equation}
The difference algebra $\mathcal{V}$ is obtained by localizing the difference polynomial algebra freely generated by the coefficients of $U$ at $\Delta$ and all its shifts.
Moreover, 
\begin{equation}
J=\begin{pmatrix}
0&\mathcal S^{-1}\\
-\mathcal S&0
\end{pmatrix}.
\label{eq:scalar-bar-U}
\end{equation}
The algebra of bare operators is $\mathfrak a= M_2(\mathbb{C}[\mathcal{S}, \mathcal{S}^{-1}])$
and the specialization of the Pfaff--Toda splitting of Section~\ref{sec:multicomponent-hierarchy} is $\mathfrak g=\mathfrak g_+\oplus\mathfrak g_-$, where
\begin{equation*}
\mathfrak g_+
=
\left\{
M\in M_2\bigl(\mathcal V[\mathcal S,\mathcal S^{-1}]\bigr)
\;\middle|\;
M^*=JMJ
\right\},
\qquad
\mathfrak g_-=
\begin{pmatrix}
\mathcal O(\mathcal S)&\mathcal O(1)\\
\mathcal O(\mathcal S^2)&\mathcal O(1)
\end{pmatrix},
\end{equation*}
since $a=b=c=0$ in~\eqref{negativesplitting} for scalars. 
The projection onto $\mathfrak g_+$ along $\mathfrak g_-$ is
\begin{equation}
\Pi_+
\begin{pmatrix}
A&B\\ C&D
\end{pmatrix}
=
\begin{pmatrix}
A_{\leq0}-(\mathcal S^{-1}D^*\mathcal S)_{>0}
&
B_{<0}+(\mathcal S^{-1}B^*\mathcal S^{-1})_{\geq0}
\\[1mm]
C_{\leq0}+(\mathcal S C^*\mathcal S)_{>0}
&
D_{<0}-(\mathcal S A^*\mathcal S^{-1})_{\geq0}
\end{pmatrix}.
\label{explicitPiplus}
\end{equation}
For $x\in\mathfrak a$, we use the notation
\begin{equation*}
L_x=UxU^{-1},
\qquad
P_x^+=\Pi_+(L_x),
\qquad
P_x^-=\Pi_-(L_x).
\end{equation*}
The evolutionary derivation of $\mathcal{V}$ associated with $x$ is
\begin{equation}
D_x(U)
=
-P_x^-U
=
P_x^+U-Ux .
\end{equation}
By Proposition~\ref{prop:dressing_derivations_zero},
\begin{equation}
D_x(L_y)=[P_x^+,L_y]-L_{[x,y]},
\qquad
[D_x,D_y]=-D_{[x,y]},
\qquad x,y\in\mathfrak a .
\end{equation}

\subsection{The continuous Pfaff--Toda flows}
\label{subsec:takasaki-flows}
We first restrict the dressing representation of $M_2(\mathbb{C}[\mathcal S, \mathcal S^{-1}])$ to the two
commuting diagonal families of bare operators and identify the resulting
Sato equations with the continuous Pfaff--Toda hierarchy of Takasaki \cite{Taka}.
For $n\geq1$, we introduce 
\begin{equation}
T_n:=\mathcal S^{-n}E_{11},
\qquad
\widebar T_n:=-\mathcal S^{-n}E_{22},
\label{TnbarTn}
\end{equation}
where $E_{ij}$ denotes the standard matrix unit.
Since the operators $T_n$ and $\widebar T_n$ commute pairwise, so do their dressing derivations.
By Definition~\ref{def:dressing_derivative}, 
\begin{equation} \label{ptflows}
D_{T_n}U=P_nU-UT_n, 
\qquad D_{\widebar T_n}U=\widebar P_nU-U \widebar T_n,
\end{equation}
where \begin{equation}
P_n:=\Pi_+(UT_nU^{-1}),
\qquad
\widebar P_n:=\Pi_+(U\widebar T_nU^{-1}).
\label{eq:Pn-Pbarn}
\end{equation}
Takasaki's dressing representation of Pfaff--Toda uses, in addition to
the $\mathcal S$-adic dressing operator $U$, an opposite
$\mathcal S^{-1}$-adic dressing operator. In our construction the latter
is not independent, but is determined by $U$ through the Pfaff involution.
We define
\begin{equation}\label{eq:scalar-barU}
\widebar U:=-J(U^*)^{-1}J 
=
\begin{pmatrix}
\widebar W_1&V_1\\
\widebar W_2&V_2
\end{pmatrix} 
\, \in \, M_2(\mathcal V((\mathcal S^{-1}))).
\end{equation}
We express the auxiliary operators $P_n$ and $\widebar P_n$ explicitly in terms of the  pseudodifference operators $V_i,W_i, \widebar V_i, \widebar W_i$. 
\begin{lemma}\label{lem:P_barP}
    For every $n \geq 1$, we have 
    \begin{equation}
        P_n=
        \begin{pmatrix}
            A_n & B_n \\
            C_n & D_n 
        \end{pmatrix}, 
        \qquad 
        \widebar P_n=
        \begin{pmatrix}
            \widebar A_n & \widebar B_n\\
            \widebar C_n & \widebar D_n
        \end{pmatrix}
    \end{equation}
    where 
\begin{eqnarray*}
        \begin{aligned}
            \,\,A_n&\!=\! \bigl(W_1\mathcal S^{-n-1}V_2^*\mathcal S\bigr)_{\leq0}\!
+\!\bigl(V_1\mathcal S^{n+1}W_2^*\mathcal S\bigr)_{>0},&\!\!
B_n&\!=\! -\bigl(W_1\mathcal S^{-n-1}V_1^*\mathcal S^{-1}\bigr)_{<0}\!-\!\bigl(V_1\mathcal S^{n+1}W_1^*\mathcal S^{-1}\bigr)_{\geq0}, \\[1mm]
\,\,C_n&\!=\! \bigl(W_2\mathcal S^{-n-1}V_2^*\mathcal S\bigr)_{\leq0}\!
+\!\bigl(V_2\mathcal S^{n+1}W_2^*\mathcal S\bigr)_{>0},&\!\!
D_n&\!=\!
-\bigl(W_2\mathcal S^{-n-1}V_1^*\mathcal S^{-1}\bigr)_{<0}\!
-\!\bigl(V_2\mathcal S^{n+1}W_1^*\mathcal S^{-1}\bigr)_{\geq0}, 
 \end{aligned}
    \end{eqnarray*}
and 
\begin{eqnarray*}
        \begin{aligned}    
\,\,\widebar A_n&\!=\!\bigl(\widebar W_1\mathcal S^{n-1}\widebar V_2^*\mathcal S\bigr)_{>0}\!
+\!\bigl(\widebar V_1\mathcal S^{-n+1}\widebar W_2^*\mathcal S\bigr)_{\leq0},&\!\!
            \widebar B_n&\!=\!-\bigl(\widebar W_1\mathcal S^{n-1}\widebar V_1^*\mathcal S^{-1}\bigr)_{\geq0}\!
-\!\bigl(\widebar V_1\mathcal S^{-n+1}\widebar W_1^*\mathcal S^{-1}\bigr)_{<0}, \\[1mm]
\,\,\widebar C_n&\!=\!\bigl(\widebar W_2\mathcal S^{n-1}\widebar V_2^*\mathcal S\bigr)_{>0}\!
+\!\bigl(\widebar V_2\mathcal S^{-n+1}\widebar W_2^*\mathcal S\bigr)_{\leq0} ,&\!\!
            \widebar D_n&\!=\!-\bigl(\widebar W_2\mathcal S^{n-1}\widebar V_1^*\mathcal S^{-1}\bigr)_{\geq0}\!
-\!\bigl(\widebar V_2\mathcal S^{-n+1}\widebar W_1^*\mathcal S^{-1}\bigr)_{<0}. 
        \end{aligned}
    \end{eqnarray*}
\end{lemma}
\begin{proof}
    Using~\eqref{explicitPiplus}, it suffices to compute both $L_{T_n}$ and $L_{\widebar T_n}$ in terms of the operators $V_i,W_i, \widebar V_i, \widebar W_i$. By definition of $\widebar U$ we have $U^{-1}=-J\widebar U^* J$, which gives  
    \begin{equation*}
U^{-1}
=
\begin{pmatrix}
\mathcal S^{-1}V_2^*\mathcal S
&
-\mathcal S^{-1}V_1^*\mathcal S^{-1}
\\[1mm]
-\mathcal S\widebar W_2^*\mathcal S
&
\mathcal S\widebar W_1^*\mathcal S^{-1}
\end{pmatrix}.
\end{equation*}
Hence the Lax operator $L_{T_n}=UT_nU^{-1}$ is explicitly written as
\begin{equation*}
\begin{aligned}
L_{T_n}&= 
\begin{pmatrix}
    W_1 & \widebar V_1 \\
    W_2 &  \widebar V_2 
\end{pmatrix}
\begin{pmatrix}
 \mathcal{S}^{-n} & 0 \\
0  & 0   
\end{pmatrix}
\begin{pmatrix}
\mathcal S^{-1}V_2^*\mathcal S
&
-\mathcal S^{-1}V_1^*\mathcal S^{-1}
\\[1mm]
-\mathcal S\widebar W_2^*\mathcal S
&
\mathcal S\widebar W_1^*\mathcal S^{-1}
\end{pmatrix}
\\
&= 
\begin{pmatrix}
 W_1 \mathcal{S}^{-n} \mathcal S^{-1}V_2^* \mathcal S
&
-W_1 \mathcal{S}^{-n} \mathcal S^{-1}V_1^*\mathcal S^{-1}
\\[1mm]
W_2 \mathcal{S}^{-n}\mathcal S^{-1}V_2^*\mathcal S
&
-W_2 \mathcal{S}^{-n}\mathcal S^{-1}V_1^*\mathcal S^{-1} 
\end{pmatrix}.
\end{aligned}
\end{equation*}
Similarly, we get 
\begin{equation*}
L_{\widebar T_n}=
\begin{pmatrix}
\widebar V_1 \mathcal{S}^{-n} \mathcal S\widebar W_2^* \mathcal S 
&
-\widebar V_1 \mathcal{S}^{-n} \mathcal S\widebar W_1^*\mathcal  S^{-1}
\\[1mm]
\widebar V_2 \mathcal{S}^{-n}\mathcal S\widebar W_2^*\mathcal S
&
-\widebar V_2 \mathcal{S}^{-n}\mathcal S\widebar W_1^*\mathcal S^{-1}
\end{pmatrix}.
\end{equation*}
We conclude the proof after applying $\Pi_+$, given that
\begin{equation*}
P_n= \Pi_+(L_{T_n}), \qquad \widebar P_n= \Pi_+(L_{\widebar T_n}).
\end{equation*}
\end{proof}
Since $P_n,\widebar P_n\in\mathfrak g_+$, they satisfy
\begin{equation*}
P_n^*=JP_nJ,
\qquad
\widebar P_n^*=J\widebar P_nJ,
\end{equation*}
and differentiating $\widebar U=-J(U^*)^{-1}J$ therefore gives
\begin{equation}
D_{T_n}\widebar U=P_n\widebar U+\widebar U \mathcal S^n E_{22},
\qquad
D_{\widebar T_n}\widebar U= \widebar P_n\widebar U-\widebar U \mathcal{S}^n E_{11}.
\label{ptbarflows}
\end{equation}
\begin{proposition}
\label{prop:takasaki_comparison}
Under the identification $\mathcal S=\,\textup{e}^{-\partial_s}$, the dressing data
$(U,\widebar U)$ and the dressing derivations $(D_{T_n},D_{\widebar T_n})$
defined above coincide with the continuous Pfaff--Toda hierarchy of
Takasaki~\cite{Taka}. 
\end{proposition}
\begin{proof}
Under $\mathcal S=\,\textup{e}^{-\partial_s}$, the dressing relations above and
evolution equations coincide term by term with Takasaki's formulas.
The correspondence is summarized in the following table. Notice that, because $\mathcal S=\,\textup{e}^{-\partial_s}$, the positive and negative
degree truncations in Takasaki's $\,\textup{e}^{\partial_s}$ convention are reversed
in our $\mathcal S$-convention.

\begin{center}\small 
\renewcommand{\arraystretch}{1.3}
\begin{tabular}{@{}lll@{}}
\toprule
\textbf{Data}
& \textbf{Present paper}
& \textbf{Takasaki \cite{Taka}} \\
\midrule

\textup{Shift operator}
& $\mathcal S$
& $\,\textup{e}^{-\partial_s}$ \\[1mm]
\hline 
$U,\widebar U$ and algebraic relation
& $\widebar{U}=-J(U^*)^{-1}J$
& $U^*= J_s\,\widebar{U}^{-1}J_s^{-1} $ \\[1mm]
$U\!=\!\!\begin{pmatrix}
W_1 & \widebar V_1\\
W_2 & \widebar V_2
\end{pmatrix}\!\!,~ \widebar U\!=\!\!\begin{pmatrix}
\widebar W_1&V_1\\
\widebar W_2&V_2
\end{pmatrix} $ & $J=\!\begin{pmatrix}
0&\mathcal S^{-1}\\
-\mathcal S&0
\end{pmatrix}$  & $J_s=\!\begin{pmatrix}
0&\,\textup{e}^{\partial_s}\\
-\,\textup{e}^{-\partial_s}&0
\end{pmatrix}$ \\[6mm]
\hline 
$t_n$-flows of $U,\widebar U$
& $D_{T_n}U=P_nU-UT_n$
& $\partial_{t_n} U = P_n U-(W_1 E_{11}+W_2E_{21})\,\textup{e}^{n\partial_s}$ \\[1mm]
& $D_{T_n}\widebar U=P_n\widebar U+\widebar U \mathcal S^n E_{22}$ 
& $\partial_{t_n} \widebar{U} = P_n \widebar{U} +(V_1 E_{21}+V_2E_{22})\,\textup{e}^{-n\partial_s}$ \\[1mm]

\hline 
$\widebar t_n$-flows of $U,\widebar U$
& $D_{\widebar T_n}U=\widebar P_nU-U \widebar T_n$
& $\partial_{\widebar{t}_n} U = \widebar{P}_n U +(\widebar{V}_1 E_{21}+\widebar{V}_2E_{22})\,\textup{e}^{n\partial_s}$ \\[1mm]
& $D_{\widebar T_n}\widebar U= \widebar P_n\widebar U-\widebar U \mathcal{S}^n E_{11}$
& $\partial_{\widebar{t}_n} \widebar{U} = \widebar{P}_n \widebar{U}-(\widebar{W}_1 E_{11}+\widebar{W}_2E_{21})\,\textup{e}^{-n\partial_s}$ \\[1mm]
\hline 
$P_n,\widebar{P}_n$ and their coefficients & Lemma~\ref{lem:P_barP} & \cite[Sec.\ 3.3, Theorem 2, Theorem 3]{Taka} \\
$\displaystyle
P_n\!=\!\!\begin{pmatrix}
A_n&B_n\\ C_n&D_n
\end{pmatrix}\!\!,~ \widebar P_n\!=\!\!\begin{pmatrix}
\widebar A_n&\widebar B_n\\
\widebar C_n&\widebar D_n
\end{pmatrix}$
& 
& \\[6mm]
\bottomrule
\end{tabular}
\end{center}
\end{proof}
\begin{remark}[The second discrete Pfaff--Toda direction] 
Takasaki's formulation of the Pfaff--Toda hierarchy contains a
second discrete variable $r$ \cite{Taka}. This direction can be recovered rationally from the dressing operator $U$, without introducing a second shift in the definition of the difference algebra~$\mathcal V$. We write:
\begin{equation}
    \widebar V_1=\eta+\eta_1 \mathcal S+\mathcal O(\mathcal S^2),
    \qquad
    \widebar V_2=\Delta+\delta_1 \mathcal S+\mathcal O(\mathcal S^2), \qquad \Gamma=
    \begin{pmatrix}
        \mathcal S^{-1}&0\\
        0&\mathcal S
    \end{pmatrix}.
\end{equation}
On the open set where $\eta$ is invertible, we define
\begin{equation}\label{eq:beta_gamma_a_b}
    \beta=\mathcal S^{-1}(\eta\Delta^{-1}),
    \qquad
    \gamma=\Delta\eta^{-1}, \qquad  a=
    \bigl(
        \beta \mathcal S^{-1}(\delta_1)-\mathcal S^{-1}(\eta_1)
    \bigr)\eta^{-1},
    \qquad
    b=\beta\gamma,
\end{equation}
and introduce the finite-band matrix difference operator
\begin{equation*}
    K(U)=
    \begin{pmatrix}
        \mathcal S^{-1}+a+b \mathcal S & -\beta \mathcal S^{-1}\\
        \gamma \mathcal S & 0
    \end{pmatrix}.
\end{equation*}
We may then define a rational transformation of the dressing variables by
\begin{equation}\label{eq:rational_transf_U}
    \mathcal T(U):=K(U)U\,\Gamma^{-1}.
\end{equation}
Since $\mathcal V$ is freely generated as an $\mathcal S$-difference algebra by the coefficients of $U$, the definition~\eqref{eq:rational_transf_U} determines an $\mathcal S$-difference homomorphism $\mathcal T:\mathcal V\to \operatorname{Frac}(\mathcal V)$, with $\mathcal T \mathcal S=\mathcal S \mathcal T$. Indeed, the upper right entry of~\eqref{eq:rational_transf_U} is
\begin{equation}\label{eq:T_V1}
    \mathcal T(\widebar V_1)=\Bigl(
        (\mathcal S^{-1}+a+b \, \mathcal S)\widebar V_1
        -\beta \, \mathcal S^{-1}\widebar V_2 \Bigr)\mathcal S^{-1}.
\end{equation}
The definitions of $\beta$ and $a$ in~\eqref{eq:beta_gamma_a_b} are precisely the conditions that the coefficients of $\mathcal S^{-1}$ and~$\mathcal S^0$ in the expression in parentheses in~\eqref{eq:T_V1}
vanish. Hence $T(\widebar V_1)=\mathcal O(1)$, while the other entries immediately give
\begin{equation*}
    \mathcal T(W_1)=1+\mathcal O(\mathcal S), \qquad 
    \mathcal T(W_2)=\mathcal O(\mathcal S^2), \qquad 
    \mathcal T(\widebar V_2)=\gamma \,\mathcal S\widebar V_1\mathcal S^{-1}=\mathcal T(\Delta)+\mathcal O(\mathcal S),  
\end{equation*}
with $\mathcal T(\Delta)=\gamma \, \mathcal S(\eta)=\Delta\,\frac{\mathcal S(\eta)}{\eta}$. Thus, $\mathcal T(U)$ has again the normalized dressing form~\eqref{eq:scalar-U}.
Moreover, the definitions of $\beta,\gamma$ and $b$ in~\eqref{eq:beta_gamma_a_b} imply
\begin{equation}\label{eq:KU_symplectic}
    K(U)J K(U)^*=J.
\end{equation}
Taking into account~\eqref{eq:scalar-barU}, $\Gamma^*=\Gamma^{-1}$ and $\Gamma^{-1}J=J\Gamma$, this gives $\mathcal T(\widebar U)=K(U)\,\widebar U\,\Gamma$.
Therefore,
\begin{equation}\label{eq:TU_Gamma}
    \mathcal T(U)\Gamma=K(U)U, \qquad \mathcal T(\widebar U)\Gamma^{-1}=K(U)\widebar U.
\end{equation}
With the convention $\mathcal S=\,\textup{e}^{-\partial_s}$, the equations in~\eqref{eq:TU_Gamma} are exactly Takasaki's discrete $r\mapsto r+1$ dressing equations~\cite[(3.8)]{Taka}, and the operator $K(U)$ has the finite-band form displayed in~\cite[(3.9)]{Taka}.

For this reason, the second discrete Pfaff--Toda direction need not be included as an independent shift in the original definition of $\mathcal V$. Indeed, it is realized as a rational transformation of the dressing variables determined by $U$~\eqref{eq:scalar-U} itself.
\end{remark}

\begin{remark}
One may ask whether the continuous Pfaff--Toda hierarchy admits a
formulation purely in terms of a pair of Lax operators, as in $2$D
Toda.  Such a pair can indeed be constructed from the opposite
dressing frames $U$ and $\widebar U$.  Unlike in $2$D Toda, however, the
two operators are not independent, because their dressing frames
satisfy
\begin{equation*}
\widebar U=-J(U^*)^{-1}J.
\end{equation*}
Any Lax-pair formulation must therefore retain an additional
compatibility relation encoding this common Pfaff dressing datum.
For this reason we work directly at the dressing level.
\end{remark}

\subsection{The odd extension of Pfaff--Toda}
\label{sec:odd_extension_PT}
The diagonal Pfaff--Toda flows do not exhaust the $N=1$ bare operator algebra. Here we introduce the additional off-diagonal directions $Q^+$ and $Q^-$, and we state the commutation relations between the corresponding dressing derivations and the diagonal ones. These relations follow directly from the Lie algebra structure of the bare operators and from the zero-curvature representation of Section~\ref{sec:dressing_zero_curvature}.

For $m\geq 0$, we define
\begin{equation*}
Q_m^+ := \mathcal S^{-m-1}E_{21},
\qquad
Q_m^- := \mathcal S^{-m}E_{12}.
\end{equation*}
The corresponding dressing derivations $D_{Q_m^+}$ and $D_{Q_m^-}$,
together with the diagonal Pfaff--Toda directions
$D_{T_n}$ and $D_{\widebar T_n}$, form a non-abelian algebra of evolutionary
derivations.

\begin{proposition}
For $m,n\geq 0$ and $k\geq 1$, we have 
\begin{eqnarray*}
    \begin{aligned}
[D_{T_k},D_{Q_n^+}]&=D_{Q_{k+n}^+},
&\qquad [D_{\widebar T_k},D_{Q_n^+}]&=D_{Q_{k+n}^+}, \\
[D_{T_k},D_{Q_n^-}]&=-D_{Q_{k+n}^-},
&\qquad [D_{\widebar T_k},D_{Q_n^-}]&=-D_{Q_{k+n}^-}, 
\end{aligned}
\end{eqnarray*}
and
\begin{eqnarray*}
    \begin{aligned}
[D_{Q_m^+},D_{Q_n^+}]&=[D_{Q_m^-},D_{Q_n^-}]=0, &\qquad
[D_{Q_m^+},D_{Q_n^-}]&=D_{T_{m+n+1}}+D_{\widebar T_{m+n+1}}.
\end{aligned}
\end{eqnarray*}
\end{proposition}

\begin{proof}
The bare operators satisfy the relations
\begin{eqnarray*}
\begin{aligned}
[T_k,Q_n^+]&=[\widebar T_k,Q_n^+]=-Q_{k+n}^+,
&\qquad [T_k,Q_n^-]&=[\widebar T_k,Q_n^-]=Q_{k+n}^-,\\[1mm]
[Q_m^+,Q_n^+]&=[Q_m^-,Q_n^-]=0, &\qquad 
[Q_m^+,Q_n^-]&=-T_{m+n+1}-\widebar T_{m+n+1}.
\end{aligned}
\end{eqnarray*}
The proposition follows from these relations by Proposition~\ref{prop:dressing_derivations_zero}.
\end{proof}

\subsection{The Adler--Pfaff hierarchy inside the odd extension}
\label{sec:odd_extension}
We now recombine the diagonal and off-diagonal directions into the powers of a single bare operator~$\Lambda$. To compare the resulting flows with the Adler--Pfaff hierarchy of Section~\ref{sec:adler-pfaff}, we apply a fixed shift gauge~$Q$ and a diagonal degree-zero gauge~$H$. The main result is that the symmetrically gauged operator $\mathcal L$ satisfies exactly the Adler--Pfaff Lax equations on the extended phase space of Section~\ref{subsec:extended-pfaff-space}.

We introduce the operator
\begin{equation}\label{fundope}
\Lambda=\begin{pmatrix}
0&1\\
\mathcal S^{-1}&0
\end{pmatrix},
\qquad
\Lambda^2=\mathcal S^{-1}I_2.
\end{equation}
The even powers of $\Lambda$ belong entirely to the diagonal
Pfaff--Toda sector, whereas the odd powers involve precisely the two
additional off-diagonal directions of Section~\ref{sec:odd_extension_PT}:
\begin{equation}
\Lambda^{2n}=T_n-\widebar T_n,
\qquad
\Lambda^{2n+1}=Q_n^++Q_n^-.
\label{lamrel}
\end{equation}
At the level of evolutionary derivations,~\eqref{lamrel} becomes
\begin{equation}
D_{\Lambda^{2n}}
=
D_{T_n}-D_{\widebar T_n},
~~ n\geq1,
\qquad
D_{\Lambda^{2n+1}}
=
D_{Q_n^+}+D_{Q_n^-},
~~ n\geq0.
\label{eq:pfaff-odd-flows}
\end{equation}
Since the powers of $\Lambda$ commute and $-D$ is a Lie algebra
morphism, the derivations $D_{\Lambda^k}$ pairwise commute. By applying Proposition~\ref{prop:dressing_derivations_zero} to the dressed operator $\Lambda$, 
\begin{equation}\label{eq:ungauged-pfaff-dressing}
    L_{\Lambda}=U\Lambda\, U^{-1} 
\end{equation}
the flows with respect to the dressing derivations are determined by
\begin{equation}\label{eq:ungauged-pfaff-flow}
D_{\Lambda^k}L_{\Lambda}=\bigl[\Pi_+(L_{\Lambda}^k),L_{\Lambda}\bigr].
\end{equation}
This is not yet the symmetric presentation used for the
Adler--Pfaff Lax operator in Section~\ref{sec:adler-pfaff}: there are two
discrepancies. First, the
Pfaff--Toda splitting of Section~\ref{sec:multicomponent-hierarchy} is defined using the
shift-dependent symplectic matrix $J$, whereas the Pfaff splitting of Section~\ref{sec:adler-pfaff} is written with the constant symplectic matrix $J_0$.
Second, even after removing these shifts, the degree-zero part of the corresponding auxiliary operator differs from the symmetric Adler--Pfaff normalization.  The first discrepancy is removed by a fixed shift gauge $Q$, while the second is corrected by a diagonal degree-zero gauge $H$. 
We define
\begin{equation}\label{eq:Vtilde}
\widetilde{\mathcal V}:=\mathcal V\bigl[\mathcal S^j(\rho^{\pm1})\mid j\in\mathbb Z\bigr],\qquad\rho^2=\mathcal S^{-1}(\Delta). 
\end{equation}
For all $x \in \mathfrak a$, we extend $D_x$ uniquely as an evolutionary derivation of $\tilde{\mathcal V}$ by
\begin{equation*}
D_x(\rho)= \frac{\mathcal S^{-1}(D_x\Delta)}{2\rho},
\end{equation*}
and we further extend it coefficient-wise to $M_2(\widetilde{\mathcal V}((\mathcal S )))$. By uniqueness of this extension, the Lie relations among the
derivations $D_x$ are preserved on $\widetilde{\mathcal V}$.
We introduce
\begin{equation}
Q=\begin{pmatrix}
1&0\\
0&\mathcal S^{-1}
\end{pmatrix},
\qquad
H=\begin{pmatrix}
\rho&0\\
0&\rho^{-1}
\end{pmatrix},
\qquad
G:=HQ=\begin{pmatrix}
\rho&0\\
0&\rho^{-1}\mathcal S^{-1}
\end{pmatrix}.
\label{QHG}
\end{equation}
Let $\mathcal{L}$ be the symmetrically gauged operator
\begin{equation}\label{eq:symmetric-pfaff-lax}
\mathcal L:=GL_{\Lambda}G^{-1}.
\end{equation}
By the asymptotic form of $U$ in~\eqref{eq:scalar-U}, $\mathcal{L}$ has the form
\begin{equation*}
\mathcal L=
\begin{pmatrix}
\mathcal O(\mathcal S^{-1}) & \mathcal O(1)\\
\mathcal O(\mathcal S^{-2}) & \mathcal O(\mathcal S^{-1})
\end{pmatrix},
\end{equation*}
where the coefficient of $\mathcal S^{-2}$ in the $(2,1)$-entry is invertible. Thus $\mathcal L$ has precisely the support of the extended Adler--Pfaff phase space $\mathcal M_{\rm ext}$ introduced in Section~\ref{subsec:extended-pfaff-space}.
Differentiating~\eqref{eq:symmetric-pfaff-lax} along $D_{\Lambda^k}$ gives
\begin{equation}
D_{\Lambda^k}\mathcal L
=
[B_k,\mathcal L],
\label{eq:gauged-pfaff-lax}
\end{equation}
where
\begin{equation}
B_k=(D_{\Lambda^k}G)G^{-1}+G\,\Pi_+(L_{\Lambda}^k)G^{-1}.
\label{eq:gauged-pfaff-auxiliary}
\end{equation}
The symmetric gauge is chosen so that $B_k$ coincides with the Adler--Pfaff auxiliary operator.
\begin{lemma}
For every $k\geq1$,
\begin{equation}
B_k
=
P_+^{\mathrm{Pf}}(\mathcal L^k),
\label{eq:pfaff-gauge-identity}
\end{equation}
where $P_\pm^{\rm Pf}$ denote the coefficient-wise extensions of the Pfaff projections of Section~\ref{sec:adler-pfaff} to $M_2\bigl(\widetilde{\mathcal V}((\mathcal S))\bigr)$.
\end{lemma}

\begin{proof}
We treat the two factors $Q$ and $H$ in~\eqref{QHG} separately. First, we conjugate $L_\Lambda$ and transport the Pfaff--Toda splitting by~$Q$. We introduce $\widehat L:=QL_{\Lambda}Q^{-1}$
and define the transported projection as
\begin{equation*}
\widehat\Pi_+(Y):=Q\,\Pi_+(Q^{-1}YQ)Q^{-1}.
\end{equation*}
The kernel of $\widehat\Pi_+$ is
\begin{equation}\label{eq:ker_hatPi}
{\rm ker}\, \widehat\Pi_+ =Q\,g_-Q^{-1}=\begin{pmatrix}
\mathcal O(\mathcal S)&\mathcal O(\mathcal S)\\
\mathcal O(\mathcal S)&\mathcal O(1)
\end{pmatrix},
\end{equation}
whereas the image of $\widehat\Pi_+$ is
\begin{equation}\label{eq:im_hatPi}
{\rm im}\, \widehat\Pi_+=Q\,g_+Q^{-1}=\left\{
Y\in M_2(\widetilde{\mathcal V}[\mathcal S,\mathcal S^{-1}])\ \middle|\
Y^*=J_0YJ_0
\right\}.
\end{equation}
Here we work over the extension $\widetilde{\mathcal V}$ defined in~\eqref{eq:Vtilde} and use
$QJQ^{-1}=J_0$. Explicitly, the action of~$\widehat{\Pi}_+$ on an element of
$M_2(\widetilde{\mathcal V}((\mathcal S)))$ is
\begin{equation}\label{eq:q-transported-projection}
\widehat\Pi_+
\begin{pmatrix}
A&B\\
C&D
\end{pmatrix}=\begin{pmatrix}
A_{\leq0}-(D^*)_{>0}&B_{\leq0}+(B^*)_{>0}
\\[1mm]
C_{<0}+(C^*)_{\geq0}&D_{<0}-(A^*)_{\geq0}
\end{pmatrix}.
\end{equation}
The second gauge $H$ has degree zero and satisfies
\begin{equation*}
H^*J_0H=J_0.
\end{equation*}
Thus, the conjugation by $H$ preserves both the image and the kernel of $\widehat\Pi_+$, and therefore commutes with the projection.
Since $Q$ is independent of the hierarchy variables and $\mathcal L=H\widehat L H^{-1}$, equation~\eqref{eq:gauged-pfaff-auxiliary} becomes
\begin{equation}\label{eq:Bk-after-QH}
B_k=(D_{\Lambda^k}H)H^{-1}+\widehat\Pi_+(\mathcal L^k).
\end{equation}
It remains to determine the logarithmic derivative of $H$.
Let $\widehat U:=Q U Q^{-1}$ be the conjugation of $U$ in~\eqref{eq:scalar-U} by $Q$ in~\eqref{QHG}, then
\begin{equation*}
\widehat U=
\begin{pmatrix}
1&0\\
0&\rho^2
\end{pmatrix}
+\mathcal O(\mathcal S).
\end{equation*}
On the other hand, since
\begin{equation*}
D_{\Lambda^k} U
=-\Pi_-(L^k)U
=
\bigl(\Pi_+(L^k)-L^k\bigr)U,
\end{equation*}
and $D_{\Lambda^k}(Q)=0$, for $\widehat U$ we obtain
\begin{equation*}
D_{\Lambda^k}\widehat U
=
\bigl(
\widehat\Pi_+(\widehat L^k)-\widehat L^k
\bigr)\widehat U.
\end{equation*}
The operator in parentheses belongs to $Q\,g_-Q^{-1}$, whose only
possible degree-zero entry is the $(2,2)$-entry.  Taking the leading
coefficient of that entry and using
\eqref{eq:q-transported-projection} gives
\begin{equation*}
D_{\Lambda^k}(\rho^2)
=
-\operatorname{tr}\bigl((\widehat L^k)_0\bigr)\rho^2.
\end{equation*}
Since $H$ has degree zero we can write $(\mathcal L^k)_0=H(\widehat L^k)_0H^{-1}$,
and therefore $\text{tr}\bigl((\widehat L^k)_0\bigr)=\text{tr}\bigl((\mathcal L^k)_0\bigr)$.
Hence
\begin{equation}
\frac{D_{\Lambda^k}\rho}{\rho}
=
-\frac12
\operatorname{tr}\bigl((\mathcal L^k)_0\bigr),
\label{eq:rho-pfaff-evolution}
\end{equation}
and
\begin{equation*}
(D_{\Lambda^k} H)H^{-1}
=
\frac12
\operatorname{tr}\bigl((\mathcal L^k)_0\bigr)
\begin{pmatrix}
-1&0\\
0&1
\end{pmatrix}.
\end{equation*}
Substituting this into~\eqref{eq:Bk-after-QH}, and writing
\begin{equation*}
\mathcal L^k=
\begin{pmatrix}
E&F\\
G&K
\end{pmatrix},
\end{equation*}
we get
\begin{equation*}
B_k=
\begin{pmatrix}
E_{<0}
+\frac12(E-K)_0
-(K^*)_{>0}
&
F_{\leq0}+(F^*)_{>0}
\\[1mm]
G_{<0}+(G^*)_{\geq0}
&
K_{<0}
+\frac12(K-E)_0
-(\,E^*)_{>0}
\end{pmatrix}.
\end{equation*}
Comparing this with the definition~\eqref{plussplittngpfaffpseudo} of the Pfaff projection,
we obtain
\begin{equation*}
B_k=P_+^{\rm Pf}(\mathcal L^k),
\end{equation*}
which proves the lemma.
\end{proof}

\begin{proposition}\label{prop:Pfaff_flows_symm_gauged}
The symmetrically gauged operator $\mathcal L$ satisfies the
Adler--Pfaff Lax equations
\begin{equation}\label{eq:Lax_eq_symm_gauged}
D_{\Lambda^k}\mathcal L=\bigl[P_+^{\rm Pf}(\mathcal L^k),\mathcal L\bigr]=\bigl[-P_-^{\rm Pf}(\mathcal L^k),\mathcal L\bigr],
\qquad k\geq1.
\end{equation}
In particular, the commuting subhierarchy generated by the powers
of $\Lambda$ in the odd Pfaff--Toda extension is carried by the
symmetric gauge to the Adler--Pfaff Lax equations on the extended
phase space of Section~\ref{subsec:extended-pfaff-space}.
\end{proposition}

\begin{proof}
The first equality follows from
\eqref{eq:gauged-pfaff-lax} and
\eqref{eq:pfaff-gauge-identity}. Since
\begin{equation*}
\mathcal L^k
=
P_+^{\rm Pf}(\mathcal L^k)
+
P_-^{\rm Pf}(\mathcal L^k)
\end{equation*}
and $[\mathcal L^k,\mathcal L]=0$, we also have
\begin{equation*}
\bigl[P_+^{\rm Pf}(\mathcal L^k),\mathcal L\bigr]
=
\bigl[-P_-^{\rm Pf}(\mathcal L^k),\mathcal L\bigr].
\end{equation*}
\end{proof}

\section{Toda reductions}
\label{sec:toda-reductions}

We study suitable reductions of the Matrix Pfaff--Toda hierarchy obtained by restricting both the dressing operator and the bare algebra. We first impose a block-diagonal condition for arbitrary $N$ and show that the universal zero-curvature representation descends to this reduction.
Then, further restricting to the diagonal Laurent subalgebra and changing the dressing presentation yields the multicomponent $2D$ Toda hierarchy.

For $N=1$, we place the resulting scalar Toda hierarchy in the same symmetric gauge as the Pfaff hierarchy of Section~\ref{sec:odd_extension}. The anti-diagonal Toda flows then give the even Pfaff hierarchy, while the Krichever--Zabrodin dressing constraint as in~\cite{KriZab} yields C--Toda.

\subsection{Block-diagonal dressing reduction}
Recall that the universal dressing operator of the
Matrix Pfaff--Toda hierarchy has the form~\eqref{eq:dressing-U}
\begin{equation*}
U_N=
\begin{pmatrix}
W_1 & \widebar V_1\\
W_2 & \widebar V_2
\end{pmatrix}
=
\begin{pmatrix}
I_N+a_0+\mathcal O(\mathcal S) & b_{-1}\mathcal S^{-1}+\mathcal O(1)\\
c_1 \mathcal S+\mathcal O(\mathcal S^2) & \Delta+\mathcal O(\mathcal S)
\end{pmatrix},
\end{equation*}
where $a_0,b_{-1},c_1\in\mathfrak n_-,\Delta\in \mathfrak b_-$
and $\Delta$ is invertible over $\mathcal V_N$. We introduce the diagonal and anti-diagonal versions of $U_N$ as 
\begin{equation}\label{eq:UdiagUadiag}
U_N^{\mathrm{diag}}=\begin{pmatrix}
W_1&0\\
0&\widebar V_2
\end{pmatrix},
\qquad
U_N^{\mathrm{adiag}}=
\begin{pmatrix}
0&\widebar V_1\\
W_2&0
\end{pmatrix}.
\end{equation}
Let $\mathcal I_{\mathrm{adiag}}\subset\mathcal V_N$ be the difference ideal
generated by the coefficients of $W_2$ and $\widebar V_1$.  Since the
coefficients of $U_N$ are free before localization, and the localization
involves only $\Delta$ and its shifts, the quotient is identified
with the difference algebra generated by the two diagonal blocks:
\begin{equation}\label{eq:VN_diag}
\mathcal V_N/\mathcal I_{\mathrm{adiag}}\simeq \mathcal V_N^{\mathrm{diag}}.
\end{equation}
We denote the corresponding specialization by
\begin{equation}\label{eq:block-specialization}
\pi:\mathcal V_N\longrightarrow\mathcal V_N^{\mathrm{diag}},
\qquad
W_2,\widebar V_1\longmapsto0,
\end{equation}
and extend $\pi$ coefficient-wise to matrix pseudodifference operators.
We restrict the bare algebra to the block-diagonal associative
subalgebra
\begin{equation*}
\mathfrak a_N^{\mathrm{diag}}
=
\left\{
\begin{pmatrix}
H&0\\
0&K
\end{pmatrix}
\ \middle|\
H,K\in M_N(\mathbb C[\mathcal S,\mathcal S^{-1}])
\right\}
\subset\mathfrak a_N.
\label{eq:block-bare-algebra}
\end{equation*}
For $x\in\mathfrak a_N^{\mathrm{diag}}$, recall that
\begin{equation*}
D_x(U_N)=-P_x^-U_N,
\qquad
P_x^-=\Pi_-(L_x),
\qquad
L_x=U_NxU_N^{-1}.
\end{equation*}
The following lemma shows that the block-diagonal specialization $U_N^{\rm adiag}=0$ is
dynamically consistent.

\begin{lemma}\label{lem:diagonal}
For every $x\in\mathfrak a_N^{\mathrm{diag}}$, the ideal
$\mathcal I_{\mathrm{adiag}}$ is $D_x$-invariant.  Moreover,
\begin{equation}
\pi\bigl(D_x(U_N^{\mathrm{diag}})\bigr)
=
-\Pi_-\bigl(L_x^{\mathrm{diag}}\bigr)
U_N^{\mathrm{diag}},
\label{eq:diag-reduced-dressing}
\end{equation}
where
\begin{equation*}
L_x^{\mathrm{diag}}
:=
U_N^{\mathrm{diag}}x
\bigl(U_N^{\mathrm{diag}}\bigr)^{-1}.
\end{equation*}
\end{lemma}

\begin{proof}
Since $x$ is block-diagonal,
\begin{equation*}
\pi(L_x)=\pi(U_N)x\,\pi(U_N^{-1})
=
U_N^{\mathrm{diag}}x
\bigl(U_N^{\mathrm{diag}}\bigr)^{-1}
=
L_x^{\mathrm{diag}}.
\end{equation*}
Here we used that $\pi$ is an algebra homomorphism and $U_N$ is
invertible, so that
\begin{equation*}
\pi(U_N^{-1})=\pi(U_N)^{-1}.
\end{equation*}
Moreover, $\pi$ commutes with the projections $\Pi_\pm$.  Indeed, their
definition involves shifts, Laurent coefficient projections, the
formal adjoint, transpose, and triangular projections of the matrix
coefficients, and all these operations commute with the specialization $\pi$.
Therefore
\begin{equation*}
\pi(P_x^-)
=
\pi\bigl(\Pi_-(L_x)\bigr)
=
\Pi_-\bigl(\pi(L_x)\bigr)
=
\Pi_-\bigl(L_x^{\mathrm{diag}}\bigr).
\end{equation*}
Applying $\pi$ to the dressing equation gives
\begin{equation*}
\pi\bigl(D_x(U_N)\bigr)
=
-\Pi_-\bigl(L_x^{\mathrm{diag}}\bigr)
U_N^{\mathrm{diag}}.
\end{equation*}
The right-hand side is block-diagonal, since both
$L_x^{\mathrm{diag}}$ and its projections are block-diagonal.  On the
other hand,
\begin{equation*}
\pi\bigl(D_x(U_N)\bigr)
=
\pi\bigl(D_x(U_N^{\mathrm{diag}})\bigr)
+
\pi\bigl(D_x(U_N^{\mathrm{adiag}})\bigr),
\end{equation*}
where the first term is block-diagonal and the second block-antidiagonal.
Comparing the two parts gives
$\pi\bigl(D_x(U_N^{\mathrm{adiag}})\bigr)=0$ and
\eqref{eq:diag-reduced-dressing}. Since $\mathcal I_{\mathrm{adiag}}$ is generated as a difference ideal by the
coefficients of $W_2$ and $\widebar V_1$, we conclude that
\begin{equation*}
D_x(\mathcal I_{\mathrm{adiag}})\subset \mathcal I_{\mathrm{adiag}}.
\end{equation*}
\end{proof}
Therefore, for all $x\in\mathfrak a_N^{\mathrm{diag}}$ the evolutionary derivation $D_x$ descends
to the quotient $\mathcal V_N^{\mathrm{diag}}$~\eqref{eq:VN_diag} as
\begin{equation*}
D_x^{\mathrm{diag}}\bigl(\pi(f)\bigr)
:=
\pi\bigl(D_x(f)\bigr),
\qquad
f\in\mathcal V_N.
\label{eq:descended-derivation}
\end{equation*}
We henceforth omit the superscript and denote the descended
derivation again by $D_x$.
The universal Lax representation descends at the same time.  For
$x\in\mathfrak a_N^{\mathrm{diag}}$, we introduce
\begin{equation*}
L_x^{\mathrm{diag}}
=
U_N^{\mathrm{diag}}x
\bigl(U_N^{\mathrm{diag}}\bigr)^{-1},
\qquad
P_x^{\mathrm{diag}}
=
-\Pi_+\bigl(L_x^{\mathrm{diag}}\bigr).
\end{equation*}

\begin{proposition}
\label{prop:diagonal_dressing_derivative}
The maps
\begin{equation*}
x\longmapsto L_x^{\rm diag}, \qquad x\longmapsto -D_x, \qquad x\longmapsto P_x^{\rm diag}
\end{equation*}
define a zero-curvature Lax representation of
$\mathfrak a_N^{\rm diag}$ on $\mathcal V_N^{\rm diag}$.  In particular, for all
$x,y\in\mathfrak a_N^{\mathrm{diag}}$,
\begin{equation}
D_x\bigl(L_y^{\mathrm{diag}}\bigr)
=
\bigl[
\Pi_+(L_x^{\mathrm{diag}}),
L_y^{\mathrm{diag}}
\bigr]
-
L_{[x,y]}^{\mathrm{diag}},
\label{eq:diag-lax-representation}
\end{equation}
and
\begin{equation}
[D_x,D_y]=-D_{[x,y]}.
\label{eq:diag-derivation-bracket}
\end{equation}
Moreover,
\begin{equation}
D_x(P_y^{\mathrm{diag}})
-
D_y(P_x^{\mathrm{diag}})
+
[P_x^{\mathrm{diag}},P_y^{\mathrm{diag}}]
+
P_{[x,y]}^{\mathrm{diag}}=0.
\label{eq:diag-zero-curvature}
\end{equation}
\end{proposition}

\begin{proof}
By Proposition~\ref{prop:dressing_derivations_zero}, the maps
$(L,-D,-P^+)$ form a zero-curvature Lax representation of
$\mathfrak a_N$.
By the preceding lemma, the ideal $\mathcal I_{\mathrm{adiag}}$ is preserved by all $D_x$, $x\in\mathfrak a_N^{\mathrm{diag}}$, while $\pi$ commutes with~$\Pi_\pm$.  Applying $\pi$ to the universal Lax, Lie-algebra, and
the zero-curvature identities therefore gives the relations~\eqref{eq:diag-lax-representation}-~\eqref{eq:diag-zero-curvature}.
\end{proof}

To obtain a commuting hierarchy, we select the following family of bare operators
\begin{equation*}
\mathfrak d_N
=
\{
\operatorname{diag}(c_1,\ldots,c_N)
\mid
c_1,\ldots,c_N\in\mathbb C
\}
\subset M_N(\mathbb C),
\end{equation*}
and define
\begin{equation}\label{eq:cartan-bare-algebra}
\mathfrak h_N=\left\{
\begin{pmatrix}
H&0\\
0&K
\end{pmatrix}
\ \middle|\
H,K\in \mathcal{S}^{-1} \mathfrak{d}_N[\mathcal S^{-1}]
\right\}
\subset
\mathfrak a_N^{\mathrm{diag}}.
\end{equation}
Notice that the word ``diagonal'' now refers to the entries inside each
$N\times N$ block, in contrast with the block-diagonal reduction above.
The Lie algebra $\mathfrak h_N$ is abelian, and hence
\begin{equation*}
[D_x,D_y]=0,
\qquad
x,y\in\mathfrak h_N.
\end{equation*}

\subsection{The multicomponent 2D Toda hierarchy}
\label{subsec:multicomponent-Toda}
Restricting the block-diagonal bare algebra
$\mathfrak a_N^{\rm diag}$ to the abelian Laurent subalgebra
$\mathfrak h_N$ yields $2N$ families of pairwise commuting flows
on the two diagonal dressing blocks. These flows are closely related
to the multicomponent 2D Toda hierarchy~\cite{UT84,MMAF09,TZ25}, but this identification is not yet expressed in the standard
dressing presentation. Indeed, both dressing blocks are expanded in the same
$\mathcal S$-adic direction, while the Pfaff--Toda projection has a
nontrivial boundary contribution at degree zero. We therefore twist the
second dressing block and normalize the first one to obtain the
standard multicomponent 2D Toda dressing representation. 

We introduce the map $\iota\colon  M_N(\mathcal{V}^{\rm diag}_N((\mathcal{S}))) \rightarrow M_N(\mathcal{V}^{\rm diag}_N((\mathcal{S}^{-1}))) $
\begin{equation}\label{eq:iota}
\iota(D) :=-\mathcal S^{-1}D^* \mathcal S. 
\end{equation}
Since the formal adjoint is anti-multiplicative, $-\iota$ is an algebra anti-homomorphism and 
$\iota$ is a Lie algebra homomorphism.
Denote by $w_0$ the leading term of the matrix $W_1$. By~\eqref{eq:dressing-U}, $w_0-I_N \in \mathfrak{n}_{-}$, so $w_0$ is unit lower triangular, hence invertible. We define the pair of matrix pseudodifference operators 
\begin{equation}\label{eq:Toda-dressing}
W:=w_0^{-1}W_1,
\qquad
\widebar W:=-w_0^{-1}\,\iota(\widebar V_2^{-1}).
\end{equation}
They have the opposite asymptotic forms 
\begin{equation}\label{eq:WbarW_asymptotics}
W=I_N+\mathcal O(\mathcal S) \in M_N(\mathcal{V}_N^{\rm diag}((\mathcal{S}))),
\qquad
\widebar W=\widebar w_0+\mathcal O(\mathcal S^{-1}) \in M_N(\mathcal{V}_N^{\rm diag}((\mathcal{S}^{-1}))),
\end{equation}
where $\widebar w_0$ is invertible.
As in Section~\ref{sec:SZ-realization}, consider the bare matrices
\begin{equation}
T_{\alpha,n}=\mathcal S^{-n}
\begin{pmatrix}
E_{\alpha \alpha}&0\\
0&0
\end{pmatrix},
\qquad
\widebar T_{\alpha,n}
=
-\mathcal S^{-n}
\begin{pmatrix}
0&0\\
0&E_{\alpha \alpha}
\end{pmatrix},
\qquad
n\geq1, \qquad \alpha=1,\ldots,N.
\label{eq:Toda-bare-directions}
\end{equation}
By Proposition~\ref{prop:diagonal_dressing_derivative}, the corresponding evolutionary derivations
commute:
\begin{equation*}
[D_x,D_y]=0,
\qquad
x,y\in
\{T_{\alpha,n},\widebar T_{\alpha,n}\}.
\end{equation*}
\begin{lemma}
\label{lem:Toda-Sato}
For every $\alpha=1,\ldots,N$ and $n\geq1$, the operators $W$ and
$\widebar W$ in~\eqref{eq:Toda-dressing} satisfy
\begin{align}
D_{T_{\alpha,n}}W
&=
B_{\alpha,n}W-W\mathcal S^{-n}E_{\alpha\alpha},
&
D_{T_{\alpha,n}}\widebar W
&=
B_{\alpha,n}\widebar W,
\label{eq:Toda-Sato-t}\\[1mm]
D_{\widebar T_{\alpha,n}}W
&=
\widebar B_{\alpha,n}W,
&
D_{\widebar T_{\alpha,n}}\widebar W
&=
\widebar B_{\alpha,n}\widebar W-\widebar W \mathcal S^nE_{\alpha\alpha},
\label{eq:Toda-Sato-bart}
\end{align}
where  
\begin{equation*} B_{\alpha,n}:=(W\mathcal{S}^{-n}E_{\alpha\alpha}W^{-1})_{\leq0},
\qquad
\widebar B_{\alpha,n}:=(\widebar W E_{\alpha\alpha}\mathcal{S}^n \widebar W^{-1})_{>0}. 
\end{equation*}
\end{lemma}

\begin{proof}
For all $x=\textup{diag}(H,K) \in \mathfrak a_N^{\rm diag}$, the block-diagonal matrix $\Pi_+(L_x^{\rm diag})$ is by definition in $\mathfrak{g}_N^+,$ hence it can be written in the form 
\begin{equation*}
\Pi_+(L_x^{\rm diag})= (B_x,F_x) , \qquad B_x=\iota (F_x).
\end{equation*}
Introducing the elements 
\begin{equation*}
    A_x:=W_1HW_1^{-1},  \qquad
        \widebar{A}_x := \iota(\widebar{V}_2K\widebar{V}_2^{-1}),
\end{equation*}
we can rewrite $B_x$ in its decomposition, where $(\,\cdot\,)_\textup{u}$ denotes the upper triangular part (including the diagonal), and $(\,\cdot\,)_\textup{sl}$ the strictly lower triangular part:
\begin{equation}\label{eq:Bx_2dToda}
        B_x := (A_x)_{<0}+((A_x)_0)_{\textup{u}}+
((\widebar{A}_x)_0)_{\textup{sl}}+ (\widebar{A}_x)_{>0}.  
\end{equation}
Applying Definition~\ref{def:dressing_derivative} to the reduced dressing equations gives
\begin{equation*}
        D_x W_1 = B_xW_1-W_1H, \qquad 
        D_x \widebar{V}_2 =F_x\widebar{V}_2-\widebar{V}_2K.
\end{equation*}
Since $-\iota$ in~\eqref{eq:iota} is anti-multiplicative and $D_x\circ\iota=\iota\circ D_x$, we deduce that 
\begin{equation*}
        D_xW_1 = B_xW_1-W_1H, \qquad 
        D_x \,\iota (\widebar{V}_2^{-1}) =B_x\, \iota (\widebar{V}_2^{-1})-\iota (\widebar{V}_2^{-1}) \, \iota(K).
\end{equation*}
Rewriting in terms of $W$ and $\widebar{W}$ in~\eqref{eq:Toda-dressing} yields
\begin{align}
     \label{keyeqW}
        D_xW &= \widehat{B}_xW-WH \\
    \label{keyeqWbar}
        D_x \widebar{W} &=\widehat{B}_x\widebar{W}-\widebar{W} \iota(K),
\end{align}
where $\widehat{B}_x:=w_0^{-1}B_xw_0-w_0^{-1}D_x(w_0)$. We conclude the proof by considering separately the cases $K=0$ and $H=0$.
\\[2mm]
\underline{Case $1$: $x= T_{\alpha,n}$, $H=\mathcal{S}^{-n} E_{\alpha \alpha}$, $K=0$.}
\\
We have $\widebar{A}_x=0$ in~\eqref{eq:Bx_2dToda}, hence $B_x$, and a fortiori $\widehat{B}_x$, contains only non-positive powers of the shift~$\mathcal{S}$. Given the asymptotics of $W$ in~\eqref{eq:WbarW_asymptotics}, $D_x(W)W^{-1}$ is strictly positive. Multiplying~\eqref{keyeqW} on the right by $W^{-1}$, we obtain
\begin{equation*}
\widehat{B}_{T_{\alpha,n}}=(W\mathcal{S}^{-n}E_{\alpha \alpha}W^{-1})_{\leq 0}=B_{\alpha,n}.
\end{equation*}
\\[2mm]
\underline{Case $2$: $x=\widebar{T}_{\alpha,n}$, $H=0$, $K=-\mathcal{S}^{-n}E_{\alpha \alpha}$.}
\\
Here $\iota(K)=\mathcal{S}^nE_{\alpha \alpha}$. Multiplying~\eqref{keyeqW} on the right by $W^{-1}$ shows that $\widehat{B}_x$ is purely positive. Since $\widebar{W}$ is a Laurent series in $\mathcal{S}^{-1}$, the product $D_x(\widebar{W})\widebar{W}^{-1}$ contains no positive power of the shift. Hence multiplying~\eqref{keyeqWbar} on the right by $\widebar{W}^{-1}$ gives
\begin{equation*}
\widehat{B}_{\widebar{T}_{\alpha,n}}=(\widebar{W}\mathcal{S}^{n}E_{\alpha \alpha}\widebar{W}^{-1})_{>0}=\widebar B_{\alpha,n}.
\end{equation*}
\end{proof}
We introduce the following pair of matrix pseudodifference operators 
\begin{equation}
    L_{\rm Toda}=W\mathcal{S}^{-1}W^{-1} \in M_N(\mathcal{V}_N^{\rm diag}((\mathcal{S}))), \qquad  \widebar L_{\rm Toda}=\widebar W \mathcal{S} \, \widebar W^{-1} \in M_N(\mathcal{V}_N^{\rm diag}((\mathcal{S}^{-1}))),
    \label{Todapair}
\end{equation}
and define the projectors 
\begin{equation}
    C_{\alpha}=WE_{\alpha\alpha}W^{-1}, \qquad  \widebar C_{\alpha}=\widebar W E_{\alpha\alpha} \widebar W^{-1}, \qquad \alpha=1,\dots,N.
\end{equation}
\begin{proposition}
\label{prop:multicomponent-Toda}
The data defined above form the continuous $N$-component $2$D Toda
hierarchy.  More precisely:
\begin{enumerate}
\item
The Lax operators and dressed projectors have the asymptotic forms
\begin{align*}
L_{\rm Toda}&= \mathcal S^{-1}+u_0+u_1 \mathcal S+\cdots,& 
C_{\alpha}&=
E_{\alpha\alpha}+c_{\alpha,1} \mathcal S+\cdots,
\\[1mm]
\widebar L_{\rm Toda}&=
\widebar u_{-1} \mathcal S+\widebar u_0+\widebar u_1 \mathcal S^{-1}+\cdots, & 
\widebar C_{\alpha}&=\widebar c_{\alpha,0}+\widebar c_{\alpha,1} \mathcal S^{-1}+\cdots,
\end{align*}
where $\widebar u_{-1}$ is invertible.  
\item
For every $n \geq 1$, $\alpha=1,\dots,N$ and $A\in
\{L_{\rm Toda},\widebar L_{\rm Toda},C_1,\ldots,C_N,\widebar C_1,\ldots,\widebar C_N\}$,
one has
\begin{equation}
D_{T_{\alpha,n}}A=\bigl[(L_{\rm Toda}^nC_{\alpha})_{\leq 0}\,,A\bigr],
\qquad
D_{\widebar T_{\alpha,n}}A=\bigl[(\widebar{L}_{\rm Toda}^n \widebar C_{\alpha})_{>0}\,,A\bigr].
\label{eq:multicomponent-Toda-Lax}
\end{equation}

\item
The flows pairwise commute.
\end{enumerate}
\end{proposition}
\begin{proof}
The first assertion follows immediately from the asymptotics~\eqref{eq:WbarW_asymptotics}
\begin{equation*}
W=I_N+\mathcal O( \mathcal S),
\qquad
\widebar W=\widebar w_0+\mathcal O(\mathcal S^{-1}),
\qquad
\widebar w_0\in GL_N(\mathcal{V}_N^{\rm diag}).
\end{equation*}
The second assertion follows from Lemma~\ref{lem:Toda-Sato} by differentiating the definitions of
$L_{\mathrm{Toda}}$, $\widebar L_{\mathrm{Toda}}$, $C_{\alpha}$, and $\widebar C_{\alpha}$. Indeed, from the definition of the pair $(L_{\rm Toda} , \widebar L_{\rm Toda})$ in~\eqref{Todapair}, for any  $n \ge 1$ and $\alpha=1,\dots,N$ we have
\begin{equation*}
B_{\alpha,n}=(L_{\rm Toda}^nC_{\alpha})_{\leq 0}, \qquad \widebar B_{\alpha,n}= (\widebar L_{\rm Toda}^n \widebar C_{\alpha})_{>0}.
\end{equation*}
Lastly, the third assertion follows from the fact that $[D_x,D_y]=0$ for any pair $(x,y) \in \mathfrak h_N \times \mathfrak h_N$, with $\mathfrak h_N$ in~\eqref{eq:cartan-bare-algebra}. 
\end{proof}

\subsection{\texorpdfstring{Even Pfaff as an anti-diagonal scalar 2D Toda subhierarchy}{2dToda}}

Specializing to $N=1$, we place the scalar Toda and Pfaff hierarchies in the same symmetric gauge. In this gauge, the even powers of the fundamental Pfaff bare operator correspond to anti-diagonal combinations of the two Toda families. Below we show that these flows coincide with the even Pfaff hierarchy.

For $N=1$ one has $\mathfrak n_-=0$, hence the asymptotic for the pair $(W,\widebar W)$ in~\eqref{eq:Toda-dressing} is 
\begin{equation}\label{eq:WbarW_n1_definition}
W=W_1=1+\mathcal O(\mathcal S),\qquad \widebar W=-\iota(\widebar V_2^{-1})=\mathcal S^{-1}(\widebar V_2^{-1})^*\mathcal S.
\end{equation}
From now on, we write $T_n:=T_{1,n}$ and $\widebar T_n:=\widebar T_{1,n}$ for all $n \ge  1$. By Proposition~\ref{prop:multicomponent-Toda}, the corresponding flows form the scalar $2D$ Toda
hierarchy, with Lax operators
\begin{equation*}
L_{\rm Toda}=W\mathcal S^{-1}W^{-1},
\qquad
\widebar L_{\rm Toda}=\widebar W\mathcal S\, \widebar W^{-1}, 
\end{equation*}
and by the definition of $\widebar W$ in~\eqref{eq:Toda-dressing}, we obtain the asymptotics  
\begin{equation*}\widebar W
=\widebar w_0+\mathcal O(\mathcal S^{-1}),
\qquad \widebar w_0=\bigl(\mathcal S^{-1}(\Delta)\bigr)^{-1}.\end{equation*}
Then, following the symmetric normalization of Krichever--Zabrodin~\cite{KriZab}, we pass to the difference algebra extension
\begin{equation}\label{eq:Vdiag_tilde}
\widetilde{\mathcal V}^{\rm diag}_1:=\mathcal V^{\rm diag}_1\bigl[\mathcal S^j(\rho^{\pm1})\mid j\in\mathbb Z\bigr], \qquad
\rho^2=\mathcal S^{-1}(\Delta)=\widebar w_0^{-1}.
\end{equation}
For $x\in\{T_n,\widebar T_n\}$, we extend $D_x$ uniquely to
$\widetilde{\mathcal V}^{\rm diag}_1$ by setting
\begin{equation*}
D_x(\rho)
=
\frac{\mathcal S^{-1}(D_x\Delta)}{2\rho},
\end{equation*}
and requiring it to commute with $\mathcal S$. By uniqueness of this extension, the commutation relations among
these derivations are preserved on
$\widetilde{\mathcal V}_1^{\rm diag}$. Hence, we can define the symmetrically normalized dressing operators
\begin{equation}\label{eq:WbarWevenPfaff}
W^{\rm s}:=\rho W, \qquad \widebar W^{\,\rm s}:=\rho\widebar W,
\end{equation}
 and the corresponding symmetrically gauged Toda Lax operators 
\begin{equation}\label{eq:LbarL}
L^{\rm s}:=W^{\rm s}\mathcal S^{-1}(W^{\rm s})^{-1}=\rho L_{\rm Toda}\rho^{-1},
\qquad
\widebar L^{\,\rm s}:=\widebar W^{\,\rm s}\mathcal S(\widebar W^{\,\rm s})^{-1}=
\rho\widebar L_{\rm Toda}\rho^{-1}.
\end{equation}
We compare this normalization with the symmetric gauge of the Pfaff hierarchy introduced in Section~\ref{sec:odd_extension}. Recall that the symmetric gauge $G$ factors as the product of the shift gauge $Q$ and the degree-zero gauge $H$ in~\eqref{QHG} as:
\begin{equation*}
Q=
\begin{pmatrix}
1&0\\
0&\mathcal S^{-1}
\end{pmatrix},
\qquad
H=
\begin{pmatrix}
\rho&0\\
0&\rho^{-1}
\end{pmatrix},
\qquad
G:=HQ,
\end{equation*}
and let $\mathcal{L}=GL_\Lambda G^{-1}$ be the symmetrically gauged Pfaff operator in~\eqref{eq:symmetric-pfaff-lax}. Also, we note that the projection $\pi\colon \mathcal{V}_N \to \mathcal{V}_N^{\textup{diag}}$ in~\eqref{eq:block-specialization} preserves the surviving leading coefficient $\Delta$ in $\widebar{V}_2$, i.e.\ $\pi(\Delta)=\Delta$. The specialization $\pi$ extends to its square-root extension by $\pi(\rho)=\rho$, and therefore $\pi(G)=G$. 
We define $M:=\pi( \mathcal L)$, and by the form of $L_\Lambda$ in~\eqref{eq:ungauged-pfaff-dressing} we can write
\begin{equation}\label{eq:M_evenPf}
M=GU^{\rm diag}\Lambda(U^{\rm diag})^{-1}G^{-1},
\end{equation}
where $\Lambda$ is the fundamental bare operator associated with Pfaff~\eqref{fundope}. 
\begin{lemma}\label{lem:M_Msquare}
The operator $M$ in~\eqref{eq:M_evenPf} and its square are given in terms of $W^{\rm s},\widebar W^{\,\rm s},L^{\rm s},\widebar{L}^{\,\rm s}$ as 
\begin{align}
\label{eq:M_inL}
M&=\begin{pmatrix}
0&W^{\rm s}\mathcal S\bigl(\widebar W^{\,\rm s}\bigr)^{\!*}
\\[1mm]
\bigl(\bigl(\widebar W^{\,\rm s}\bigr)^{\!*}\bigr)^{-1}
\mathcal S^{-2}(W^{\rm s})^{-1}&0
\end{pmatrix}=:\begin{pmatrix}
0&L_1\\
L_2&0
\end{pmatrix},\\[3mm]
\label{eq:M2_inL}
M^2&=\begin{pmatrix}
L^{\rm s}&0\\
0&\bigl(\widebar L^{\,\rm s}\bigr)^{\!*}
\end{pmatrix}=\begin{pmatrix}
L_1L_2&0\\
0&L_2L_1
\end{pmatrix}.
\end{align} 
\end{lemma}
\begin{proof}
From the definition of $\widebar{W}^{\rm s}$ in~\eqref{eq:WbarWevenPfaff} and taking into account~\eqref{eq:WbarW_n1_definition}, the inverse of its adjoint is
\begin{equation*}
\bigl((\widebar W^{\,\rm s})^*\bigr)^{-1}=\rho^{-1}(\widebar W^*)^{-1}=\rho^{-1}\mathcal S^{-1}\widebar V_2\mathcal S.
\end{equation*}
Therefore
\begin{equation*}
\begin{pmatrix}
W^{\rm s}&0\\
0&\bigl((\widebar W^{\,\rm s})^*\bigr)^{-1}
\end{pmatrix}
Q
=
\begin{pmatrix}
\rho W&0\\
0&\rho^{-1}\mathcal S^{-1}\widebar V_2
\end{pmatrix}
=
HQU^{\rm diag}
=
GU^{\rm diag}.
\end{equation*}
The conjugation of $\Lambda$ in~\eqref{fundope} by $Q$ yields
\begin{equation*}
Q\Lambda Q^{-1}=\begin{pmatrix}
0&\mathcal S\\
\mathcal S^{-2}&0
\end{pmatrix},
\end{equation*}
such that we obtain
\begin{equation*}
M=\begin{pmatrix}
W^{\rm s}&0\\
0&((\widebar W^{\,\rm s})^*)^{-1}
\end{pmatrix}
Q\Lambda Q^{-1}
\begin{pmatrix}
(W^{\rm s})^{-1}&0\\
0&(\widebar W^{\,\rm s})^*
\end{pmatrix},
\end{equation*}
which gives the first identity. Squaring it yields
\begin{equation*}
M^2=\begin{pmatrix}
W^{\rm s}\mathcal S^{-1}(W^{\rm s})^{-1}&0\\
0&
((\widebar W^{\,\rm s})^*)^{-1}
\mathcal S^{-1}(\widebar W^{\,\rm s})^*
\end{pmatrix}
=
\begin{pmatrix}
L^{\rm s}&0\\
0&(\widebar L^{\,\rm s})^*
\end{pmatrix},
\end{equation*}
using the definition of the pair $L^{\rm s},\widebar L^{\rm s}$ in~\eqref{eq:LbarL} and the adjoint.
\end{proof}
The square of the reduced Pfaff operator simultaneously encodes $L^{\rm s}$
 and $(\widebar L^{\rm s})^*$, hence
the two symmetrically gauged scalar Toda Lax operators. 

\begin{proposition}\label{prop:antidiagM_evenflows}
For every $n\geq1$, the anti-diagonal Toda derivation $D_n^{\rm ad}:=D_{T_n}-D_{\widebar T_n}$
acts on $M$ as the $2n$-th Pfaff flow:
\begin{equation}\label{eq:antidiag_even_flows}
D_n^{\rm ad}M=\bigl[P_+^{\rm Pf}(M^{2n}),M\bigr]=\bigl[-P_-^{\rm Pf}(M^{2n}),M\bigr].
\end{equation}
Thus the anti-diagonal subhierarchy of symmetrically gauged scalar
$2$D Toda coincides, under the realization by $M$, with the even
Pfaff hierarchy.
\end{proposition}
\begin{proof}
By~\eqref{lamrel}, the dressing derivative is $D_{\Lambda^{2n}}=D_{T_n}-D_{\widebar T_n}$.
Since $\Lambda^{2n}$ is block diagonal, Lemma~\ref{lem:diagonal} and the extension $\pi(\rho)=\rho$ give
\begin{equation*}
D_n^{\rm ad}\pi(f)=\pi\bigl(D_{\Lambda^{2n}}f\bigr)    
\end{equation*}
on the square-root extension $\widetilde{\mathcal V}^{\rm diag}_1$ in~\eqref{eq:Vdiag_tilde}. In particular,
\begin{equation*}
D_n^{\rm ad}M
=
\pi\bigl(D_{\Lambda^{2n}}\mathcal L\bigr).
\end{equation*}
Applying Proposition~\ref{prop:Pfaff_flows_symm_gauged} to the even powers gives
\begin{equation*}
D_{\Lambda^{2n}}L=\bigl[P_+^{\rm Pf}(L^{2n}),L\bigr].
\end{equation*}
Since the Pfaff projection $P_+^{\rm Pf}$ in~\eqref{plussplittngpfaffpseudo} is defined coefficient-wise in terms of Laurent projections, the formal adjoint, and the degree-zero matrix decomposition, it commutes with the block-diagonal projection $\pi$. Hence
\begin{equation*}
D_n^{\rm ad}M=\bigl[P_+^{\rm Pf}(M^{2n}),M\bigr].
\end{equation*}
Finally, since $[M^{2n},M]=0$ and
\begin{equation*}
M^{2n}=P_+^{\rm Pf}(M^{2n})+P_-^{\rm Pf}(M^{2n}),
\end{equation*}
we also obtain
\begin{equation*}
D_n^{\rm ad}M=\bigl[-P_-^{\rm Pf}(M^{2n}),M\bigr].
\end{equation*}
\end{proof}

\subsection{The C--Toda reduction}
We restrict the scalar Toda dressing space of the preceding subsection to the C--Toda locus of Krichever--Zabrodin \cite{KriZab},
\begin{equation}
    \widebar W^{\rm s}(W^{\rm s})^*=1.
    \label{eq:C-Toda-dressing}
\end{equation}
On this locus, using Lemma~\ref{lem:M_Msquare}, we obtain
\begin{equation*}
    L_1=W^{\rm s}\mathcal S(W^{\rm s})^{-1}=(L^{\rm s})^{-1},
    \qquad L_2=W^{\rm s}\mathcal S^{-2}(W^{\rm s})^{-1}=(L^{\rm s})^2,
\end{equation*}
and in particular,
\begin{equation}
    L_1^2L_2=1.
    \label{Ctodarel}
\end{equation}
Krichever and Zabrodin prove in \cite{KriZab} that the dressing constraint~\eqref{eq:C-Toda-dressing} is preserved by the anti-diagonal Toda flows. By Proposition~\ref{prop:antidiagM_evenflows}, these flows coincide under the Pfaff realization with the even Pfaff hierarchy.
Thus~\eqref{Ctodarel} is preserved by the derivations $D_n^{\rm ad}$ in~\eqref{eq:antidiag_even_flows}. Setting
\begin{equation*}
    L^C:=L_1^{-1}=L^{\rm s},
\end{equation*}
the operators $M$ in~\eqref{eq:M_inL} and $M^2$ in~\eqref{eq:M2_inL} become 
\begin{equation*}
    M=\begin{pmatrix}
    0&(L^C)^{-1}\\
    (L^C)^2&0
    \end{pmatrix},
    \qquad
    M^2=L^C I_2.
    \label{eq:C-Toda-Pfaff}
\end{equation*}
Consequently, $M^{2n}=(L^C)^nI_2$,
and the Pfaff projection reduces to
\begin{equation}
    P_+^{\rm Pf}(M^{2n})=B_n^C I_2,
    \qquad
    B_n^C:=\big((L^C)^n\big)_{<0}-\big(((L^C)^n)^*\big)_{>0}.
    \label{eqq}
\end{equation}
Substituting \eqref{eqq} into the $(1,2)$-entry of the even Pfaff Lax equation~\eqref{eq:antidiag_even_flows} gives
\begin{equation*}
    D_n^{\rm ad}(L^C)^{-1}
    =
    [B_n^C,(L^C)^{-1}],
\end{equation*}
or equivalently
\begin{equation*}
    D_n^{\rm ad}L^C
    =
    \left[
    \big((L^C)^n\big)_{<0}
    -
    \big(((L^C)^n)^*\big)_{>0},
    L^C
    \right].
    \label{eq:C-Toda-Lax}
\end{equation*}
This is the C--Toda Lax hierarchy \cite{KriZab} in our shift and time conventions.

\section{Conclusion}

In this paper we have developed a matrix pseudodifference operator
framework which places the Adler--Pfaff and Pfaff--Toda hierarchies,
together with several Toda-type reductions, in a common dressing
picture. 
We realize the classical Adler--Pfaff hierarchy as a local
$2\times2$ matrix pseudodifference Lax hierarchy, and
introduce the Matrix Pfaff--Toda hierarchy by dressing the full matrix Laurent algebra $M_{2N}(\mathbb{C}[\mathcal{S}, \mathcal{S}^{-1}])$. For $N=1$, its diagonal subalgebra gives Takasaki's continuous Pfaff--Toda hierarchy, while the off-diagonal directions supply the odd flows and, together with the diagonal directions, reconstruct the Adler--Pfaff hierarchy. For general $N$,
the Savchenko--Zabrodin tau-functions on the regular factorization locus realize the commuting diagonal sector of the matrix hierarchy.
The same construction also contains the Toda reductions considered in Section~\ref{sec:toda-reductions}. Block-diagonal dressing yields the multicomponent $2D$
Toda hierarchy: for $N=1$, its anti-diagonal subhierarchy coincides
with the even Pfaff hierarchy after symmetric gauging, while the
Krichever--Zabrodin dressing constraint gives the C--Toda reduction.

A main feature of this framework is that the Pfaff and Toda pictures are no longer treated as separate operator formalisms. The matrix pseudodifference dressing representation makes their relation visible at the level of the bare algebra: the continuous Pfaff–Toda hierarchy appears as a commuting diagonal sector, while in the one-component case the additional matrix directions reconstruct the full Adler–Pfaff hierarchy. At the same time, suitable restrictions of the dressing operator and of the bare algebra produce the Toda reductions within the same formalism. Thus the Pfaff, Pfaff–Toda and Toda hierarchies considered here arise as different sectors and reductions of a single matrix dressing construction.

The regular factorization locus arising in the comparison with the Savchenko--Zabrodin hierarchy deserves further study. We have shown that it contains nontrivial families of Clifford-group elements, but the question whether it is generic in the admissible class is still open. 

The Hamiltonian geometry underlying the Matrix Pfaff--Toda
hierarchy remains to be understood. The operator formulation developed here may provide a useful starting point for this question.
It also makes finite-band and
finite-dimensional reductions particularly accessible. It would
be interesting to determine whether distinguished reductions of
the Pfaff and Pfaff--Toda Lax operators admit an interpretation in
terms of difference Drinfeld--Sokolov reduction and the deformed
$W$-algebras of Frenkel--Reshetikhin, Semenov-Tian-Shansky and
Sevostyanov
\cite{FRSTS98,STSSev98}.

\appendix
\section{Proof of the Savchenko--Zabrodin realization theorem}
\label{sec:SZ-identification}
This appendix proves the Savchenko--Zabrodin realization theorem
stated in Section~\ref{sec:SZ-realization}. Starting from the fermionic bilinear identity
of \cite{SavchenkoZabrodin}, we construct a pair of raw matrix dressing operators, normalize them on the regular factorization locus, and derive the diagonal Matrix Pfaff--Toda Sato equations. We also establish non-emptiness of the admissible and regular loci.

Let $g$ be a Clifford-group element satisfying the fermionic bilinear
identity \cite[Eq.~(3.3)]{SavchenkoZabrodin},
\begin{equation}
\begin{aligned}
&\sum_{\gamma=1}^{N}\operatorname{res}_{z}\frac{dz}{z}\,
\Bigl(
\psi^{(\gamma)}(z)g\otimes\psi^{*(\gamma)}(z)g
+\psi^{*(\gamma)}(z)g\otimes\psi^{(\gamma)}(z)g
\Bigr)=
\\
&
\sum_{\gamma=1}^{N}\operatorname{res}_{z}\frac{dz}{z}\,
\Bigl(
g\psi^{(\gamma)}(z)\otimes g\psi^{*(\gamma)}(z)
+g\psi^{*(\gamma)}(z)\otimes g\psi^{(\gamma)}(z)
\Bigr).
\end{aligned}
\label{fermionic-bilinear}
\end{equation}

Here $\psi^{(\alpha)}(z)$ and $\psi^{*(\alpha)}(z)$ are the standard charged fermion fields of the $\alpha$-th component, $|n\rangle$ denotes the charged vacuum of charge $n\in\mathbb Z^N$, and $g$ acts on the multicomponent fermionic Fock space.

\subsection{Fermionic conventions and coefficient algebra}

We establish the notation, sign conventions, and bosonization rules underlying the Savchenko--Zabrodin formalism, and introduce the coefficient algebra $\mathcal V_g$ obtained by localizing the tau-function algebra at the shifts of $\tau_0$.

Let $\boldsymbol {s}=(s_1,\ldots,s_N)$, $\boldsymbol {r}=(r_1,\ldots,r_N)\in\mathbb Z^N$,
and set
\begin{equation*}
\boldsymbol {n}=\boldsymbol {s}+\boldsymbol {r},
\qquad
\widebar{\boldsymbol {n}}=\boldsymbol {r}-\boldsymbol {s}.
\label{nr-definition}
\end{equation*}
For $\alpha=1,\ldots,N$, we denote by $\boldsymbol {e}_\alpha$ the $\alpha$-th standard basis vector of $\mathbb Z^N$ and write
$
\boldsymbol {1}=(1,\ldots,1).
$
Write the time variables as
\begin{equation*}
\boldsymbol {t}=(\boldsymbol {t}_1,\ldots,\boldsymbol {t}_N),
\qquad
\boldsymbol {t}_\alpha=(t_{\alpha,1},t_{\alpha,2},\ldots),
\qquad
\widebar{\boldsymbol {t}}
=
(\widebar{\boldsymbol {t}}_1,\ldots,\widebar{\boldsymbol {t}}_N),
\qquad
\widebar{\boldsymbol {t}}_\alpha
=
(\widebar t_{\alpha,1},\widebar t_{\alpha,2},\ldots).
\end{equation*}

\subsubsection*{Fermionic signs}

We introduce the sign convention used by Savchenko--Zabrodin~\cite[Eq.~(3.6)]{SavchenkoZabrodin}. For $m\in\mathbb Z^N$ and $\alpha,\beta=1,\ldots,N$, we define 
\begin{equation*}
\epsilon_\alpha(\boldsymbol {m})=
(-1)^{m_{\alpha+1}+\cdots+m_N}.
\end{equation*}
We further introduce a doubled index sign function as 
\begin{equation}\label{epsilon-alphabeta}
\epsilon_{\alpha\beta}
=
\begin{cases}
1,& \alpha\leq\beta,\\
-1,& \alpha>\beta,
\end{cases}
\qquad
\epsilon_{\alpha\beta}(\boldsymbol {m})=
\epsilon_{\alpha\beta}\,
\epsilon_\alpha(\boldsymbol {m})
\epsilon_\beta(\boldsymbol {m}),
\end{equation}
and equivalently, we can write
\begin{equation}\label{epsilon-alphabeta-explicit}
\epsilon_{\alpha\beta}(\boldsymbol {m})=
\begin{cases}
(-1)^{m_{\alpha+1}+\cdots+m_\beta},
& \alpha<\beta,\\
1,& \alpha=\beta,\\
-(-1)^{m_{\beta+1}+\cdots+m_\alpha},
& \alpha>\beta.
\end{cases}
\end{equation}
This is the sign factor appearing in
\cite[Eq.~(3.15)]{SavchenkoZabrodin}. Since
$\boldsymbol {n}=\boldsymbol {s}+\boldsymbol {r}$, we shall repeatedly use
\begin{equation}\label{epsilon-identities}
\epsilon_{\alpha\beta}\,
\epsilon_{\alpha\beta}(\boldsymbol {s})
\epsilon_{\alpha\beta}(\boldsymbol {r})=
\epsilon_{\alpha\beta}(\boldsymbol {n}),
\qquad
\epsilon_{\alpha\beta}(\boldsymbol {s})
\epsilon_{\alpha\beta}(\boldsymbol {r})=
\epsilon_\alpha(\boldsymbol {n})
\epsilon_\beta(\boldsymbol {n}).
\end{equation}

\subsubsection*{Fermionic currents}

If $J_k^{(\alpha)}$ denotes the fermionic current mode of the
$\alpha$-th component, we introduce
\begin{equation*}
J(\boldsymbol {t})=\sum_{\alpha=1}^N\sum_{k\geq1}
t_{\alpha,k}J_k^{(\alpha)},
\qquad
\widebar J(\widebar{\boldsymbol {t}})=
\sum_{\alpha=1}^N\sum_{k\geq1}
\widebar t_{\alpha,k}J_{-k}^{(\alpha)},
\label{current-operators}
\end{equation*}
interpreted as the current operators associated with the times $\boldsymbol  t,\widebar{\boldsymbol  t}$ introduced in
\cite[§2]{SavchenkoZabrodin}. We also write, following
\cite[Eq.~(2.4)]{SavchenkoZabrodin},
\begin{equation*}
\xi(\boldsymbol {t}_\alpha,z)
=
\sum_{k\geq1}t_{\alpha,k}\,z^k.
\label{xi-definition}
\end{equation*}
We define the Miwa shifts by
\begin{equation*}
\big([z^{-1}]_\beta\big)_{\alpha,k}=
\delta_{\alpha\beta}\frac{z^{-k}}{k},
\qquad
\big([z]_\beta\big)_{\alpha,k}=
\delta_{\alpha\beta}\frac{z^k}{k}, \qquad \alpha, \beta =1, \ldots, N, \, \, \,  k \geq 1.
\label{miwa-shifts}
\end{equation*}

The current--fermion commutation relations are given in
\cite[Eq.~(2.3)]{SavchenkoZabrodin}. In the form that will be used below,
they read
\begin{equation}
\begin{aligned}
\,\textup{e}^{J(\boldsymbol {t})}\psi^{(\beta)}(z)&=
\,\textup{e}^{\xi(\boldsymbol {t}_\beta,z)}
\psi^{(\beta)}(z)\,\textup{e}^{J(\boldsymbol {t})},&
\,\textup{e}^{J(\boldsymbol {t})}\psi^{*(\beta)}(z)&=
\,\textup{e}^{-\xi(\boldsymbol {t}_\beta,z)}
\psi^{*(\beta)}(z)\,\textup{e}^{J(\boldsymbol {t})},
\\
\psi^{(\beta)}(z)\,\textup{e}^{-\widebar J(\widebar{\boldsymbol {t}})}&=
\,\textup{e}^{\xi(\widebar{\boldsymbol {t}}_\beta,z^{-1})}
\,\textup{e}^{-\widebar J(\widebar{\boldsymbol {t}})}
\psi^{(\beta)}(z),&
\psi^{*(\beta)}(z)\,\textup{e}^{-\widebar J(\widebar{\boldsymbol {t}})}&=
\,\textup{e}^{-\xi(\widebar{\boldsymbol {t}}_\beta,z^{-1})}
\,\textup{e}^{-\widebar J(\widebar{\boldsymbol {t}})}
\psi^{*(\beta)}(z).
\end{aligned}
\label{fermion-current}
\end{equation}
The second line is the equivalent form of the barred relations in
\cite[Eq.~(2.3)]{SavchenkoZabrodin}, adapted to commuting the fermions
through $\,\textup{e}^{-\widebar J(\widebar{\boldsymbol  t})}$.

\subsubsection*{Bosonization rules}
We shall use the standard multicomponent bosonization rules. In the
notation above, \cite[Eq.~(3.5)]{SavchenkoZabrodin} of Savchenko--Zabrodin
reads, for $\boldsymbol  m\in\mathbb Z^N$,
\begin{subequations}\label{bosonization}
\begin{align}
\bra{\boldsymbol  m}\psi^{(\beta)}(z)&=
\epsilon_\beta(\boldsymbol  m)
\, z^{m_\beta-1}
\bra{\boldsymbol  m-\boldsymbol  e_\beta}
\textup{e}^{-J([z^{-1}]_\beta)},
\\[1mm]
\bra{\boldsymbol  m}\psi^{*(\beta)}(z)&=\epsilon_\beta(\boldsymbol  m)
\, z^{-m_\beta}
\bra{\boldsymbol  m+\boldsymbol  e_\beta}
\textup{e}^{J([z^{-1}]_\beta)},\\[1mm]
\psi^{(\beta)}(z)\ket{\boldsymbol  m}&=
\epsilon_\beta(\boldsymbol  m)
\,z^{m_\beta}
\,\textup{e}^{\widebar J([z]_\beta)}
\ket{\boldsymbol  m+\boldsymbol  e_\beta},
\\[1mm]
\psi^{*(\beta)}(z)\ket{\boldsymbol  m}&=
\epsilon_\beta(\boldsymbol  m)
\, z^{-m_\beta+1}
\,\textup{e}^{-\widebar J([z]_\beta)}
\ket{\boldsymbol  m-\boldsymbol  e_\beta}.
\end{align}
\end{subequations}
We will repeatedly use the component-charge selection rule
\begin{equation}
\bra{\boldsymbol  a}H\ket{\boldsymbol  b}=0
\qquad\text{if $H$ has no component of charge
$\boldsymbol  a-\boldsymbol  b$,}
\label{chargeselection}
\end{equation}
as well as the charge-independent vacuum expectation
\begin{equation}
\bra{\boldsymbol  m}
\,\textup{e}^{J(\boldsymbol  t)}\,\textup{e}^{-\widebar J(\widebar{\boldsymbol  t})}
\ket{ \boldsymbol  m }=\exp\!\left(
-\sum_{\alpha=1}^N\sum_{k\ge1}
k\,t_{\alpha,k}\widebar t_{\alpha,k}
\right).
\label{charge independence}
\end{equation}
\subsubsection*{Matrix tau-function}

The tau-function of the multicomponent Pfaff--Toda hierarchy is defined
by Savchenko--Zabrodin \cite[Eq.~(3.8)]{SavchenkoZabrodin} as
\begin{equation}
\tau(\boldsymbol  n,\widebar{\boldsymbol  n},\boldsymbol  t,\widebar{\boldsymbol  t})
=
\bra{\boldsymbol  n}
\,\textup{e}^{J(\boldsymbol  t)} g \,\textup{e}^{-\widebar J(\widebar{\boldsymbol  t})}
\ket{-\widebar{\boldsymbol  n}}.
\label{tau-function}
\end{equation}
Passing to the independent variables $\boldsymbol  s,\boldsymbol  r$, they
introduce in \cite[Eq.~(3.13)]{SavchenkoZabrodin} the array
\begin{equation}
\tau_{\alpha\beta}
(\boldsymbol  s,\boldsymbol  r,\boldsymbol  t,\widebar{\boldsymbol  t})
:=
\tau(
\boldsymbol  n+\boldsymbol  e_\alpha-\boldsymbol  e_\beta,
\widebar{\boldsymbol  n},
\boldsymbol  t,\widebar{\boldsymbol  t}),
\label{matrix-tau}
\end{equation}
which may be regarded as an $N\times N$ matrix tau-function. Its
diagonal entries have the common value
\begin{equation}\label{tau-zero}
\tau_0:=\tau_{\alpha\alpha}=
\tau(\boldsymbol  n,\widebar{\boldsymbol  n},\boldsymbol  t,\widebar{\boldsymbol  t}).
\end{equation}
All normalized expressions below involve ratios of tau-functions
together with the fermionic sign factors introduced in~\eqref{epsilon-alphabeta}--\eqref{epsilon-alphabeta-explicit}. For a
Clifford-group element $g$, let $\mathcal A_g$ be the commutative
 algebra generated by all discrete translates
\begin{equation*}
\tau_{\alpha\beta}
(\boldsymbol  s+\boldsymbol  a,\boldsymbol  r+\boldsymbol  b, \boldsymbol  t,\widebar{\boldsymbol  t}),
\qquad
\epsilon_\alpha (\boldsymbol  s+\boldsymbol  r+\boldsymbol  a+\boldsymbol  b),
\end{equation*}
for $\boldsymbol  a,\boldsymbol  b\in\mathbb Z^N$ and
$\alpha,\beta=1,\ldots,N$, together with their continuous derivatives,
where the tau-functions are those associated with $g$ through
\eqref{tau-function}--\eqref{matrix-tau}.
We equip $\mathcal A_g$ with the automorphism $\mathcal S\colon \mathcal{A}_g \to \mathcal{A}_g$
\begin{equation}
(\mathcal S f)(\boldsymbol  s,\boldsymbol  r,\boldsymbol  t,\widebar{\boldsymbol  t})
=
f(\boldsymbol  s-\boldsymbol  1,\boldsymbol  r,\boldsymbol  t,\widebar{\boldsymbol  t}).
\label{shiftdef}
\end{equation}
\subsubsection*{The algebra $\mathcal V_g$ for admissible $g$}
Let $\mathcal M_g\subset\mathcal A_g$ be the multiplicative subset
generated by the iterated action of the automorphism $\mathcal{S}$ on the zeroth tau-function~\eqref{tau-zero} as $\left\{\mathcal S^k(\tau_0)\mid k\in\mathbb Z
\right\}$. We call $g$ \emph{admissible} if
$0\notin\mathcal M_g$.
For an admissible Clifford-group element $g$, we define its coefficient
difference--differential algebra by
\begin{equation}\label{eq:diff_diff_algebra_Clifford}
\mathcal V_g:=\mathcal M_g^{-1}\mathcal A_g.
\end{equation}
When the generators of $\mathcal A_g$ are regarded as functions of the
hierarchy variables, one may think of this localization as restricting
to the locus on which all shifts $\mathcal S^k(\tau_0)$ are non-zero.
Since $\mathcal S(\mathcal M_g)=\mathcal M_g$, the shift
automorphism extends uniquely to $\mathcal V_g$, and the continuous derivations likewise extend uniquely to $\mathcal V_g$ and commute
with $\mathcal S$. 
\begin{remark}
\label{rem:admissible-cliff-group}
The admissible locus is non-empty and contains nontrivial families of Clifford-group elements. Indeed, let
\begin{equation}\label{eq:admissible_g_form}
g_A:=\exp\!\left(
\frac12
\sum_{\alpha,\beta=1}^N
\sum_{i,j\in\mathbb Z}
A_{ij}^{\alpha\beta}
\psi_i^{(\alpha)}\psi_j^{(\beta)}
\right),
\end{equation}
where only finitely many coefficients $A_{ij}^{\alpha\beta}$ are
non-zero and $A_{ij}^{\alpha\beta}=-A_{ji}^{\beta\alpha}$.
These elements form a particular class of the Clifford-group elements
of \cite[Eq.~(3.1)]{SavchenkoZabrodin}.

At $\boldsymbol  r=\boldsymbol  t=\widebar{\boldsymbol  t}=\boldsymbol 0$, each monomial in
the $p$-th nonconstant term of the exponential changes the total
charge by $2p$. Hence
\begin{equation*}
\tau_0(\boldsymbol  s,\boldsymbol 0,\boldsymbol 0,\boldsymbol 0)
=
\langle\boldsymbol  s|g_A|\boldsymbol  s\rangle
=1.
\end{equation*}
Since
\begin{equation*}
\mathcal S^k(\tau_0)
(\boldsymbol  s,\boldsymbol 0,\boldsymbol 0,\boldsymbol 0)
=
\tau_0(\boldsymbol  s-k\boldsymbol 1,\boldsymbol 0,\boldsymbol 0,\boldsymbol 0)
=1,
\qquad k\in\mathbb Z,
\end{equation*}
every element
\begin{equation*}
u=
\prod_{j=1}^{\ell}\mathcal S^{k_j}(\tau_0)
\in\mathcal M_{g_A}
\end{equation*}
satisfies $u(\boldsymbol  s,\boldsymbol 0,\boldsymbol 0,\boldsymbol 0)=1$.
Hence $u\neq0$ for every $u\in\mathcal M_{g_A}$, and therefore
$0\notin\mathcal M_{g_A}$.
Thus every $g_A$ of the form~\eqref{eq:admissible_g_form} is admissible. In particular, the
admissible locus contains families of arbitrarily large finite dimension.
\end{remark}

\noindent
\emph{From now on, we assume that the Clifford-group element $g$ is admissible}. In particular, $\tau_0$ in~\eqref{tau-zero} and all of its $\mathcal S$-translates are invertible in $\mathcal V_g$~\eqref{eq:diff_diff_algebra_Clifford}.

\subsection{Matrix wave functions and matrix residue identity}
\label{subsec:matrix-wave-functions}
We define eight $N \times N$  matrix wave functions from the fermionic bilinear identity and assemble them into $2N \times 2N$ raw wave matrices, obtaining the matrix residue identity that encodes the bilinear constraint.

We apply suitable tensor products of charged bra states to the left of the fermionic bilinear identity~\eqref{fermionic-bilinear}, and the corresponding unshifted charged ket states to its right. 

\begin{proposition}[Matrix form of the fermionic bilinear identity]
\label{prop:matrix-bilinear-identity}
For $\alpha,\beta=1,\ldots,N$, define the four unbarred wave
matrices by
\begin{subequations}\label{eq:bosonized_unbarred}
\begin{align}
(\Psi_1)_{\alpha\beta}(z)
&:=\epsilon_\alpha(\boldsymbol  n)
\frac{
\bra{\boldsymbol  n+\boldsymbol  e_\alpha}
\,\textup{e}^{J(\boldsymbol  t)}
\psi^{(\beta)}(z)g\,\textup{e}^{-\widebar J(\widebar{\boldsymbol  t})}
\ket{-\widebar{\boldsymbol  n}}
}{\tau_0},
\label{eq:Psi1-def}
\\[1mm]
(\Psi_2)_{\alpha\beta}(z)
&:=\epsilon_\alpha(\boldsymbol  n)
\frac{
\bra{\boldsymbol  n-\boldsymbol  e_\alpha}
\,\textup{e}^{J(\boldsymbol  t)}
\psi^{(\beta)}(z)g\,\textup{e}^{-\widebar J(\widebar{\boldsymbol  t})}
\ket{-\widebar{\boldsymbol  n}}
}{\tau_0},
\label{eq:Psi2-def}
\\[1mm]
(\Psi_3)_{\alpha\beta}(z)
&:=z^{-1}\epsilon_\alpha(\boldsymbol  n)
\frac{
\bra{\boldsymbol  n+\boldsymbol  e_\alpha}
\,\textup{e}^{J(\boldsymbol  t)}
\psi^{*(\beta)}(z)g\,\textup{e}^{-\widebar J(\widebar{\boldsymbol  t})}
\ket{-\widebar{\boldsymbol  n}}
}{\tau_0},
\label{eq:Psi1star-def}
\\[1mm]
(\Psi_4)_{\alpha\beta}(z)
&:=z^{-1}\epsilon_\alpha(\boldsymbol  n)
\frac{
\bra{\boldsymbol  n-\boldsymbol  e_\alpha}
\,\textup{e}^{J(\boldsymbol  t)}
\psi^{*(\beta)}(z)g\,\textup{e}^{-\widebar J(\widebar{\boldsymbol  t})}
\ket{-\widebar{\boldsymbol  n}}
}{\tau_0},
\label{eq:Psi2star-def}
\end{align}
\end{subequations} 
and the four barred wave matrices by
\begin{subequations}\label{eq:bosonized_barred}
\begin{align}
(\widebar\Psi_1)_{\alpha\beta}(z)
&:=
\epsilon_\alpha(\boldsymbol  n)
\,\textup{e}^{\xi(\widebar{\boldsymbol  t}_\beta,z^{-1})}
\frac{\bra{\boldsymbol  n+\boldsymbol  e_\alpha}\,\textup{e}^{J(\boldsymbol  t)}g
\,\textup{e}^{-\widebar J(\widebar{\boldsymbol  t})}\psi^{(\beta)}(z)\ket{-\widebar{\boldsymbol  n}}}{\tau_0},
\label{eq:barPsi1-def}
\\
(\widebar\Psi_2)_{\alpha\beta}(z)
&:=
\epsilon_\alpha(\boldsymbol  n)
\,\textup{e}^{\xi(\widebar{\boldsymbol  t}_\beta,z^{-1})}
\frac{
\bra{\boldsymbol  n-\boldsymbol  e_\alpha}
\,\textup{e}^{J(\boldsymbol  t)}
g
\,\textup{e}^{-\widebar J(\widebar{\boldsymbol  t})}
\psi^{(\beta)}(z)
\ket{-\widebar{\boldsymbol  n}}}{\tau_0},
\label{eq:barPsi2-def}
\\
(\widebar\Psi_3)_{\alpha\beta}(z)
&:=
z^{-1}\epsilon_\alpha(\boldsymbol  n)
\,\textup{e}^{-\xi(\widebar{\boldsymbol  t}_\beta,z^{-1})}
\frac{
\bra{\boldsymbol  n+\boldsymbol  e_\alpha}
\,\textup{e}^{J(\boldsymbol  t)}
g
\,\textup{e}^{-\widebar J(\widebar{\boldsymbol  t})}
\psi^{*(\beta)}(z)
\ket{-\widebar{\boldsymbol  n}}}{\tau_0},
\label{eq:barPsi1star-def}
\\
(\widebar\Psi_4)_{\alpha\beta}(z)
&:=
z^{-1}\epsilon_\alpha(\boldsymbol  n)
\,\textup{e}^{-\xi(\widebar{\boldsymbol  t}_\beta,z^{-1})}
\frac{
\bra{\boldsymbol  n-\boldsymbol  e_\alpha}
\,\textup{e}^{J(\boldsymbol  t)}
g
\,\textup{e}^{-\widebar J(\widebar{\boldsymbol  t})}
\psi^{*(\beta)}(z)
\ket{-\widebar{\boldsymbol  n}}}{\tau_0}.
\label{eq:barPsi2star-def}
\end{align}
\end{subequations}
Introduce the $2N\times2N$  matrices
\begin{equation*}
\Phi_{\mathrm{raw}}
:=
\begin{pmatrix}
\Psi_1&\widebar\Psi_3\\
\Psi_2&\widebar\Psi_4
\end{pmatrix},
\qquad
\widebar\Phi_{\mathrm{raw}}
:=
\begin{pmatrix}
\widebar\Psi_1&\Psi_3\\
\widebar\Psi_2&\Psi_4
\end{pmatrix},
\qquad
J_0:=
\begin{pmatrix}
0&I_N\\
-I_N&0
\end{pmatrix},
\label{eq:raw-wave-matrices}
\end{equation*}
 and let $(\boldsymbol  s',\boldsymbol  r',\boldsymbol  t',\widebar{\boldsymbol  t}')$ denote an independent copy of the variables. Then we have
\begin{equation}
\operatorname{res}_{z}
\Phi_{\mathrm{raw}}(\boldsymbol  s,\boldsymbol  r,\boldsymbol  t,\widebar{\boldsymbol  t})\,
J_0\,
(\widebar\Phi_{\mathrm{raw}}(\boldsymbol  s',\boldsymbol  r',\boldsymbol  t',\widebar{\boldsymbol  t}'))^\top\,dz
=
\operatorname{res}_{z}
\widebar\Phi_{\mathrm{raw}}(\boldsymbol  s,\boldsymbol  r,\boldsymbol  t,\widebar{\boldsymbol  t})\,
J_0\,
(\Phi_{\mathrm{raw}}(\boldsymbol  s',\boldsymbol  r',\boldsymbol  t',\widebar{\boldsymbol  t}'))^\top\,dz.
\label{eq:raw-matrix-residue}
\end{equation}
\end{proposition}

\begin{proof}
Throughout the proof, a prime denotes evaluation at the independent variables $(\boldsymbol  s',\boldsymbol  r',\boldsymbol  t',\widebar{\boldsymbol  t}')$. We introduce the notation $
\boldsymbol  n'=\boldsymbol  s'+\boldsymbol  r'$ and $
\widebar{\boldsymbol  n}'=\boldsymbol  r'-\boldsymbol  s'.
$ The matrix residue identity~\eqref{eq:raw-matrix-residue} follows from applying~\eqref{fermionic-bilinear} to  
\begin{equation}\label{eq:left_mult}
    \bra{\boldsymbol l} := \bra{\boldsymbol  n+\sigma\boldsymbol  e_\alpha}
\,\textup{e}^{J(\boldsymbol  t)}
\otimes
\bra{\boldsymbol  n'+\sigma'\boldsymbol  e_\beta}
\,\textup{e}^{J(\boldsymbol  t')}
\end{equation}
on the left and
\begin{equation}\label{eq:right_mult}
\ket{ \boldsymbol r}:=\textup{e}^{-\widebar J(\widebar{\boldsymbol  t})}
\ket{-\widebar{\boldsymbol  n}}
\otimes
\,\textup{e}^{-\widebar J(\widebar{\boldsymbol  t}')}
\ket{-\widebar{\boldsymbol  n}'}
\end{equation}
on the right. Here, $\alpha,\beta\in\{1,\ldots,N\}$ and
$\sigma,\sigma'\in\{+1,-1\}$. More precisely, the correspondence between the $2 \times 2$ blocks of~\eqref{eq:raw-matrix-residue} and the pair of signs $(\sigma,\sigma')$ is given by the table 
\begin{equation*}
\begin{array}{c|c}
(\sigma,\sigma') & \text{block} \\
\hline  
(+1,+1) & (1,1)\\
(+1,-1) & (1,2)\\
(-1,+1) & (2,1)\\
(-1,-1) & (2,2).
\end{array}
\end{equation*}
We treat the $(1,1)$-entry in detail, the others are obtained through identical computations. For $(\sigma,\sigma')=(1,1)$ and with $\bra{\boldsymbol{l}} $ and $\ket{\boldsymbol{r}}$ in~\eqref{eq:left_mult} and~\eqref{eq:right_mult}, we define $Z_1:= \bra{\boldsymbol{l}} \text{LHS~\eqref{fermionic-bilinear}} \ket{\boldsymbol{r}} $, then 
\begin{equation*}
\begin{aligned}
Z_1
&=
\sum_{\gamma=1}^N\operatorname{res}_z\frac{dz}{z}
\Big(
\bra{ \boldsymbol  n+\boldsymbol  e_\alpha}\,\textup{e}^{J(\boldsymbol  t)}
\psi^{(\gamma)}(z)g \,\textup{e}^{-\widebar J(\widebar{\boldsymbol  t})}
\ket{-\widebar{\boldsymbol  n}}
 \times
\bra{ \boldsymbol  n'+\boldsymbol  e_\beta}\,\textup{e}^{J(\boldsymbol  t')}
\psi^{*(\gamma)}(z)g \,\textup{e}^{-\widebar J(\widebar{\boldsymbol  t}')}
\ket{-\widebar{\boldsymbol  n}'} \Big)
\\
&+ \sum_{\gamma=1}^N\operatorname{res}_z\frac{dz}{z}
\Big(
\bra{ \boldsymbol  n+\boldsymbol  e_\alpha}\,\textup{e}^{J(\boldsymbol  t)}
\psi^{*(\gamma)}(z)g \,\textup{e}^{-\widebar J(\widebar{\boldsymbol  t})}
\ket{-\widebar{\boldsymbol  n} }
\times
\bra{ \boldsymbol  n'+\boldsymbol  e_\beta}\,\textup{e}^{J(\boldsymbol  t')}
\psi^{(\gamma)}(z)g \,\textup{e}^{-\widebar J(\widebar{\boldsymbol  t}')}
\ket{-\widebar{\boldsymbol  n}' }
\Big)
\\[1mm]
&=
\frac{\tau_0\tau_0'}
{\epsilon_\alpha(\boldsymbol  n)\epsilon_\beta(\boldsymbol  n')}
\operatorname{res}_z
\sum_{\gamma=1}^N
\Big(
(\Psi_1)_{\alpha\gamma}(\Psi_3')_{\beta\gamma}
+
(\Psi_3)_{\alpha\gamma}(\Psi_1')_{\beta\gamma}
\Big)\,dz
\\
&=
\frac{\tau_0\tau_0'}
{\epsilon_\alpha(\boldsymbol  n)\epsilon_\beta(\boldsymbol  n')}
\operatorname{res}_z
\Big(
\Psi_1(\Psi_3')^\top
+
\Psi_3(\Psi_1')^\top
\Big)_{\alpha\beta}\,dz
\\
&=
\frac{\tau_0\tau_0'}
{\epsilon_\alpha(\boldsymbol  n)\epsilon_\beta(\boldsymbol  n')}
\operatorname{res}_z
\Big(
\Phi_{\rm raw}E_+(\widebar\Phi_{\rm raw}')^\top+
\widebar\Phi_{\rm raw}E_-(\Phi_{\rm raw}')^\top
\Big)_{11,\alpha\beta}\,dz,
\end{aligned}
\end{equation*}
where
\begin{equation}\label{eq:EplusEminus}
E_+:=
\begin{pmatrix}
0&I_N\\
0&0
\end{pmatrix},
\qquad
E_-:=\begin{pmatrix}
0&0\\
I_N&0
\end{pmatrix}.
\end{equation}
Second, we introduce $Z_2:= \bra{\boldsymbol{l}} \text{RHS~\eqref{fermionic-bilinear}} \ket{\boldsymbol{r}} $, with $\bra{\boldsymbol{l}} $ and $\ket{\boldsymbol{r}}$ in~\eqref{eq:left_mult} and~\eqref{eq:right_mult}, and using the fermion-current relations~\eqref{fermion-current}, we have
\begin{equation*}
\begin{aligned}
Z_2
&=
\sum_{\gamma=1}^N
\operatorname{res}_z\frac{dz}{z}
\Big(
\bra{ \boldsymbol  n+\boldsymbol  e_\alpha}\,\textup{e}^{J(\boldsymbol  t)}
g\psi^{(\gamma)}(z)\,\textup{e}^{-\widebar J(\widebar{\boldsymbol  t})}
\ket{-\widebar{\boldsymbol  n} }
\times
\bra{ \boldsymbol  n'+\boldsymbol  e_\beta}\,\textup{e}^{J(\boldsymbol  t')}
g\psi^{*(\gamma)}(z)\,\textup{e}^{-\widebar J(\widebar{\boldsymbol  t}')}
\ket{-\widebar{\boldsymbol  n}' }
\Big)
\\
&\quad+
\sum_{\gamma=1}^N
\operatorname{res}_z\frac{dz}{z}
\Big(
\bra{ \boldsymbol  n+\boldsymbol  e_\alpha}\,\textup{e}^{J(\boldsymbol  t)}
g\psi^{*(\gamma)}(z)\,\textup{e}^{-\widebar J(\widebar{\boldsymbol  t})}
\ket{-\widebar{\boldsymbol  n} }
\times
\bra{ \boldsymbol  n'+\boldsymbol  e_\beta}\,\textup{e}^{J(\boldsymbol  t')}
g\psi^{(\gamma)}(z)\,\textup{e}^{-\widebar J(\widebar{\boldsymbol  t}')}
\ket{-\widebar{\boldsymbol  n}'}
\Big)
\\
&=
\sum_{\gamma=1}^N
\operatorname{res}_z\frac{dz}{z}
\Big(
\,\textup{e}^{\xi(\widebar{\boldsymbol  t}_\gamma,z^{-1})-\xi(\widebar{\boldsymbol  t}'_\gamma,z^{-1})}
\bra{ \boldsymbol  n+\boldsymbol  e_\alpha}\,\textup{e}^{J(\boldsymbol  t)}
g \,\textup{e}^{-\widebar J(\widebar{\boldsymbol  t})}\psi^{(\gamma)}(z)
\ket{-\widebar{\boldsymbol  n}} \times
\\
&\hspace{55ex}\times
\bra{ \boldsymbol  n'+\boldsymbol  e_\beta}\,\textup{e}^{J(\boldsymbol  t')}
g \,\textup{e}^{-\widebar J(\widebar{\boldsymbol  t}')}\psi^{*(\gamma)}(z)
\ket{-\widebar{\boldsymbol  n}' }
\Big)
\\
&\quad+
\sum_{\gamma=1}^N
\operatorname{res}_z\frac{dz}{z}
\Big(
\,\textup{e}^{-\xi(\widebar{\boldsymbol  t}_\gamma,z^{-1})+\xi(\widebar{\boldsymbol  t}'_\gamma,z^{-1})}
\bra{ \boldsymbol  n+\boldsymbol  e_\alpha}\,\textup{e}^{J(\boldsymbol  t)}
g \,\textup{e}^{-\widebar J(\widebar{\boldsymbol  t})}\psi^{*(\gamma)}(z)
\ket{-\widebar{\boldsymbol  n} }
\times
\\
&\hspace{55ex}\times\bra{ \boldsymbol  n'+\boldsymbol  e_\beta}\,\textup{e}^{J(\boldsymbol  t')}
g \,\textup{e}^{-\widebar J(\widebar{\boldsymbol  t}')}\psi^{(\gamma)}(z)
\ket{-\widebar{\boldsymbol  n}' }
\Big)
\\
&=
\frac{\tau_0\tau_0'}
{\epsilon_\alpha(\boldsymbol  n)\epsilon_\beta(\boldsymbol  n')}
\operatorname{res}_z
\sum_{\gamma=1}^N
\Big(
(\widebar\Psi_1)_{\alpha\gamma}(\widebar\Psi'_3)_{\beta\gamma}
+
(\widebar\Psi_3)_{\alpha\gamma}(\widebar\Psi'_1)_{\beta\gamma}
\Big)\,dz
\\
&=
\frac{\tau_0\tau_0'}
{\epsilon_\alpha(\boldsymbol  n)\epsilon_\beta(\boldsymbol  n')}
\operatorname{res}_z
\Big(
\widebar\Psi_1(\widebar\Psi'_3)^\top
+
\widebar\Psi_3(\widebar\Psi'_1)^\top
\Big)_{\alpha\beta}\,dz
\\
&=
\frac{\tau_0\tau_0'}
{\epsilon_\alpha(\boldsymbol  n)\epsilon_\beta(\boldsymbol  n')}
\operatorname{res}_z
\Big(
\widebar\Phi_{\rm raw}E_+(\Phi'_{\rm raw})^\top
+
\Phi_{\rm raw}E_-(\widebar\Phi'_{\rm raw})^\top
\Big)_{11,\alpha\beta}\,dz .
\end{aligned}
\end{equation*}
Since $Z_1=Z_2$ for all $\alpha,\beta=1,\ldots,N$, and given that
$\tau_0$ and $\tau_0'$ are invertible in the corresponding unprimed and primed copies of $\mathcal V_g$, cancellation of the
common prefactor gives
\begin{equation*}
\operatorname{res}_z
\Big(
\Phi_{\rm raw}E_+(\widebar\Phi_{\rm raw}')^\top
+
\widebar\Phi_{\rm raw}E_-(\Phi_{\rm raw}')^\top
\Big)_{11}\,dz = \operatorname{res}_z
\Big(
\widebar\Phi_{\rm raw}E_+(\Phi'_{\rm raw})^\top
+
\Phi_{\rm raw}E_-(\widebar\Phi'_{\rm raw})^\top
\Big)_{11}\,dz .
\end{equation*}
Using the relation $J_0=E_+-E_-$ with~\eqref{eq:EplusEminus}, we conclude that the matrix residue identity holds on the $(1,1)$ block. 
\end{proof}

 To translate the residue identity~\eqref{eq:raw-matrix-residue} into
pseudodifference-operator form, we first separate the universal
dependence of the wave functions on the charge variables and
continuous times. The following bosonization formulas make this
separation explicit.

\begin{lemma}[Bosonized matrix wave functions]
\label{lem:bosonized-wave-functions}
The eight matrix wave functions introduced above are given by
\begin{subequations}\label{eq:bosonized_wavefunction}
\begin{align}
(\Psi_1)_{\alpha\beta}(z)
&=
\epsilon_{\alpha\beta}(\boldsymbol  n)\,
z^{n_\beta+\delta_{\alpha\beta}-1}
\,\textup{e}^{\xi(\boldsymbol  t_\beta,z)}
\frac{
\tau_{\alpha\beta}
(\boldsymbol  s,\boldsymbol  r,
 \boldsymbol  t-[z^{-1}]_\beta,\widebar{\boldsymbol  t})
}{
\tau_0
},
\label{eq:Psi1-bos}
\\
(\Psi_2)_{\alpha\beta}(z)
&=
\epsilon_{\alpha\beta}(\boldsymbol  n)\,
z^{n_\beta-\delta_{\alpha\beta}-1}
\,\textup{e}^{\xi(\boldsymbol  t_\beta,z)}
\frac{
\tau_{\alpha\beta}
(\boldsymbol  s-\boldsymbol  e_\alpha,
 \boldsymbol  r-\boldsymbol  e_\alpha,
 \boldsymbol  t-[z^{-1}]_\beta,\widebar{\boldsymbol  t})
}{
\tau_0
},
\label{eq:Psi2-bos}
\\
(\Psi_3)_{\alpha\beta}(z)
&=
\epsilon_{\alpha\beta}(\boldsymbol  n)\,
z^{-n_\beta-\delta_{\alpha\beta}-1}
\,\textup{e}^{-\xi(\boldsymbol  t_\beta,z)}
\frac{
\tau_{\beta\alpha}
(\boldsymbol  s+\boldsymbol  e_\alpha,
 \boldsymbol  r+\boldsymbol  e_\alpha,
 \boldsymbol  t+[z^{-1}]_\beta,\widebar{\boldsymbol  t})
}{
\tau_0
},
\label{eq:Psi1star-bos}
\\
(\Psi_4)_{\alpha\beta}(z)
&=
\epsilon_{\alpha\beta}(\boldsymbol  n)\,
z^{-n_\beta+\delta_{\alpha\beta}-1}
\,\textup{e}^{-\xi(\boldsymbol  t_\beta,z)}
\frac{
\tau_{\beta\alpha}
(\boldsymbol  s,\boldsymbol  r,
 \boldsymbol  t+[z^{-1}]_\beta,\widebar{\boldsymbol  t})
}{
\tau_0
},
\label{eq:Psi2star-bos}
\end{align}
\end{subequations}
and
\begin{subequations}\label{eq:bosonized_wavefunction_barred}
\begin{align}
(\widebar\Psi_1)_{\alpha\beta}(z)
&=
\epsilon_\alpha(\boldsymbol  n)\epsilon_\beta(\boldsymbol  n)\,
z^{-\widebar n_\beta}
\,\textup{e}^{\xi(\widebar{\boldsymbol  t}_\beta,z^{-1})}
\frac{
\tau_{\alpha\beta}
(\boldsymbol  s+\boldsymbol  e_\beta,
 \boldsymbol  r,\boldsymbol  t,
 \widebar{\boldsymbol  t}-[z]_\beta)
}{
\tau_0
},
\label{eq:barPsi1-bos}
\\
(\widebar\Psi_2)_{\alpha\beta}(z)
&=
\epsilon_\alpha(\boldsymbol  n)\epsilon_\beta(\boldsymbol  n)\,
z^{-\widebar n_\beta}
\,\textup{e}^{\xi(\widebar{\boldsymbol  t}_\beta,z^{-1})}
\frac{
\tau_{\beta\alpha}
(\boldsymbol  s,
 \boldsymbol  r-\boldsymbol  e_\beta,
 \boldsymbol  t,
 \widebar{\boldsymbol  t}-[z]_\beta)
}{
\tau_0
},
\label{eq:barPsi2-bos}
\\
(\widebar\Psi_3)_{\alpha\beta}(z)
&=
\epsilon_\alpha(\boldsymbol  n)\epsilon_\beta(\boldsymbol  n)\,
z^{\widebar n_\beta}
\,\textup{e}^{-\xi(\widebar{\boldsymbol  t}_\beta,z^{-1})}
\frac{
\tau_{\alpha\beta}
(\boldsymbol  s,
 \boldsymbol  r+\boldsymbol  e_\beta,
 \boldsymbol  t,
 \widebar{\boldsymbol  t}+[z]_\beta)
}{
\tau_0
},
\label{eq:barPsi1star-bos}
\\
(\widebar\Psi_4)_{\alpha\beta}(z)
&=
\epsilon_\alpha(\boldsymbol  n)\epsilon_\beta(\boldsymbol  n)\,
z^{\widebar n_\beta}
\,\textup{e}^{-\xi(\widebar{\boldsymbol  t}_\beta,z^{-1})}
\frac{
\tau_{\beta\alpha}
(\boldsymbol  s-\boldsymbol  e_\beta,
 \boldsymbol  r,
 \boldsymbol  t,
 \widebar{\boldsymbol  t}+[z]_\beta)
}{
\tau_0
}.
\label{eq:barPsi2star-bos}
\end{align}
\end{subequations}
\end{lemma}

\begin{proof}
We prove the first unbarred~\eqref{eq:Psi1-bos} and barred~\eqref{eq:barPsi1-bos} formulas, the remaining six follow in exactly the same way.
By definition of $\Psi_1$ in~\eqref{eq:Psi1-def} and using the current-fermion relations~\eqref{fermion-current}, we obtain
\begin{equation*}
\begin{aligned}
(\Psi_1)_{\alpha\beta}(z)
&=
\frac{\epsilon_\alpha(\boldsymbol  n)}{\tau_0}
\bra{\boldsymbol  n+\boldsymbol  e_\alpha}
\,\textup{e}^{J(\boldsymbol  t)}
\psi^{(\beta)}(z)
g
\,\textup{e}^{-\widebar J(\widebar{\boldsymbol  t})}
\ket{-\widebar{\boldsymbol  n}} \\
&=\frac{
\epsilon_\alpha(\boldsymbol  n)
\,\textup{e}^{\xi(\boldsymbol  t_\beta,z)}
}{\tau_0}
\bra{\boldsymbol  n+\boldsymbol  e_\alpha}
\psi^{(\beta)}(z)
\,\textup{e}^{J(\boldsymbol  t)}
g
\,\textup{e}^{-\widebar J(\widebar{\boldsymbol  t})}
\ket{-\widebar{\boldsymbol  n}}.
\end{aligned}
\end{equation*}
The left-action bosonization formula~\eqref{bosonization} applied to $\boldsymbol  m=\boldsymbol  n+\boldsymbol  e_\alpha$ reads
\begin{equation*}
\bra{ \boldsymbol  n+\boldsymbol  e_\alpha}\psi^{(\beta)}(z)
=
\epsilon_\beta(\boldsymbol  n+\boldsymbol  e_\alpha)\,
z^{n_\beta+\delta_{\alpha \beta}-1}
\bra{ \boldsymbol  n+\boldsymbol  e_\alpha-\boldsymbol  e_\beta}
\,\textup{e}^{-J([z^{-1}]_\beta).}
\end{equation*}
Combining these identities gives
\begin{align*}
(\Psi_1)_{\alpha\beta}(z)
&=
\frac{
\epsilon_\alpha(\boldsymbol  n)
\epsilon_\beta(\boldsymbol  n+\boldsymbol  e_\alpha)
z^{n_\beta+\delta_{\alpha\beta}-1}
\,\textup{e}^{\xi(\boldsymbol  t_\beta,z)}
}{\tau_0}
\times
\bra{
\boldsymbol  n+\boldsymbol  e_\alpha-\boldsymbol  e_\beta
}
\,\textup{e}^{-J([z^{-1}]_\beta)}
\,\textup{e}^{J(\boldsymbol  t)}
g
\,\textup{e}^{-\widebar J(\widebar{\boldsymbol  t})}
\ket{-\widebar{\boldsymbol  n}}\\
&=\frac{
\epsilon_{\alpha \beta}(\boldsymbol  n)
z^{n_\beta+\delta_{\alpha\beta}-1}
\,\textup{e}^{\xi(\boldsymbol  t_\beta,z)}
}{\tau_0}
\times
\bra{
\boldsymbol  n+\boldsymbol  e_\alpha-\boldsymbol  e_\beta
}
\,\textup{e}^{-J([z^{-1}]_\beta)}
\,\textup{e}^{J(\boldsymbol  t)}
g
\,\textup{e}^{-\widebar J(\widebar{\boldsymbol  t})}
\ket{-\widebar{\boldsymbol  n}}\\
&= \frac{
\epsilon_{\alpha \beta}(\boldsymbol  n)
z^{n_\beta+\delta_{\alpha\beta}-1}
\,\textup{e}^{\xi(\boldsymbol  t_\beta,z)}
}{\tau_0}
\times
\bra{
\boldsymbol  n+\boldsymbol  e_\alpha-\boldsymbol  e_\beta
}
\,\textup{e}^{J(\boldsymbol  t-[z^{-1}]_\beta)}
g
\,\textup{e}^{-\widebar J(\widebar{\boldsymbol  t})}
\ket{-\widebar{\boldsymbol  n}} \\
&= \frac{
\epsilon_{\alpha \beta}(\boldsymbol  n)
z^{n_\beta+\delta_{\alpha\beta}-1}
\,\textup{e}^{\xi(\boldsymbol  t_\beta,z)}
}{\tau_0}
\times
\tau_{\alpha\beta}
(\boldsymbol  s,\boldsymbol  r,
 \boldsymbol  t-[z^{-1}]_\beta,\widebar{\boldsymbol  t}).
\end{align*}
We have used above that 
the positive current modes commute and the fermionic sign identities,
\begin{equation*}
\epsilon_\beta(\boldsymbol  n+\boldsymbol  e_\alpha)=\epsilon_{\alpha\beta}\epsilon_\beta(\boldsymbol  n), \qquad \epsilon_\alpha(\boldsymbol  n)\epsilon_\beta(\boldsymbol  n+\boldsymbol  e_\alpha)=\epsilon_{\alpha\beta}(\boldsymbol  n).
\end{equation*}
For the barred formula~\eqref{eq:barPsi1-bos}, we first apply the
right-action bosonization identity~\eqref{bosonization} to 
$\boldsymbol  m=-\widebar{\boldsymbol  n}$: 
\begin{equation*}
\psi^{(\beta)}(z)\ket{-\widebar{\boldsymbol  n}}
=
\epsilon_\beta(-\widebar{\boldsymbol  n})\,
z^{-\widebar n_\beta}
\,\textup{e}^{\widebar J([z]_\beta)}
\ket{-\widebar{\boldsymbol  n}+\boldsymbol  e_\beta},
\end{equation*}
and using the definition of $\widebar \Psi_1$  in~\eqref{eq:barPsi1-def}, we obtain 
\begin{equation*}
\begin{aligned}
(\widebar\Psi_1)_{\alpha\beta}(z)
&=
\frac{\epsilon_\alpha(\boldsymbol  n)\,\textup{e}^{\xi(\widebar{\boldsymbol  t}_\beta,z^{-1})}}{\tau_0}
\bra{ \boldsymbol  n+\boldsymbol  e_\alpha}
\,\textup{e}^{J(\boldsymbol  t)}g \,\textup{e}^{-\widebar J(\widebar{\boldsymbol  t})}
\psi^{(\beta)}(z)
\ket{-\widebar{\boldsymbol  n}}
\\
&=
\frac{\epsilon_\alpha(\boldsymbol  n)\epsilon_\beta(-\widebar{\boldsymbol  n})
z^{-\widebar n_\beta}\,\textup{e}^{\xi(\widebar{\boldsymbol  t}_\beta,z^{-1})}}{\tau_0}
\bra{ \boldsymbol  n+\boldsymbol  e_\alpha}
\,\textup{e}^{J(\boldsymbol  t)}g \,\textup{e}^{-\widebar J(\widebar{\boldsymbol  t})}
\,\textup{e}^{\widebar J([z]_\beta)}
\ket{-\widebar{\boldsymbol  n}+\boldsymbol  e_\beta}
\\
&=
\frac{\epsilon_\alpha(\boldsymbol  n)\epsilon_\beta(\boldsymbol  n)
z^{-\widebar n_\beta}\,\textup{e}^{\xi(\widebar{\boldsymbol  t}_\beta,z^{-1})}}{\tau_0}
\bra{ \boldsymbol  n+\boldsymbol  e_\alpha}
\,\textup{e}^{J(\boldsymbol  t)}g \,\textup{e}^{-\widebar J(\widebar{\boldsymbol  t}-[z]_\beta)}
\ket{-\widebar{\boldsymbol  n}+\boldsymbol  e_\beta}
\\
&=
\epsilon_\alpha(\boldsymbol  n)\epsilon_\beta(\boldsymbol  n)
z^{-\widebar n_\beta}\,\textup{e}^{\xi(\widebar{\boldsymbol  t}_\beta,z^{-1})}
\frac{\tau_{\alpha\beta}(\boldsymbol  s+\boldsymbol  e_\beta,\boldsymbol  r,\boldsymbol  t,\widebar{\boldsymbol  t}-[z]_\beta)}
{\tau_0}.
\end{aligned}
\end{equation*}
Here we have used $\epsilon_\beta(-\widebar {\boldsymbol {n}})=\epsilon_\beta(\widebar{\boldsymbol  n} )=\epsilon_\beta(\boldsymbol  n)$, 
which follows from $ \boldsymbol  n-\widebar{\boldsymbol  n}=2 \boldsymbol  s$, as well as the commutativity of the negative current modes. This proves~\eqref{eq:barPsi1-bos}. The remaining six formulas follow from the corresponding bosonization identities by the same computation.
\end{proof}

\subsection{Raw dressing operators}
\label{sec:SZ-dressing-operators}
We factor the matrix wave functions into bare spectral factors and matrix pseudodifference operators, defining the raw dressing pair $(X_N,\widebar X_N)$ and determining their asymptotic structure.

Lemma~\ref{lem:bosonized-wave-functions} shows what should be regarded as the
\emph{bare} part of the wave functions.  For each
$\beta=1,\ldots,N$, we define
\begin{equation}
\begin{aligned}
\chi_\beta(z)
&:=
z^{n_\beta}
\,\textup{e}^{\xi(\boldsymbol  t_\beta,z)}
=
z^{s_\beta+r_\beta}
\,\textup{e}^{\xi(\boldsymbol  t_\beta,z)},
&
\chi_\beta^\lor(z)
&:=
z^{-n_\beta}
\,\textup{e}^{-\xi(\boldsymbol  t_\beta,z)}
=
z^{-s_\beta-r_\beta}
\,\textup{e}^{-\xi(\boldsymbol  t_\beta,z)},
\\
\widebar\chi_\beta(z)
&:=
z^{-\widebar n_\beta}
\,\textup{e}^{\xi(\widebar{\boldsymbol  t}_\beta,z^{-1})}
=
z^{s_\beta-r_\beta}
\,\textup{e}^{\xi(\widebar{\boldsymbol  t}_\beta,z^{-1})},
&
\widebar\chi_\beta^\lor(z)
&:=
z^{\widebar n_\beta}
\,\textup{e}^{-\xi(\widebar{\boldsymbol  t}_\beta,z^{-1})}
=
z^{-s_\beta+r_\beta}
\,\textup{e}^{-\xi(\widebar{\boldsymbol  t}_\beta,z^{-1})}.
\end{aligned}
\label{defchis}
\end{equation}
We package these \textit{bare wave functions} into four diagonal $N \times N$ matrices 
\begin{equation*}
\chi
=
\operatorname{diag}(\chi_1,\ldots,\chi_N),
\quad
\chi^\lor
=
\operatorname{diag}(\chi_1^\lor,\ldots,\chi_N^\lor),
 \quad
\widebar\chi
=
\operatorname{diag}(\widebar\chi_1,\ldots,\widebar\chi_N),
\quad
\widebar\chi^\lor
=
\operatorname{diag}
(\widebar\chi_1^\lor,\ldots,\widebar\chi_N^\lor),
\end{equation*}
with which we populate the two diagonal $2N \times 2N$ matrices
\begin{equation}\label{bareframes}
\Xi(z):=
\begin{pmatrix}
\chi(z)&0\\
0&\widebar\chi^\lor(z)
\end{pmatrix},
\qquad
\Xi^\lor(z):=
\begin{pmatrix}
\widebar\chi(z)&0\\
0&\chi^\lor(z)
\end{pmatrix}.
\end{equation}
The bosonized formulas \eqref{eq:bosonized_wavefunction} and~\eqref{eq:bosonized_wavefunction_barred} express the wave functions as their bare spectral factors multiplied by formal series at either $z=\infty$ or $z=0$.  We now interpret these formal series as matrix pseudodifference operators.
By definition~\eqref{shiftdef}, the shift $\mathcal S$ acts on the bare wave functions as follows
\begin{equation}
\mathcal S\chi_\beta=z^{-1}\chi_\beta,
\qquad
\mathcal S\chi_\beta^\lor=z\chi_\beta^\lor,
\qquad
\mathcal S\widebar\chi_\beta=z^{-1}\widebar\chi_\beta,
\qquad
\mathcal S\widebar\chi_\beta^\lor=z\widebar\chi_\beta^\lor.
\label{shiftaction}
\end{equation}
Thus powers of the spectral parameter appearing after the bare factors
have been removed can be translated directly into powers of
$\mathcal S$.

\begin{lemma}
\label{lem:XNbarXN-def}
There exists a unique pair of matrix pseudodifference operators
\begin{equation*}
X_N\in M_{2N}\bigl(\mathcal{V}_{g}[[\mathcal S]]\bigr),
\qquad
\widebar X_N\in M_{2N}\bigl(\mathcal{V}_{g}[[\mathcal S^{-1}]]\bigr),
\end{equation*}
such that
\begin{equation*}
\Phi_{\rm raw}(z)=X_N\Xi(z),
\qquad
\widebar\Phi_{\rm raw}(z)=\widebar X_N\Xi^\lor(z).
\end{equation*}
Following Takasaki's notation and writing them as
\begin{equation}\label{eq:XNbarXN}
X_N=
\begin{pmatrix}
W_1 & \widebar V_1\\
W_2 & \widebar V_2
\end{pmatrix},
\qquad
\widebar X_N=
\begin{pmatrix}
\widebar W_1 & V_1\\
\widebar W_2 & V_2
\end{pmatrix},
\end{equation}
their asymptotic forms are
\begin{equation*}
X_N=
\begin{pmatrix}
I_N+\mathcal O(\mathcal S) & \widebar V_{1,0}+\mathcal O(\mathcal S)\\
W_{2,1}\mathcal S+\mathcal O(\mathcal S^2) & \widebar V_{2,0}+\mathcal O(\mathcal S)
\end{pmatrix},
\qquad
\widebar X_N=
\begin{pmatrix}
\widebar W_{1,0}+\mathcal O(\mathcal S^{-1})
&
V_{1,-1}\mathcal S^{-1}+\mathcal O(\mathcal S^{-2})
\\
\widebar W_{2,0}+\mathcal O(\mathcal S^{-1})
&
I_N+\mathcal O(\mathcal S^{-1})
\end{pmatrix}.
\end{equation*}
Moreover,
\begin{equation*}
(W_{2,1})_{\alpha\alpha}
=
(V_{1,-1})_{\alpha\alpha}
=
0,
\qquad
\alpha=1,\ldots,N.
\end{equation*}
\end{lemma}
\begin{proof}
Note that, by the definition of $\mathcal V_g$  in~\eqref{eq:diff_diff_algebra_Clifford}, all coefficients
occurring in the wave functions of Lemma~\ref{lem:bosonized-wave-functions}, including the
fermionic sign factors and the coefficients of the tau-function
quotients, belong to $\mathcal V_g$. We consider first $\Psi_1$ and $\Psi_2$  in~\eqref{eq:bosonized_wavefunction}. By Lemma~\ref{lem:bosonized-wave-functions}, after factoring
$\chi_\beta$ from the $\beta$-th column, the remaining factors are formal
power series in $z^{-1}$. Indeed,
\begin{equation*}
\tau\bigl(\boldsymbol  t\pm[z^{-1}]_\beta\bigr)
=
\exp\left(
\pm\sum_{k\geq1}\frac{z^{-k}}{k}
\frac{\partial}{\partial t_{\beta,k}}
\right)\tau(\boldsymbol  t).
\end{equation*}
Since $
\mathcal S^k\chi_\beta=z^{-k}\chi_\beta,
$
each formal series
$
\sum_{k\geq0}A_kz^{-k}
$
determines uniquely a matrix pseudodifference operator
$
\sum_{k\geq0}A_k\mathcal S^k
$
having the same action on $\chi_\beta$. This determines uniquely the
blocks $W_1$ and $W_2$ of $X_N$.
The leading powers in the bosonized formulas determine their
asymptotic forms. After division by $\chi_\beta$, the diagonal entries of $\Psi_1$ are
$1+\mathcal O(z^{-1})$, while the off-diagonal entries are $\mathcal O(z^{-1})$.
For $\Psi_2$, the diagonal entries are $\mathcal O(z^{-2})$, while the off-diagonal
entries are $\mathcal O(z^{-1})$. Hence
\begin{equation*}
W_1=I_N+\mathcal O(\mathcal S),
\qquad
W_2=W_{2,1}\mathcal S+\mathcal O(\mathcal S^2),
\qquad
(W_{2,1})_{\alpha\alpha}=0.
\end{equation*}

The wave functions $\Psi_3$ and $\Psi_4$  in~\eqref{eq:bosonized_wavefunction} are treated in the same way.
After factoring $\chi^\lor_\beta$, their remaining factors are formal
power series in $z^{-1}$, while
\begin{equation*}
\mathcal S^{-k}\chi^\lor_\beta=z^{-k}\chi^\lor_\beta.
\end{equation*}
Therefore they determine uniquely the blocks $V_1$ and $V_2$. The
leading powers in~\eqref{eq:Psi1star-bos}-\eqref{eq:Psi2star-bos}, together with the expressions of $\chi_\beta, \chi_\beta^\lor, \widebar \chi_\beta, \widebar \chi_\beta^\lor$ in~\eqref{defchis}, give
\begin{equation*}
V_1=V_{1,-1}\mathcal S^{-1}+\mathcal O(\mathcal S^{-2}),
\qquad
V_2=I_N+\mathcal O(\mathcal S^{-1}),
\qquad
(V_{1,-1})_{\alpha\alpha}=0.
\end{equation*}

For the barred wave functions in~\eqref{eq:bosonized_wavefunction_barred}, the relevant expansions are instead
taken at $z=0$. After factoring $\widebar\chi_\beta$ from
$\widebar\Psi_1,\widebar\Psi_2$ and $\widebar\chi^\lor_\beta$ from
$\widebar\Psi_3,\widebar\Psi_4$, the remaining tau-quotients are formal power
series in $z$. Since
\begin{equation*}
\mathcal S^{-k}\widebar\chi_\beta=z^k\widebar\chi_\beta,
\qquad
\mathcal S^k\widebar\chi^\lor_\beta=z^k\widebar\chi^\lor_\beta,
\end{equation*}
these expansions determine uniquely the remaining blocks
$\widebar W_1,\widebar W_2,\widebar V_1,\widebar V_2$. Their constant terms give
\begin{equation*}
\widebar W_i=\widebar W_{i,0}+\mathcal O(\mathcal S^{-1}),
\qquad i=1,2,
\end{equation*}
and
\begin{equation*}
\widebar V_1=\widebar V_{1,0}+\mathcal O(\mathcal S),
\qquad
\widebar V_2=\widebar V_{2,0}+\mathcal O(\mathcal S).
\end{equation*}

Thus the corresponding block matrices $X_N$ and $\widebar X_N$ are uniquely
determined and have the asserted asymptotic forms.
\end{proof}

When $N=1$, the lemma forces
$
W_{2,1}=V_{1,-1}=0.
$
Consequently,
$
W_2=\mathcal O(\mathcal S^2),
V_1=\mathcal O(\mathcal S^{-2}),
$
as in Takasaki's one-component Pfaff--Toda dressing formalism \cite{Taka}.

\subsection{Algebraic consequences of the matrix residue identity}
\label{sec:bilinear-algebraic-relations}
We translate the matrix residue identity into an algebraic relation between the raw dressing operators, showing that $X_NJ_N\widebar X_N^*=\widebar X_NJ_NX_N^*=:K_N$, and compute $K_N$ explicitly.

The matrix residue identity~\eqref{eq:raw-matrix-residue} is formulated in the spectral
variable $z$, whereas the dressing data of Section~\ref{sec:SZ-dressing-operators} are pseudodifference operators in the shift
$\mathcal S$. The link between the two is provided by the
following residue--coefficient correspondence.

\begin{lemma}[Matrix residue--coefficient correspondence]
\label{lem:matrix-residue}
Let $m,\ell\in\mathbb Z$, and suppress the common dependence on
$\boldsymbol  r,\boldsymbol  t,\widebar{\boldsymbol  t}$, and vary only the discrete variable $\boldsymbol s$. For a matrix pseudodifference operator~$A$, we denote by $(A)_m$ the coefficient of $\mathcal S^m$.
\begin{enumerate}
    \item 
Suppose that, for some $p,q\in\mathbb Z$,
\begin{equation*}
P=\sum_{i\ge p}P_i\mathcal S^i,
\qquad
Q=\sum_{j\le q}Q_j\mathcal S^j,
\qquad
P_i,Q_j\in M_N(\mathcal V_g).
\end{equation*}
Then
\begin{align}
\label{eq:res-chi-chivee}
\operatorname{res}_z
z^\ell
(P\chi)(\boldsymbol  s)
(Q\chi^\lor)(\boldsymbol  s-m\boldsymbol  1)^\top
\frac{dz}{z}&=
\bigl(P\mathcal S^{-\ell}Q^*\bigr)_m(\boldsymbol  s), \\
\label{eq:res-barchivee-barchi}
\operatorname{res}_z
z^\ell
(P\widebar\chi^\lor)(\boldsymbol  s)
(Q\widebar\chi)(\boldsymbol  s-m\boldsymbol  1)^\top
\frac{dz}{z}
&=
\bigl(P\mathcal S^\ell Q^*\bigr)_m(\boldsymbol  s).
\end{align}
\item
Suppose that, for some $p,q\in\mathbb Z$,
\begin{equation*}
P=\sum_{i\le p}P_i\mathcal S^i,
\qquad
Q=\sum_{j\ge q}Q_j\mathcal S^j,
\qquad
P_i,Q_j\in M_N(\mathcal V_g).
\end{equation*}
Then
\begin{align}
\label{eq:res-chivee-chi}
\operatorname{res}_z
z^\ell
(P\chi^\lor)(\boldsymbol  s)
(Q\chi)(\boldsymbol  s-m\boldsymbol  1)^\top
\frac{dz}{z}
&=
\bigl(P\mathcal S^\ell Q^*\bigr)_m(\boldsymbol  s),
\\
\label{eq:res-barchi-barchivee}
\operatorname{res}_z
z^\ell
(P\widebar\chi)(\boldsymbol  s)
(Q\widebar\chi^\lor)(\boldsymbol  s-m\boldsymbol  1)^\top
\frac{dz}{z}
&=
\bigl(P\mathcal S^{-\ell}Q^*\bigr)_m(\boldsymbol  s).
\end{align}
\end{enumerate}
\end{lemma}

\begin{proof}
It suffices to prove~\eqref{eq:res-chi-chivee} for monomials
\begin{equation*}
P=A \mathcal{S}^a, \qquad Q=B \mathcal{S}^b,\qquad
A,B \in M_N(\mathcal V_g),
\end{equation*}
then the general case follows by bilinearity. The support assumptions ensure that, for each fixed $m$, only finitely many pairs of monomials contribute.
Using~\eqref{defchis} and~\eqref{shiftaction}, in particular 
\begin{equation*}
\chi(\boldsymbol  s)\chi^\lor(\boldsymbol  s-m\boldsymbol  1)=z^m I_N,
\end{equation*}
we get
\begin{equation*}
    \begin{aligned}
     \text{LHS of~\eqref{eq:res-chi-chivee}} &= \operatorname{res}_z
z^\ell
(A \mathcal{S}^a\chi)(\boldsymbol  s)
(B \mathcal{S}^b\chi^\lor)(\boldsymbol  s-m\boldsymbol  1)^\top
\frac{dz}{z}
  \\
  &= \operatorname{res}_z
z^{\ell-a+b}
A(\boldsymbol  s)\chi(\boldsymbol  s)
\chi^\lor(\boldsymbol  s-m\boldsymbol  1)B^\top(\boldsymbol  s-m\boldsymbol  1)
\frac{dz}{z}\\
 &= \operatorname{res}_z
z^{\ell-a+b+m-1}
A(\boldsymbol  s)B^\top(\boldsymbol  s-m\boldsymbol  1)
dz\\[1mm]
&= \delta_{\ell+b+m,a}
A(\boldsymbol  s)B^\top(\boldsymbol  s-m\boldsymbol  1) = \delta_{\ell+b+m,a}
(A \mathcal{S}^{m}(B^\top))(\boldsymbol  s) \\[1mm]
&= \delta_{\ell+b+m,a}
(A \mathcal{S}^{a-\ell-b}(B^\top))(\boldsymbol  s)
=
(P \mathcal{S}^{-\ell}Q^*)_m(\boldsymbol  s)= \text{RHS of~\eqref{eq:res-chi-chivee}}.
    \end{aligned}
\end{equation*}
The remaining identities~\eqref{eq:res-barchivee-barchi}--\eqref{eq:res-barchi-barchivee} follow by the same
monomial computation, using the corresponding shift relations in
\eqref{shiftaction}.
\end{proof}
We apply Lemma~\ref{lem:matrix-residue} to the matrix residue identity~\eqref{eq:raw-matrix-residue}.

\begin{proposition}
\label{prop:raw-pfaff-relation}
The pair of raw dressing operators $(X_N,\widebar X_N)$ in~\eqref{eq:XNbarXN} satisfies the relation
\begin{equation}
X_NJ_N\widebar X_N^*
=
\widebar X_NJ_NX_N^*,
\label{eq:raw-pfaff-relation}
\end{equation}
where $J_N$ is the shifted symplectic matrix~\eqref{eq:JN}.

\end{proposition}

\begin{proof}
We specialize the matrix residue identity~\eqref{eq:raw-matrix-residue} at $(\boldsymbol  s',\boldsymbol  r', \boldsymbol  t',\widebar{\boldsymbol  t}')=(\boldsymbol  s-m\boldsymbol  1, \boldsymbol  r, \boldsymbol  t, \widebar{\boldsymbol  t})$, $m \in \mathbb Z$. Keeping the convention of Lemma~\ref{lem:matrix-residue}, we get
\begin{equation*}
\operatorname{res}_{z}
\Phi_{\mathrm{raw}}(\boldsymbol  s)\,
J_0\,
(\widebar\Phi_{\mathrm{raw}}(\boldsymbol  s-m\boldsymbol  1))^\top\,dz
=
\operatorname{res}_{z}
\widebar\Phi_{\mathrm{raw}}(\boldsymbol  s)\,
J_0\,
(\Phi_{\mathrm{raw}}(\boldsymbol  s-m\boldsymbol  1))^\top\,dz.
\label{eq:raw-matrix-residue-specialized}
\end{equation*}
By Lemma~\ref{lem:XNbarXN-def}, we can rewrite it in terms of the raw dressing operators
\label{help}
\begin{equation}
\operatorname*{res}_{z}
(X_N \Xi(z))(\boldsymbol  s)\,
J_0\,
(\widebar X_N \Xi^\lor(z))^\top(\boldsymbol  s-m\boldsymbol  1)\,dz
=
\operatorname*{res}_{z}
(\widebar X_N \Xi^\lor(z))(\boldsymbol  s)
J_0\,
( X_N \Xi(z))^\top(\boldsymbol  s-m\boldsymbol  1)\,dz.
\label{eq:raw-matrix-residue-specialized1}
\end{equation}
We consider the $(1,2)$-block of the LHS in~\eqref{eq:raw-matrix-residue-specialized1}, i.e.\ 
\begin{equation*}
\operatorname{res}_z 
\left( ( W_1 \chi)(\boldsymbol  s)(V_2 \chi^\lor)^\top(\boldsymbol  s-m \boldsymbol  1)-(\widebar V_1 \widebar \chi^\lor)(\boldsymbol  s)(\widebar W_2 \widebar \chi)^{\top}(\boldsymbol  s-m \boldsymbol  1)
\right).
\end{equation*}
Since $dz=z\,dz/z$, the first and second identities of
Lemma~\ref{lem:matrix-residue} with $\ell=1$ give
\begin{equation*}
\text{LHS of~\eqref{eq:raw-matrix-residue-specialized1}}_{12}=(W_1 \mathcal{S}^{-1}V_2^*-\widebar V_1 \mathcal{S}\widebar W_2^*)_m.
\end{equation*}
Similarly, we compute 
\begin{equation*}
\begin{aligned}
    \text{RHS of~\eqref{eq:raw-matrix-residue-specialized1}}_{12}&= \operatorname{res}_z 
\left( ( \widebar W_1 \widebar \chi)(\boldsymbol  s)(\widebar V_2 \widebar \chi^\lor)^\top(\boldsymbol  s-m \boldsymbol  1)-(V_1 \chi^\lor)(\boldsymbol  s)( W_2 \chi)^\top(\boldsymbol  s-m \boldsymbol  1)
\right) \\
&=(\widebar W_1 \mathcal{S}^{-1}\widebar V_2^*- V_1 \mathcal{S} W_2^*)_m.
\end{aligned}
\end{equation*}
Hence we have shown that 
\begin{equation*}
W_1 \mathcal{S}^{-1}V_2^*-\widebar V_1 \mathcal{S}\widebar W_2^*=\widebar W_1 \mathcal{S}^{-1}\widebar V_2^*- V_1 \mathcal{S} W_2^*.
\end{equation*}
This is exactly the $(1,2)$-block of~\eqref{eq:raw-pfaff-relation}. The three other blocks are proved in the same way.
\end{proof}
We denote the common value in Proposition~\ref{prop:raw-pfaff-relation} by
\begin{equation}
K_N:=X_NJ_N\widebar X_N^*
      =\widebar X_NJ_NX_N^*,
\label{eq:KN}
\end{equation}
and since \(J_N^*=-J_N\), taking formal adjoints gives
\begin{equation*}
K_N^*=-K_N.
\end{equation*}
Thus the fermionic bilinear identity determines a skew operator-valued form relating the two raw dressing operators. However, this skew operator is not yet $J_N$.

\begin{lemma}
\label{lem:KN-def}
The matrices $W_{2,1}$ and $V_{1,-1}$ are skew-symmetric. Moreover,
\begin{equation}
K_N=
\begin{pmatrix}
-V_{1,-1} & \mathcal S^{-1}I_N\\[2mm]
-\mathcal SI_N & W_{2,1}
\end{pmatrix}.
\label{computationKN}
\end{equation}
\end{lemma}
\begin{proof}
We use the two expressions in~\eqref{eq:KN} to bound the Laurent support of each
block of $K_N$ from both sides.
From $K_N=X_NJ_N\widebar X_N^*$ and the asymptotic forms of Lemma~\ref{lem:XNbarXN-def}, the
$(1,1)$- and $(2,2)$-blocks are supported in non-negative powers of
$\mathcal S$, while the $(1,2)$-block contains only powers
$\mathcal S^k$, $k\geq -1$. Their boundary coefficients are
\begin{equation*}
(K_N)_{11,0}=V_{1,-1}^{\top},
\qquad
(K_N)_{12,-1}=I_N,
\qquad
(K_N)_{22,0}=W_{2,1}.
\end{equation*}
On the other hand, using
$K_N=\widebar X_NJ_NX_N^*$, the $(1,1)$- and $(2,2)$-blocks are supported in
non-positive powers of $\mathcal S$, while the $(1,2)$-block contains only powers $\mathcal S^k$, $k\leq -1$. Moreover, the degree-zero
coefficient of the $(1,1)$-block is $-V_{1,-1}$. Hence
\begin{equation*}
(K_N)_{11}=-V_{1,-1},
\qquad
(K_N)_{12}=\mathcal S^{-1}I_N,
\qquad
(K_N)_{22}=W_{2,1}.
\end{equation*}
Finally, since $K_N^*=-K_N$, we obtain
\begin{equation*}
(K_N)_{21}=-\mathcal S I_N,
\qquad
V_{1,-1}^{\top}=-V_{1,-1},
\qquad
W_{2,1}^{\top}=-W_{2,1},
\end{equation*}
which proves the claim.
\end{proof}

\subsection{Regular factorization locus}
We isolate the finite-dimensional factorization problem governing the normalization of the raw dressing operators, define the regular factorization locus, and establish its non-emptiness for $N=1$ and for a family of coupled Clifford elements when $N \ge 2$. 
 
Let $B_-^\times \subset \operatorname{GL}_{N}(\mathcal V_g)$ be the subgroup of lower triangular invertible matrices and let $H \subset \operatorname{GL}_{2N}(\mathcal V_g)$ be the subgroup
\begin{equation*}
H:=
\left\{
\begin{pmatrix}
I_N+a&b\\
c&d
\end{pmatrix}
\;\middle|\;
a,b,c\in\mathfrak n_-,
\quad
d\in B_-^\times
\right\}.
\label{eq:H-group}
\end{equation*}
The closure under multiplication is immediate, while the Schur-complement formula shows that the inverse of an element of $H$ has the same triangular form.
We also introduce the symplectic group associated with $J_0$,
\begin{equation*}
\operatorname{Sp}_{2N}
:=\operatorname{Sp}_{2N}(\mathcal V_{g})
=\left\{
F\in GL_{2N}(\mathcal V_{g})
\;\middle|\;
F J_0\,F^\top=J_0
\right\}.
\label{eq:Sp-group}
\end{equation*}
The two groups have trivial intersection. Indeed, if
\begin{equation*}
h=\begin{pmatrix}I_N+a&b\\ c&d\end{pmatrix}\in H\cap \operatorname{Sp}_{2N},
\end{equation*}
then
\begin{equation*}
h^{-1}=-J_0h^\top J_0
=
\begin{pmatrix}
d^\top&-b^\top\\
-c^\top&(I_N+a)^\top
\end{pmatrix}.
\end{equation*}
Since $h^{-1}\in H$, comparison of upper and lower triangular parts forces
$a=b=c=0$ and $d=I_N$. Consequently, whenever a factorization in
$\operatorname{Sp}_{2N}(\mathcal V_g)H$ exists, it is unique.
We define
\begin{equation*}
Q_N:=
\begin{pmatrix}
I_N & 0\\
0 & \mathcal S^{-1}I_N
\end{pmatrix},
\end{equation*}
then $Q_N^*=Q_N^{-1}$, and
$Q_NJ_NQ_N^{-1}=J_0$.
Moreover, by Lemmas~\ref{lem:XNbarXN-def} and~\ref{lem:KN-def},
\begin{equation}
Q_NK_NQ_N^{-1}=\Omega,
\qquad
Q_NX_NQ_N^{-1}=A_0+\mathcal O(\mathcal S),
\label{eq:Omega-A0}
\end{equation}
where
\begin{equation*}
\Omega=
\begin{pmatrix}
-V_{1,-1} & I_N\\
-I_N & \mathcal S^{-1}(W_{2,1})
\end{pmatrix},
\qquad
A_0=
\begin{pmatrix}
I_N & 0\\
\mathcal S^{-1}(W_{2,1}) &
\mathcal S^{-1}(\widebar V_{2,0})
\end{pmatrix}.
\end{equation*}
\begin{definition}
\label{def:regular-factorization-locus}
We say that an admissible Clifford-group element $g$ belongs to the
\emph{regular factorization locus} if there exists a degree-zero matrix $R\in \operatorname{GL}_{2N}(\mathcal V_g)$
such that
\begin{equation*}
R\,\Omega\, R^\top=J_0,
\qquad
R\,A_0\in\operatorname{Sp}_{2N}(\mathcal V_g)H.
\end{equation*}
\end{definition}
In other words, regularity asks for a single degree-zero transformation which
simultaneously normalizes the skew form $\Omega$ and places the leading
coefficient $A_0$ in the factorization domain
$\operatorname{Sp}_{2N}(\mathcal V_g)H$. The regular factorization condition is automatic in the one-component
case, whereas for $N\geq2$, it is satisfied by a family of genuinely coupled Clifford-group elements.

\begin{lemma}
\label{lem:regularization-N1}
    Let $N=1$. Then every admissible Clifford-group element $g$ belongs to the regular factorization locus.
\end{lemma}
\begin{proof}
By Lemma~\ref{lem:XNbarXN-def}, $V_{1,-1}=W_{2,1}=0,$
and hence $
\Omega=J_0.
$
Moreover, by~\eqref{eq:barPsi2star-bos}, together with the bare factorization of
Lemma~\ref{lem:XNbarXN-def}, the constant term of
$\widebar\Psi_4(\widebar\chi^\lor)^{-1}$ is
$\widebar V_{2,0}=\mathcal S(\tau_0)/\tau_0$.
Therefore, by definition of the algebra $\mathcal V_g$  in~\eqref{eq:diff_diff_algebra_Clifford},
\begin{equation*}
\mathcal S^{-1}(\widebar V_{2,0})
=
\frac{\tau_0}{\mathcal S^{-1}(\tau_0)}
\in\mathcal V_g^\times.
\end{equation*}
It follows that
\begin{equation*}
A_0
=
\begin{pmatrix}
1&0\\
0&\mathcal S^{-1}(\widebar V_{2,0})
\end{pmatrix}
\in H.
\end{equation*}
Thus $R=I_2$ satisfies both conditions of Definition~~\ref{def:regular-factorization-locus}.
\end{proof}
\begin{lemma}
\label{lem:gB_cliff}
Let $N\geq 2$ and let $B=(B_{\alpha\beta})\in\mathfrak{so}_N(\mathbb C)$. The element
\begin{equation*}
g_B
:=
\exp\!\left(
\frac12\sum_{\alpha,\beta=1}^N
B_{\alpha\beta}\, \psi^{(\alpha)}_0\psi^{(\beta)}_0
\right).
\end{equation*}
belongs to the regular factorization locus.
\end{lemma}

\begin{proof}
This is a particular case of the family considered in Remark~\ref{rem:admissible-cliff-group},
and hence $g_B$ is admissible.
We first note that the localization $\mathcal V_{g_B}$ is
determined entirely by the sector $\boldsymbol r=\boldsymbol 0$. Since the fermionic modes
$\psi^{(\alpha)}_0$ anticommute, each of them occurs at most once in
every non-zero monomial in the expansion of $g_B$. Consequently, every
charge component of $g_B$ has charge
\begin{equation}
\boldsymbol  q=(q_1,\ldots,q_N),
\qquad
q_\alpha\in\{0,1\}.\label{observation}
\end{equation}
In particular, the only charge component of $g_B$ belonging to
$2\mathbb Z^N$ is the charge-zero component.
On the other hand,
\begin{equation*}
\tau_0 (\boldsymbol  s,\boldsymbol  r,\boldsymbol  t,\widebar{\boldsymbol  t})
=
\bra{\boldsymbol  s+\boldsymbol  r}\textup{e}^{J(\boldsymbol  t)}
g_B
\,\textup{e}^{-\widebar J(\widebar{\boldsymbol  t})}
\ket{\boldsymbol  s-\boldsymbol  r}.
\end{equation*}
Since the current operators have charge zero, by
\eqref{chargeselection} this matrix element can be non-zero only if
$g_B$ has a charge component equal to $2\boldsymbol  r$. By~\eqref{observation}, this is possible only for $\boldsymbol r=\boldsymbol 0$. Therefore
\begin{equation*}
\tau_0(\boldsymbol  s,\boldsymbol  r,\boldsymbol  t,\widebar{\boldsymbol  t})=0
\qquad
\text{for }\boldsymbol  r\neq  \boldsymbol 0.
\end{equation*}
Since $\tau_0$ vanishes for $\boldsymbol  r\neq\boldsymbol 0$, the localization locus of $\mathcal V_{g_B}$ is
contained in the sector $\boldsymbol  r=\boldsymbol 0$. Hence two elements
$F,G\in\mathcal A_{g_B}$ which agree on $\boldsymbol  r=\boldsymbol 0$ define the same
element of $\mathcal V_{g_B}$. Indeed, if $F=G$ on $\boldsymbol r=\boldsymbol 0$, then
$\tau_0(F-G)=0$ in $\mathcal A_{g_B}$, and hence
$F=G$ in $\mathcal V_{g_B}$.

Looking at the leading term of~\eqref{eq:Psi2-bos} and removing the bare
factor $\chi_\beta$ by Lemma~\ref{lem:XNbarXN-def}, we have for $\alpha\neq\beta$
\begin{equation}\label{eq:W21_regular-factor}
(W_{2,1})_{\alpha\beta}
=
\epsilon_{\alpha\beta}(\boldsymbol  n)
\frac{
\tau_{\alpha\beta}
(\boldsymbol  s-\boldsymbol  e_\alpha,
 \boldsymbol  r-\boldsymbol  e_\alpha,
 \boldsymbol  t,\widebar{\boldsymbol  t})
}{
\tau_0
},
\end{equation}
while $(W_{2,1})_{\alpha\alpha}=0$ by Lemma~\ref{lem:KN-def}. On the sector
$\boldsymbol r=\boldsymbol 0$, the numerator \eqref{eq:W21_regular-factor} is
\begin{equation*}
\bra{\boldsymbol  s-\boldsymbol  e_\alpha-\boldsymbol  e_\beta}
\,\textup{e}^{J(\boldsymbol  t)}
g_B
\,\textup{e}^{-\widebar J(\widebar{\boldsymbol  t})}
\ket{
\boldsymbol  s
}.
\end{equation*}
Its bra--ket charge difference is
$-\boldsymbol  e_\alpha-\boldsymbol  e_\beta$. All charge components of $g_B$
have non-negative coordinates, so this matrix element vanishes by
\eqref{chargeselection}. By the preceding localization argument,
\begin{equation*}
W_{2,1}=0 \text{   in }\mathcal V_{g_B}.
\end{equation*}
Similarly, the leading term of~\eqref{eq:barPsi2star-bos}, together
with Lemma~\ref{lem:XNbarXN-def}, gives
\begin{equation}\label{eq:V20_regular-fact}
(\widebar V_{2,0})_{\alpha\beta}
=
\epsilon_\alpha(\boldsymbol  n)\epsilon_\beta(\boldsymbol  n)
\frac{
\tau_{\beta\alpha}
(\boldsymbol  s-\boldsymbol  e_\beta,
 \boldsymbol  r,\boldsymbol  t,\widebar{\boldsymbol  t})
}{
\tau_0
}.
\end{equation}
For $\boldsymbol r=\boldsymbol 0$, the numerator in~\eqref{eq:V20_regular-fact}  is
\begin{equation*}
\bra{
\boldsymbol  s-\boldsymbol  e_\alpha}
\,\textup{e}^{J(\boldsymbol  t)}
g_B
\,\textup{e}^{-\widebar J(\widebar{\boldsymbol  t})}
\ket{
\boldsymbol  s-\boldsymbol  e_\beta
}.
\end{equation*}
If $\alpha\neq\beta$, its bra--ket charge difference is
$\boldsymbol  e_\beta-\boldsymbol  e_\alpha$, which has a negative component.
Since every charge component of $g_B$ has non-negative coordinates,
the matrix element vanishes by~\eqref{chargeselection}. If
$\alpha=\beta$, the charge difference is zero, so only the
charge-zero component of $g_B$, namely its constant term, can
contribute. By~\eqref{charge independence}, the resulting vacuum
expectation is independent of $\boldsymbol  s$, and therefore in $\mathcal V_{g_B}$,$(\widebar V_{2,0})_{\alpha\alpha}=1$. Consequently,
\begin{equation*}
\widebar V_{2,0}=I_N,
\qquad
\Omega=
\begin{pmatrix}
-V_{1,-1}&I_N\\
-I_N&0
\end{pmatrix},
\qquad
A_0=I_{2N}.
\end{equation*}
By Lemma~\ref{lem:KN-def}, $V_{1,-1}$ is skew-symmetric, hence there exists a
unique $b\in\mathfrak n_-$ such that
$b^\top-b=V_{1,-1}$.
Introducing
\begin{equation*}
R=
\begin{pmatrix}
I_N&b\\
0&I_N
\end{pmatrix},
\end{equation*}
we have that $R\in H$, and
\begin{equation*}
R\,\Omega\, R^\top
=
\begin{pmatrix}
-V_{1,-1}-b+b^\top &I_N\\
-I_N&0
\end{pmatrix}
=
J_0.
\end{equation*}
Moreover,
\begin{equation*}
R\,A_0=R\in
H\subset
\operatorname{Sp}_{2N}(\mathcal V_{g_B})H.
\end{equation*}
As $B$ varies in $\mathfrak{so}_N(\mathbb C)$, the elements $g_B$ form a $\binom N2$-parameter family in the regular factorization locus containing genuinely coupled Clifford-group elements.
\end{proof}
We do not attempt here to characterize the full regular factorization locus among the Clifford-group tau-functions of \cite{SavchenkoZabrodin}. Its size inside the admissible locus is left for future work.

\subsection{Normalization of the raw dressing operators}
\label{sec:normalization}
Here, we construct a unique finite-Laurent matrix operator $\Gamma_N$ that normalizes the raw dressing pair to satisfy the canonical Pfaff relation $Y_NJ_N\widebar Y_N^*=J_N$ with $Y_N \in G_-$.

Henceforth, we assume that the Clifford-group element $g$ belongs to the regular factorization locus of Definition~\ref{def:regular-factorization-locus}.
Applying the construction of Section~\ref{sec:multicomponent-hierarchy} with coefficient algebra
$\mathcal V_g$, we use the same notation 
$\mathfrak g=\mathfrak g_{+} \oplus \mathfrak g_{-}$ and $\Pi_{\pm},$
for the resulting Pfaff--Toda splitting and projections respectively. We denote by $G_-$ the group of invertible matrix pseudodifference operators of the
form
\begin{equation*}
U=
\begin{pmatrix}
I_N+a+\mathcal O(\mathcal S)
&
b\mathcal S^{-1}+\mathcal O(1)
\\
c\mathcal S+\mathcal O(\mathcal S^2)
&
d+\mathcal O(\mathcal S)
\end{pmatrix},
\qquad
a,b,c\in\mathfrak n_-,
\quad
d\in B_-^\times.
\end{equation*}
The closure under multiplication and inversion follows from the triangular leading term and the $ \mathcal S$-adic inversion argument used in Lemma~\ref{lem:UN_invertible}.
Its Lie algebra is $\mathfrak g_-$. Note that 
\begin{equation}
U\in G_-
\Longleftrightarrow
Q_NUQ_N^{-1}=h+\mathcal O(\mathcal S)
\quad\text{for some }h\in H.
\label{characGmin}
\end{equation}
The raw dressing operator $X_N$ in~\eqref{eq:XNbarXN} does not in general satisfy the triangular conditions defining $G_-$. Moreover, $K_N$ does not coincide with $J_N$ in general.
In order to compute the Sato evolutions and identify the fermionic dressing data with the dressing representation of Section~\ref{sec:dressing_zero_curvature}, we therefore seek a transformation
$
(X_N,\widebar X_N)
\longmapsto
(\Gamma_NX_N,\Gamma_N\widebar X_N)
$
such that
\begin{equation}
\Gamma_NK_N\Gamma_N^*=J_N
,\qquad
\Gamma_NX_N\in G_-.
\label{eq:Gamma-negative}
\end{equation}

\begin{lemma}
\label{lem:XNinvertibility}
The raw dressing operators
$X_N$ and $\widebar X_N$ in~\eqref{eq:XNbarXN} are invertible in their respective one-sided
completions. 
\end{lemma}
\begin{proof}
Let $R$ be as in Definition~\ref{def:regular-factorization-locus}. Since $R$ and $RA_0$ are
invertible, the matrix $A_0$ is invertible.
Hence, by~\eqref{eq:Omega-A0}, $Q_NX_NQ_N^{-1}=A_0+\mathcal O(\mathcal S)$
is invertible in
$M_{2N}(\mathcal V_g[[\mathcal S]])$, and therefore $X_N$ is
invertible. Similarly, the relation
$R\,\Omega\, R^\top=J_0$
implies that $\Omega$, and hence $K_N$, is invertible. From the definition~\eqref{eq:KN}, $K_N=X_NJ_N\widebar X_N^*$,
it follows that $\widebar X_N^*$, and therefore $\widebar X_N$, is invertible
in the opposite completion.
\end{proof}
\begin{proposition}
\label{prop:GammaN}
    There exists a unique matrix difference
operator
$
\Gamma_N\in M_{2N}\bigl(V_g[\mathcal S,\mathcal S^{-1}]\bigr)
$ such that~\eqref{eq:Gamma-negative} holds.
\end{proposition}
\begin{proof}
Let $R$ be as in Definition~\ref{def:regular-factorization-locus}. Since
$RA_0\in \operatorname{Sp}_{2N}H$,
there exist $f\in \operatorname{Sp}_{2N}$ and $h\in H$ such that $RA_0=f h$.
Set
$
F:=f^{-1}R.
$
By~\eqref{eq:Omega-A0}, we get
\begin{equation*}
F\,\Omega\, F^\top
=
f^{-1}J_0(f^{-1})^\top
=
J_0,
\qquad
F\,A_0=h\in H.
\end{equation*}
Finally, define 
\begin{equation}\label{eq:GammaN}
\Gamma_N:=Q_N^{-1}FQ_N.
\end{equation}
Since $F$ has degree zero and $Q_N$ is a diagonal monomial
difference operator, $\Gamma_N$ \eqref{eq:GammaN} has finite Laurent
support, hence
$\Gamma_N\in
M_{2N}\bigl(\mathcal V_g[\mathcal S,\mathcal S^{-1}]\bigr)$. 
It is clear that
\begin{equation*}
\Gamma_NK_N\Gamma_N^*
=
Q_N^{-1}F\Omega F^\top Q_N
=
Q_N^{-1}J_0Q_N
=
J_N.
\end{equation*}
On the other hand,
\begin{equation*}
Q_N\Gamma_NX_NQ_N^{-1}
=
F\bigl(A_0+\mathcal O(\mathcal S)\bigr)
=
h+\mathcal O(\mathcal S),
\qquad h\in H,
\end{equation*}
which implies
$
\Gamma_NX_N\in G_-
$ as noted in~\eqref{characGmin}. Thus $\Gamma_N$ as in~\eqref{eq:GammaN} satisfies~\eqref{eq:Gamma-negative}.

It remains to prove uniqueness. Suppose that $\Gamma_N'$ is another matrix difference operator satisfying~\eqref{eq:Gamma-negative}, and set $T=\Gamma_N' \Gamma_N^{-1}$. By construction we have $T=(\Gamma_N'X_N)(\Gamma_NX_N)^{-1} \in G_-$ and
\begin{equation*}
TJ_NT^*=J_N.
\end{equation*}
Conjugating $T$ by $Q_N$, and setting $\widehat T:=Q_NTQ_N^{-1}$, we have 
\begin{equation*}
\widehat T=h_0+\mathcal O(\mathcal S),
\qquad h_0\in H,
\end{equation*}
using~\eqref{characGmin} and since $T\in G_-$.  
Moreover, $\widehat TJ_0\widehat T^*=J_0$,
and therefore
\begin{equation}
\widehat T^*=-J_0\widehat T^{-1}J_0.
\label{auxeq}
\end{equation}
The left-hand side of~\eqref{auxeq} is a series in non-positive powers of $\mathcal S$, whereas the right-hand side is a series in non-negative powers, therefore both are of degree zero, i.e.
\begin{equation*}
\widehat T=h_0 \in H\cap Sp_{2N}.
\end{equation*}
Since \(H\cap Sp_{2N}=\{I_{2N}\}\), we conclude that
$ 
\widehat T=T=I_{2N}.
$
\end{proof}

\begin{definition}
\label{def:YN}
Let $\Gamma_N$ be the matrix difference operator of Proposition~\ref{prop:GammaN}. We define the normalized dressing operators by
\begin{equation*}
Y_N:=\Gamma_NX_N,
\qquad
\widebar Y_N:=\Gamma_N\widebar X_N.
\label{normalizeddressing}
\end{equation*}
Then Proposition~\ref{prop:GammaN} directly implies that
\begin{equation}
Y_N\in G_-,
\qquad
Y_NJ_N\widebar Y_N^*=J_N.
\label{YNrelation}
\end{equation}
Equivalently,
\begin{equation*}
\widebar Y_N=-J_N(Y_N^*)^{-1}J_N.
\end{equation*}
Thus, after normalization, the two fermionic dressing operators are related
by the same Pfaff involution underlying the splitting of Section~\ref{sec:multicomponent-hierarchy}.
Moreover, $Y_N \in M_{2N}\bigl(V_g((\mathcal S))\bigr)$ has the canonical asymptotic form
\begin{equation}
Y_N=
\begin{pmatrix}
I_N+a_0+\mathcal O(\mathcal S)
&
b_{-1}\mathcal S^{-1}+\mathcal O(1)
\\[2mm]
c_1\mathcal S+\mathcal O(\mathcal S^2)
&
\Delta+\mathcal O(\mathcal S)
\end{pmatrix},
\label{shapeofYN}
\end{equation}
where $a_0,b_{-1},c_1\in\mathfrak n_-$, $\Delta\in B_-^\times$.
\end{definition}

\subsection{Auxiliary matrix difference operators}
We derive the auxiliary linear equations for the normalized dressing operators and prove that the auxiliary operators $B_x$ belong to the positive Lie subalgebra $\mathfrak g_+$, the key step for obtaining Sato evolutions.

 By Definition~\ref{def:YN}, the normalized dressing operators satisfy $Y_NJ_N\widebar Y_N^*=J_N$, with $Y_N\in G_-$.
We recall the involution $\theta$ of
Section~\ref{sec:multicomponent-hierarchy} which characterizes 
\begin{equation*}
\mathfrak g_+=
\left\{
M\in M_{2N}\bigl(\mathcal V_g[\mathcal S,\mathcal S^{-1}]\bigr)
\mid 
\theta(M)=M
\right\},
\qquad
\theta(M):=J_NM^*J_N .
\end{equation*} 
For $\alpha=1,\ldots,N$ and
$n\geq1$, the commuting bare directions are
\begin{equation*}
T_{\alpha,n}
=
\mathcal S^{-n}
\begin{pmatrix}
E_{\alpha\alpha}&0\\
0&0
\end{pmatrix},
\qquad
\widebar T_{\alpha,n}=-\mathcal S^{-n}
\begin{pmatrix}
0&0\\
0&E_{\alpha\alpha}
\end{pmatrix}.
\end{equation*}
\begin{lemma}
\label{lem:bareframes_evolutions}
For every $\alpha=1,\ldots,N$ and $n\geq1$, the bare wave frames \eqref{bareframes}
satisfy
\begin{equation*}
\partial_{t_{\alpha,n}}\Xi
=
T_{\alpha,n}\Xi,
\qquad
\partial_{t_{\alpha,n}}\Xi^\lor
=
\theta(T_{\alpha,n})\Xi^\lor,
\end{equation*}
\begin{equation*}
\partial_{\widebar t_{\alpha,n}}\Xi
=
\widebar T_{\alpha,n}\Xi,
\qquad
\partial_{\widebar t_{\alpha,n}}\Xi^\lor
=
\theta(\widebar T_{\alpha,n})\Xi^\lor.
\end{equation*}
\end{lemma}
\begin{proof}
This follows directly from the definitions of the bare wave functions \eqref{defchis}. For instance,
\begin{equation*}
\partial_{t_{\alpha,n}}\chi_\beta
=
\delta_{\alpha\beta}z^n\chi_\beta
=
\delta_{\alpha\beta}\mathcal S^{-n}\chi_\beta,
\end{equation*}
whereas
\begin{equation*}
\partial_{t_{\alpha,n}}\chi_\beta^\lor
=
-\delta_{\alpha\beta}z^n\chi_\beta^\lor
=
-\delta_{\alpha\beta}\mathcal S^n\chi_\beta^\lor.
\end{equation*}
The identities involving the barred times follow in the same way from the definitions of
$\widebar\chi_\beta$ and $\widebar\chi_\beta^\lor$.
\end{proof}
Motivated by Lemma~\ref{lem:bareframes_evolutions}, for
$x\in\{T_{\alpha,n},\widebar T_{\alpha,n}\}$ we shall use the convention
\begin{equation}\label{eq:dx_convention}
\partial_{T_{\alpha,n}}
:=\partial_{t_{\alpha,n}},
\qquad
\partial_{\widebar T_{\alpha,n}}
:=\partial_{\widebar t_{\alpha,n}}.
\end{equation}
We first work with the raw dressing operators $X_N,\widebar X_N$ of
Section~\ref{sec:SZ-dressing-operators}. By Lemma~\ref{lem:XNinvertibility}, these operators are invertible in their
respective one-sided completions.
For $x\in\{T_{\alpha,n},\widebar T_{\alpha,n}\}$, define the raw auxiliary operators by 
\begin{equation}\label{eq:raw_auxiliary}
B_x^{\rm raw}= (\partial_xX_N)X_N^{-1}+X_NxX_N^{-1}
, 
\qquad
\widebar B_x^{\rm raw}=
(\partial_x\widebar X_N)\widebar X_N^{-1}+ \widebar X_N\theta(x)\widebar X_N^{-1}.
\end{equation}
By Lemmas $A.3$ and $A.11$, we get the auxiliary linear equations
\begin{equation*}
\partial_x\Phi_{\rm raw}
=
 B_x^{\rm raw}\Phi_{\rm raw},
\qquad
\partial_x\widebar\Phi_{\rm raw}
=
\widebar{ B}_x^{\rm raw}\widebar\Phi_{\rm raw}.
\end{equation*}

\begin{lemma}
\label{lem:Braw_barBraw}
For every $x\in\{T_{\alpha,n},\widebar T_{\alpha,n}\}$, one has
\begin{equation*}
 B_x^{\rm raw}=\widebar{B}_x^{\rm raw}.
\end{equation*}
\end{lemma}

\begin{proof}
We differentiate the matrix residue identity~\eqref{eq:raw-matrix-residue} with
respect to the corresponding unprimed time variable, keeping
all primed variables fixed. Using the auxiliary
linear equations \eqref{eq:raw_auxiliary}, we obtain
\begin{equation*}
\operatorname{res}_z
 B_x^{\rm raw}\Phi_{\rm raw}(\boldsymbol s,\boldsymbol r,\boldsymbol t, \widebar{ \boldsymbol t};z)
J_0
\widebar\Phi_{\rm raw}(\boldsymbol s',\boldsymbol r',\boldsymbol t',\widebar{\boldsymbol t}';z)^\top \,dz
\end{equation*}
\begin{equation*}
=
\operatorname{res}_z
\widebar{B}_x^{\rm raw}\widebar\Phi_{\rm raw}(\boldsymbol s,\boldsymbol r,\boldsymbol t,\widebar{\boldsymbol t};z)
J_0
\Phi_{\rm raw}(\boldsymbol s',\boldsymbol r',\boldsymbol t',\widebar{\boldsymbol t}';z)^\top \,dz .
\end{equation*}
We then specialize the primed variables in the differentiated
identity to $(\boldsymbol  s',\boldsymbol  r', \boldsymbol  t',\widebar{\boldsymbol  t}')=(\boldsymbol  s-m\boldsymbol  1, \boldsymbol  r, \boldsymbol  t, \widebar{\boldsymbol  t})$, $m \in \mathbb Z$. Keeping the convention of Lemma~\ref{lem:matrix-residue}, we get
\begin{equation*}
\operatorname{res}_{z}
(B_x^{\rm raw}\Phi_{\mathrm{raw}})(\boldsymbol  s)\,
J_0\,
(\widebar\Phi_{\mathrm{raw}}(\boldsymbol  s-m\boldsymbol  1))^\top\,dz
=
\operatorname{res}_{z}
(\widebar B_x^{\rm raw}\widebar\Phi_{\mathrm{raw}})(\boldsymbol  s)\,
J_0\,
(\Phi_{\mathrm{raw}}(\boldsymbol  s-m\boldsymbol  1))^\top\,dz.
\label{eq:raw-matrix-residue-specialized-2}
\end{equation*}
Now we introduce
\begin{equation*}
C_x:=B_x^{\rm raw}X_N
=
\partial_xX_N+X_Nx,
\qquad
\widebar{C}_x:=\widebar{ B}_x^{\rm raw}\widebar X_N
=
\partial_x\widebar X_N+\widebar X_N\theta(x).
\end{equation*}
Since $X_N\in M_{2N}(\mathcal V_g[[\mathcal S]])$ and $x$ is a Laurent difference operator,
$C_x$ is bounded from below in the $\mathcal S$-degree. Similarly,
$\widebar{C}_x$ is bounded from above. Thus the four residue--coefficient identities of
Lemma~\ref{lem:matrix-residue} apply block-wise. Since $dz=z\,dz/z$, the same coefficient extraction as in the proof
of Proposition~\ref{prop:raw-pfaff-relation} yields, for every $m\in\mathbb Z$,
\begin{equation*}
\bigl( C_xJ_N\widebar X_N^*\bigr)_m
=
\bigl(\widebar{ C}_xJ_NX_N^*\bigr)_m.
\end{equation*}
As $m$ is arbitrary, $
 C_xJ_N\widebar X_N^*
=
\widebar{ C}_xJ_NX_N^*,
$
or equivalently,
\begin{equation*}
 B_x^{\rm raw}X_NJ_N\widebar X_N^*
=
\widebar{ B}_x^{\rm raw}\widebar X_NJ_NX_N^*.
\end{equation*}
By Proposition~\ref{prop:raw-pfaff-relation}, both products multiplying the auxiliary
operators are equal to $K_N$. Hence
\begin{equation*}
B_x^{\rm raw}K_N=\widebar B_x^{\rm raw}K_N.
\end{equation*}
By Lemma~\ref{lem:XNinvertibility}, $K_N$ is invertible, and therefore
$B_x^{\rm raw}=\widebar B_x^{\rm raw}$.
\end{proof}
Transporting this identity through the normalization of Proposition~\ref{prop:GammaN}, we define
\begin{equation}\label{eq:BbarB-def}
B_x:=(\partial_xY_N)Y_N^{-1}+Y_NxY_N^{-1},
\qquad
\widebar B_x:=(\partial_x\widebar Y_N)\widebar Y_N^{-1}+\widebar Y_N\theta(x)\widebar Y_N^{-1}.
\end{equation}
Using $Y_N=\Gamma_NX_N$ and $\widebar Y_N=\Gamma_N\widebar X_N$, we obtain
\begin{equation*}
B_x
=
(\partial_x\Gamma_N)\Gamma_N^{-1}
+
\Gamma_NB_x^{\mathrm{raw}}\Gamma_N^{-1},
\qquad
\widebar B_x
=
(\partial_x\Gamma_N)\Gamma_N^{-1}
+
\Gamma_N\widebar B_x^{\mathrm{raw}}\Gamma_N^{-1}.
\end{equation*}
Hence Lemma~\ref{lem:Braw_barBraw} gives $
B_x=\widebar B_x.
$

\begin{proposition}
\label{prop:Bx_evolution}
For every \(x\in\{T_{\alpha,n},\widebar T_{\alpha,n}\}\), the normalized auxiliary operators satisfy
\begin{equation*}
B_x=\widebar B_x=\theta(B_x) \in \mathfrak{g}_+ \subset M_{2N}\bigl(\mathcal V_g[\mathcal S,\mathcal S^{-1}]\bigr).
\end{equation*}
\end{proposition}

\begin{proof}
We have already shown that \(B_x=\widebar B_x\). By~\eqref{YNrelation},
\begin{equation}
\widebar Y_N=-J_N(Y_N^*)^{-1}J_N,
\label{prop5101}
\end{equation}
which differentiated by $\partial_x$ gives
\begin{equation}
\partial_x\widebar Y_N
=
J_N(Y_N^*)^{-1}(\partial_xY_N)^*(Y_N^*)^{-1}J_N.
\label{prop5102}
\end{equation}
Combining~\eqref{prop5101} and~\eqref{prop5102}, we obtain
\begin{equation*}
(\partial_x\widebar Y_N)\widebar Y_N^{-1}
=
J_N((\partial_x Y_N) Y_N^{-1})^*J_N
=
\theta((\partial_x Y_N)Y_N^{-1}).
\end{equation*}
Similarly,
\begin{equation*}
\widebar Y_N\theta(x)\widebar Y_N^{-1}
=
J_N\bigl(Y_NxY_N^{-1}\bigr)^*J_N
=
\theta(Y_NxY_N^{-1}).
\end{equation*}
It follows that
\begin{equation*}
B_x=\widebar B_x
=
(\partial_x\widebar Y_N)\widebar Y_N^{-1}+\widebar Y_N\theta(x)\widebar Y_N^{-1}
=
\theta((\partial_x Y_N)Y_N^{-1})+\theta(Y_NxY_N^{-1})
=
\theta(B_x).
\end{equation*}
It remains to show that $B_x$ has finite Laurent support. Since $
Y_N\in M_{2N}\bigl(\mathcal V_g((\mathcal S))\bigr),
$
we have
\begin{equation*}
B_x
=
(\partial_xY_N)Y_N^{-1}
+
Y_NxY_N^{-1}
\in
M_{2N}\bigl(\mathcal V_g((\mathcal S))\bigr).
\end{equation*}
On the other hand, $\widebar Y_N\in M_{2N}\bigl(\mathcal V_g((\mathcal S^{-1}))\bigr),
$
and hence
\begin{equation*}
\widebar B_x
\in
M_{2N}\bigl(\mathcal V_g((\mathcal S^{-1}))\bigr),
\end{equation*}
so $\widebar B_x$ is bounded from above. Since $B_x=\widebar B_x$,
\begin{equation*}
B_x
\in
M_{2N}\bigl(\mathcal V_g((\mathcal S))\bigr)
\cap
M_{2N}\bigl(\mathcal V_g((\mathcal S^{-1}))\bigr)
=
M_{2N}\bigl(\mathcal V_g[\mathcal S,\mathcal S^{-1}]\bigr).
\end{equation*}
Hence, $B_x \in \mathfrak{g}_+$.
\end{proof}

\subsection{Proof of the realization theorem}

We complete the realization theorem by identifying the continuous evolutions of $Y_N$ with the commuting diagonal sector of the Matrix Pfaff--Toda hierarchy.

Recall that $\mathcal V_N$ is the localization of the difference
algebra generated by the coefficients of the universal dressing
operator $U_N$ at $\det\Delta$ and all of its $\mathcal S$-shifts.
Since $Y_N$ has the canonical form~\eqref{shapeofYN}, its lower-right
leading coefficient satisfies $\Delta\in B_-^\times$. Hence
\begin{equation*}
\det\Delta\in\mathcal V_g^\times,
\qquad
\mathcal S^k(\det\Delta)\in\mathcal V_g^\times,
\qquad k\in\mathbb Z.
\end{equation*}
Therefore, the evaluation of the universal dressing coefficients at those
of $Y_N$ extends uniquely to an $\mathcal S$-difference algebra
homomorphism
\begin{equation*}
\operatorname{ev}_g:\mathcal V_N\longrightarrow\mathcal V_g,
\qquad
U_N\longmapsto Y_N.
\end{equation*}

\begin{theorem}
\label{thm:fermionic_realisation_PfaffToda}
Let $g$ be a Clifford-group element satisfying the
Savchenko--Zabrodin fermionic bilinear identity~\eqref{fermionic-bilinear} and
belonging to the regular factorization locus of Definition~\ref{def:regular-factorization-locus}.
Then the normalized dressing operator $Y_N$ defines a realization
of the commuting diagonal sector of the Matrix Pfaff--Toda hierarchy. More precisely,
for every $\alpha=1,\ldots,N$ and $n\geq1$,
\begin{equation*}
\operatorname{ev}_g\circ D_{T_{\alpha,n}}
=
\partial_{t_{\alpha,n}}\circ\operatorname{ev}_g,
\qquad
\operatorname{ev}_g\circ D_{\widebar T_{\alpha,n}}
=
\partial_{\widebar t_{\alpha,n}}\circ\operatorname{ev}_g.
\end{equation*}
\end{theorem}

\begin{proof}
Let
$
x\in\{T_{\alpha,n},\widebar T_{\alpha,n}\}.
$
By the convention introduced in~\eqref{eq:dx_convention}, $\partial_x$ denotes
the corresponding continuous derivative.
Since $Y_N\in G_-$, its logarithmic derivative belongs to
$\mathfrak g_-$, $(\partial_xY_N)Y_N^{-1}\in\mathfrak g_-$.
On the other hand, by definition of $B_x$ in~\eqref{eq:BbarB-def},
\begin{equation*}
Y_NxY_N^{-1}
=
B_x-(\partial_xY_N)Y_N^{-1}.
\end{equation*}
By Proposition~\ref{prop:Bx_evolution}, $B_x\in\mathfrak g_+$. Hence, by uniqueness of
the Pfaff--Toda splitting,
\begin{equation*}
B_x
=
\Pi_+\!\left(Y_NxY_N^{-1}\right),
\qquad
-(\partial_xY_N)Y_N^{-1}
=
\Pi_-\!\left(Y_NxY_N^{-1}\right).
\end{equation*}
Therefore
\begin{equation*}
\partial_xY_N
=
-\Pi_-\!\left(Y_NxY_N^{-1}\right)Y_N.
\end{equation*}
On the other hand, the universal dressing derivation of Section~\ref{sec:dressing_zero_curvature}
satisfies
\begin{equation*}
D_x(U_N)
=
-\Pi_-\!\left(U_NxU_N^{-1}\right)U_N.
\end{equation*}
The homomorphism $\operatorname{ev}_g$ commutes with the shift,
formal adjoint, Laurent projections and triangular projections, and
hence with $\Pi_\pm$. Consequently,
\begin{equation*}
\operatorname{ev}_g\!\left(D_x(U_N)\right)
=
-\Pi_-\!\left(Y_NxY_N^{-1}\right)Y_N
=
\partial_xY_N.
\end{equation*}
Since the coefficients of $U_N$ generate $\mathcal V_N$ as a
localized difference algebra, it follows that
$\operatorname{ev}_g\circ D_x=\partial_x\circ\operatorname{ev}_g$.
Taking respectively $x=T_{\alpha,n}$ and
$x=\widebar T_{\alpha,n}$ proves the two stated identities.
\end{proof}

\bibliographystyle{bibstyle}
\bibliography{biblio}

\end{document}